\documentclass[11pt]{amsart}
\usepackage{graphicx}
\usepackage{amsmath}
\usepackage{amsfonts}
\usepackage{amssymb}
\usepackage{amsthm}
\usepackage{comment}
\usepackage{listings}
\usepackage{xcolor}
\definecolor{codegreen}{rgb}{0,0.6,0}
\definecolor{codegray}{rgb}{0.5,0.5,0.5}
\definecolor{codepurple}{rgb}{0.58,0,0.82}
\definecolor{backcolour}{rgb}{0.95,0.95,0.92}

\lstdefinestyle{mystyle}{
    backgroundcolor=\color{backcolour},   
    commentstyle=\color{codegreen},
    keywordstyle=\color{magenta},
    numberstyle=\tiny\color{codegray},
    stringstyle=\color{codepurple},
    basicstyle=\ttfamily\footnotesize,
    breakatwhitespace=false,         
    breaklines=true,                 
    captionpos=b,                    
    keepspaces=true,                 
    numbers=left,                    
    numbersep=5pt,                  
    showspaces=false,                
    showstringspaces=false,
    showtabs=false,                  
    tabsize=2
}
\usepackage{hyperref}
\hypersetup{
	colorlinks=true,
	linkcolor=blue,
	filecolor=blue,
	citecolor=blue,
    urlcolor=blue,
}
\usepackage{doi}
\usepackage{graphicx}
\usepackage{tikz}
\usetikzlibrary{automata,cd,arrows,calc,fit,shapes,angles,positioning,intersections,quotes,decorations.markings,decorations.pathreplacing}
\tikzset{mynode/.style={ }}
\tikzset{
	modal/.style={>=stealth’,shorten >=1pt,shorten <=1pt,auto,node distance=2.5cm,
		semithick},
	world/.style={circle,draw,minimum size=0.5cm,fill=gray!15},
	point/.style={circle,draw,inner sep=0.5mm,fill=black},
	reflexive above/.style={->,loop,looseness=7,in=120,out=60},
	reflexive below/.style={->,loop,looseness=7,in=240,out=300},
	reflexive left/.style={->,loop,looseness=7,in=150,out=210},
	reflexive right/.style={->,loop,looseness=7,in=30,out=330}
}
\tikzset{
	->,
	>=stealth,
	node distance=20mm,
	every state/.style={thick, fill=gray!5},
	initial text=$ $,
}
\tikzset{
modal/.style={
>=stealth',shorten >=1pt,shorten <=1pt,auto,
node distance=2.5cm,semithick},
point/.style={circle,draw,fill=black,inner sep=0.70mm}}
\usepackage{tkz-euclide}
\usepackage{modalops}
\usepackage{multicol}
\usepackage{bussproofs}
\usepackage{framed}
\usepackage{stmaryrd}
\usepackage{float}
\usepackage{mathrsfs}
\usepackage{csquotes}
\usepackage{adjustbox}
\usepackage{cleveref}
\usepackage{color}

\newenvironment{bprooftree}
  {\leavevmode\hbox\bgroup}
  {\DisplayProof\egroup}
\usepackage{empheq}
\usepackage{xcolor}
\definecolor{lightgreen}{HTML}{90EE90}

\makeatletter
\@ifpackageloaded{stix}{%
}{%
  \DeclareFontEncoding{LS2}{}{\noaccents@}
  \DeclareFontSubstitution{LS2}{stix}{m}{n}
  \DeclareSymbolFont{stix@largesymbols}{LS2}{stixex}{m}{n}
  \SetSymbolFont{stix@largesymbols}{bold}{LS2}{stixex}{b}{n}
  \DeclareMathDelimiter{\lBrace}{\mathopen} {stix@largesymbols}{"E8}%
                                            {stix@largesymbols}{"0E}
  \DeclareMathDelimiter{\rBrace}{\mathclose}{stix@largesymbols}{"E9}%
                                            {stix@largesymbols}{"0F}
}
\makeatother
\usepackage{stackengine}

\usepackage{geometry}
\usepackage{makecell}
\usepackage{multirow}
\author[L. Fusco]{Ludovico Fusco}\address{L. Fusco: University of Urbino Carlo Bo, Urbino, Italy}\email{l.fusco2@campus.uniurb.it}
\author[A. Aldini]{Alessandro Aldini}\address{A. Aldini: University of Urbino Carlo Bo, Urbino, Italy}\email{alessandro.aldini@uniurb.it}


\newtheorem{proposition}{Proposition}[section]
\newtheorem{lemma}[proposition]{Lemma}
\newtheorem{theorem}[proposition]{Theorem}
\newtheorem{corollary}[proposition]{Corollary}

\theoremstyle{definition}
\newtheorem{definition}[proposition]{Definition}
\newtheorem{example}[proposition]{Example}
\newtheorem{remark}[proposition]{Remark}

\newtheorem*{rexample}{Running example}
\date{December 5, 2025}
\begin{document}
\title[A Cut-Free Sequent Calculus for the Analysis of Finite-Trace Properties \dots]{A Cut-Free Sequent Calculus for the Analysis of Finite-Trace Properties in Concurrent Systems}
\begin{abstract}
We address the problem of identifying a proof-theoretic framework that enables a compositional analysis of finite-trace properties in concurrent systems, with a particular focus on those specified via prefix-closure. To this end, we investigate the interaction of a prefix-closure operator and its residual (with respect to set-theoretic inclusion) with language intersection, union, and concatenation, and introduce the variety of closure $\ell$-monoids as a minimal algebraic abstraction of finite-trace properties to be conveniently described within an analytic proof system. Closure $\ell$-monoids are division-free reducts of distributive residuated lattices equipped with a forward diamond/backward box residuated pair of unary modal operators, where the diamond is a topological closure operator satisfying $\pos(x \cdot y) \leq \pos\! x \cdot \pos\! y$. As a logical counterpart to these structures, we present $\mathsf{LMC}$, a Gentzen-style system based on the division-free fragment of the Distributive Full Lambek Calculus. In $\mathsf{LMC}$, structural terms are built from formulas using Belnap-style structural operators for monoid multiplication, meet, and diamond. The rules for the modalities and the structural diamond are taken from Moortgat’s system $\mathsf{NL}(\pos)$. We show that the calculus is sound and complete with respect to the variety of closure $\ell$-monoids and that it admits cut elimination.
\end{abstract}
\maketitle
\section{Introduction}
\label{sec:introduction}
Designing computational architectures that exhibit concurrency—such as distributed or reactive systems—requires the constant development of formal techniques to ensure that system executions align with the intended implementations~\cite{Schneider2004}. 

In many algorithmic verification frameworks, notably model checking \cite{BaierKatoen08}, computations are modelled as \emph{execution traces} derived from high-level specifications such as automata or labelled transition systems (LTSs). This approach dates back to the very origins of concurrency theory as a distinct branch of computer science~\cite{Lamport2015}, as it is implicitly adopted in Dijkstra’s solution to the mutual exclusion problem \cite{Dijkstra1965}. Models of computation that admit a trace-based specification may differ considerably from one another\footnote{For an overview of examples, see~\cite[§7.2]{clsc}.}. In each case, a trace is a finite or infinite sequence of objects defined in terms of the underlying model's states and/or actions, possibly equipped with additional data allowing to properly encode system executions. This naturally supports the representation of key system policies in terms of \emph{trace properties} (i.e., sets of execution traces) and provides a foundation for their classification according to the well-established \emph{safety-liveness dichotomy} \cite{AlpSch85,lam}. The strength of this methodology lies in the fact that traces, being essentially strings over alphabets, can be conveniently treated in language-theoretic terms. Modulo a labelling of states with propositions, trace properties thus become simply string properties expressing \emph{patterns} of system behaviours.

In the formal methods literature, it is common to find system executions represented as infinite traces, since programs may not terminate. It follows that most classical definitions and results concerning behavioural properties and their verification are formulated with respect to infinite traces—a choice reflected in the semantics of most temporal logics for computer science \cite{DGL2016}. However, applied research has shown compelling reasons to adopt \emph{finite traces} in many different contexts, and a variety of formalisms has been introduced for specifying and verifying properties of finite computations (see, e.g., \cite{BeLoMuRu18,DeGVardi13,OuWo2007,Rosu2018}). In this setting, a finite-trace property (henceforth, \emph{f-property}) is a formal language in the usual sense, i.e., a subset of (the universe of) a free monoid~$\mathbf{F}_{\Sigma}=\langle\Sigma^*, \cdot, \varepsilon\rangle$. Most interestingly, the fundamental classes of safety and liveness f-properties admit an elegant characterisation in terms of prefix-closure \cite{MaPa17,Rosu2012}. Indeed, $P \subseteq \Sigma^*$ models a safety f-property if it is \emph{prefix-closed}—that is, if~$P$ consists of all strings that are prefixes of some string in~$P$ itself (equivalently, if~$P$ contains every prefix of each of its strings). By contrast,~$P$ expresses a liveness f-property when every string in~$\Sigma^*$ admits an extension belonging to~$P$—that is, if the prefix-closure of~$P$ is~$\Sigma^*$. It is well known that prefix-closure acts as a topological closure operator on~$\wp(\Sigma^*)$, yielding a topology on~$\Sigma^*$ whose closed and dense sets model, respectively, safety and liveness f-properties.
\subsection{Our contribution}
Depending on the operations defined on it, a set of the form $\wp(\Sigma^*)$ can serve as the universe of algebras belonging to several well-known classes, including the varieties of residuated lattices~\cite{yellowbook}, residuated Boolean algebras~\cite{JoTsi93}, Kleene algebras~\cite{Kozen1994}, and action algebras~\cite{Pratt91}. These varieties provide the algebraic semantics for major \emph{substructural logics}~\cite{yellowbook,MPT,primer}, which thus offer natural tools for reasoning about f-properties across different levels of granularity. In this paper, we are interested in identifying a suitable framework within \emph{structural proof theory} for addressing questions of the following kind:
\begin{quote}
Given $P, Q_1, \dots, Q_n \subseteq \Sigma^*$, if $P$ is constructed from $Q_1, \dots, Q_n$ by applying certain operations (e.g. $\cap$, $\cup$, or language concatenation), what can be \emph{proved} about the structure of the prefix-closure of $P$?
\end{quote}
Our purpose is therefore to provide an analytic calculus enabling the structural analysis of many property patterns central to formal verification, while also affording a \emph{compositional account} of safety, liveness, and related notions. 

The design of such a proof system crucially depends on the choice of an underlying \emph{algebraic model} for f-properties. To begin with, for a free monoid~$\mathbf{F}_{\Sigma}=\langle\Sigma^*, \cdot, \varepsilon\rangle$, we endow $\wp(\Sigma^*)$ with the structure of a bounded $\ell$-monoid, taking $\cap$, $\cup$, and language concatenation as fundamental operations. We then expand this algebra, denoted $\wp(\mathbf{F}_{\Sigma})$, by two additional operations:
\begin{itemize}
\item A (forward) diamond $\pos$ corresponding to prefix-closure;
\item A (backward) box $\necv$ mapping a property~$P \subseteq \Sigma^*$ to the property~$\necv\!P$ such that $w \in \necv\!P$ iff every prefix of~$w$ belongs to~$P$.
\end{itemize}
Our box operation is the \emph{residual} \cite{BlythJan72, MPT} of prefix-closure with respect to set-theoretic inclusion: that is, for $P, Q \subseteq \Sigma^*$, $\pos\!P \subseteq Q$ iff $P \subseteq \necv\!Q$. Moreover, as we shall prove, $\necv$ plays a key role in the mathematical foundation for the finite-trace counterpart of a well-established technique for specifying safety properties through the combination of future and past modalities in temporal logic \cite{MannaPnueli90}.

We study the equational definability of this expanded structure, denoted $\wp(\mathbf{F}_{\Sigma})_+$, starting from the observation that, in a free monoid, the prefix order coincides with the \emph{left divisibility} relation. We first show that divisibility preorders on monoids always satisfy the \emph{Riesz Decomposition Property} (\emph{RDP}) \cite{anfe88, Riesz29}. This result allows us to characterise the general behaviour of inverse image operators (like our $\pos$) associated with such preorders. We then proceed with a more fine-grained examination of the algebraic properties of prefix-closure and provide a number of equations describing the interaction of $\pos$ and $\necv$ with the $\ell$-monoid structure. 

Next, we introduce the variety $\mathfrak{LMC}$ of \emph{closure $\ell$-monoids} as a minimal algebraic model for f-properties to be captured within an analytic proof system. Concretely, we work with division-free reducts of bounded distributive residuated lattices equipped with residuated pairs of unary modalities of the form forward diamond/backward box, where the diamond is a topological closure operator satisfying the inequation $\pos(x\cdot y)\leq \pos\!x\cdot \pos\!y$. The algebra $\wp(\mathbf{F}_{\Sigma})_+$ turns out to be a closure $\ell$-monoid, while it does not generate $\mathfrak{LMC}$.

Finally, we introduce $\mathsf{LMC}$, a sound and complete Gentzen-style calculus for closure $\ell$-monoids. In designing it, given the expressiveness of our algebraic framework, we decided to follow the approach established by Belnap’s \emph{Display Logic}~\cite{Belnap82}, where the Gentzen terms occurring in derivations are constructed by combining formulas through \emph{structural operators} mirroring the behaviour of the logical connectives. Unary residuation and, more generally, relationships between connectives arising from the category-theoretic notion of \emph{adjunction}~\cite{kan} can be conveniently treated in display systems \cite{CGPT2022,CRW2014,Gore1998,GJLPT2024,Wansing1998}. However, with a view to future work on the effective implementability of our research, we do not resort to the full framework of Display Logic and introduce only the minimal amount of structural machinery required to establish completeness. As a result, $\mathsf{LMC}$ is a single-conclusion system with structural operators appearing only in the left sides of~sequents.

We design $\mathsf{LMC}$ building on the division-free fragment of the Distributive Full Lambek Calculus \cite{GalJip2017,Kozak09,Restall94}, which includes structural operators for monoid multiplication and meet. In particular, the structural meet enables the derivation of the distributive laws for the lattice operations—an idea independently pioneered by Dunn \cite{Dunn1973} and Minc~\cite{Minc1976} in their work on the positive relevant logic $\mathsf{R}^+$. The syntax of structural terms also includes a structural diamond, whose rules—together with those for the modal operators—are taken from Moortgat’s system $\mathsf{NL}(\pos)$~\cite{Moortgat96}. We show that $\mathsf{LMC}$ enjoys cut elimination. This result is obtained as a corollary of a mix elimination theorem, established using a technique originally devised in the context of hypersequent calculi for fuzzy logics \cite{MOG2009} and later employed in the cut elimination of the hypersequent calculus $\mathsf{CSemFL}$ \cite{MPT} for commutative, semilinear pointed residuated lattices.
\subsection{Outline of the paper}
In Section \ref{sec:background}, we provide the necessary background for our algebraic analysis of f-properties. In Section~\ref{sect: thcomp}, we introduce our reference modelling framework for concurrent systems. In particular, we choose to work within the \emph{interleaving paradigm}, starting from a notion of trace derived from the computational model of labelled transition systems (LTSs). Section~\ref{sect: algsect} is deals with the algebraic theory of f-properties: we begin with our concrete structure $\wp(\mathbf{F}_{\Sigma})_+$ and proceed to the variety of closure $\ell$-monoids. Section~\ref{sect: LMC} presents the Gentzen-style system~$\mathsf{LMC}$ and establishes its completeness, while Section~\ref{sect: cutelim} is entirely devoted to cut elimination. Finally, in Section~\ref{sect: conc}, we discuss related work and outline directions for future research.
\section{Preliminaries}
\label{sec:background}
We begin by recalling some basic algebraic concepts and results that will be used in the subsequent sections. For background and terminology not covered here, we refer the reader to standard sources in universal algebra~\cite{berg,busa,ALV1}, lattice theory~\cite{dp,gra}, residuation theory~\cite{BlythJan72}, and the general theory of ordered algebraic structures~\cite{Fuchs63}. Throughout the paper, $\curlywedge$, $\curlyvee$, $\Rightarrow$, and $\Leftrightarrow$ denote metatheoretical conjunction, disjunction, material implication, and material equivalence, respectively.
\subsection{Universal algebra}
\label{subsec: algebra}
We follow the customary convention of identifying algebraic similarity types~${\nu\colon\mathscr{F}\rightarrow \mathbb{N}}$ with sequences~$\langle \nu(f) \rangle_{f \in \mathscr{F}}$, where the arities of operation symbols are listed in non-increasing order. As usual, algebras and their universes are denoted using boldface and italics, respectively. The symbol~$X$ will always denote a countable set of variables (ranged over by~$x, y, z, \dots$). 

Fix a similarity type~${\nu \colon\mathscr{F}\rightarrow \mathbb{N}}$. We write~$\mathbf{T}_{\nu}(X)$ for the term algebra of type~$\nu$ over~$X$. Let $\mathbf{A}$ be a $\nu$-algebra. $\mathrm{Con}(\mathbf{A})$ is the class of congruences on $\mathbf{A}$. For $t \in T_{\nu}(X)$, we denote by $t^{\mathbf{A}}$ the term operation of $\mathbf{A}$ associated with $t$. If $h \colon \mathbf{A} \to \mathbf{B}$ is a homomorphism of~$\mathbf{A}$ into a $\nu$-algebra~$\mathbf{B}$, then we denote by $\ker h$ its kernel. The Homomorphism Theorem (see, e.g., \cite[Th. 6.12]{busa}) states that $h \colon \mathbf{A} \to \mathbf{B}$ is a surjective homomorphism, then $\mathbf{A}/\ker h$ is isomorphic to $\mathbf{B}$. 

A \emph{variety} is a class of $\nu$-algebras closed under taking (isomorphic copies of) homomorphic images, subalgebras, and direct products of its members.
If~$\mathfrak{V}$ is the least variety extending a given class~$\mathfrak{X}$ of $\nu$-algebras, we say that~$\mathfrak{V}$ is \emph{generated} by~$\mathfrak{X}$ and write~$\mathfrak{V} = \mathbb{V}(\mathfrak{X})$.
When $\mathfrak{X} = \{ \mathbf{A} \}$, we write~$\mathbb{V}(\mathbf{A})$.
As showed by Tarski~\cite{tar46}, $\mathbf{A} \in \mathbb{V}(\mathfrak{X})$ iff~$\mathbf{A}$ is a homomorphic image of a subalgebra of a direct product of members of~$\mathfrak{X}$.

An \emph{equation} (of type $\nu$) over $X$ is a pair $\langle s, t \rangle\in T_{\nu}(X)^2$, denoted $s \approx t$. If $\nu$ is a type for lattice-based structures, equations of the form $s \wedge t \approx s$ and $s \vee t \approx t$ are rewritten as $s \leq t$ and referred to as \emph{inequations}. In the same context, each equation $s \approx t$ is split into inequations $s\leq t$ and $t \leq s$. A $\nu$-algebra $\mathbf{A}$ satisfies an equation $s \approx t$ (written $\mathbf{A}\models s \approx t$) iff $h(s)=h(t)$ for all $h\colon\mathbf{T}_{\nu}(X)\rightarrow \mathbf{A}$ (hence if $s^{\mathbf{A}}=t^{\mathbf{A}}$). Satisfaction of sets of equations by algebras, and of equations or sets of equations by classes of algebras, is defined in the usual way. Let $\mathfrak{K}$ be a class of $\nu$-algebras. The \textit{equational consequence relation} $\vDash_{\mathfrak{K}}$ associated with $\mathfrak{K}$ is defined as follows, for any set $\Phi\cup\{ s\approx t\}$ of equations: $\Phi\vDash_{\mathfrak{K}}s \approx t$ iff $\mathfrak{K}\models \Phi$ implies $\mathfrak{K}\models s \approx t$. An \textit{equational class} of $\nu$-algebras is the class of all models of some set $\Phi$ of equations. By Birkhoff's Theorem~\cite{bi35}, equational classes and varieties coincide. 

A \emph{quasi-equation} (of type $\nu$) over $X$ is a definite Horn clause $\bigcurlywedge \!\Phi \Rightarrow s\approx t$, with $\Phi\cup\{ s\approx t\}$ a \textit{finite} set of equations. A $\nu$-algebra $\mathbf{A}$ satisfies a quasi-equation $\bigcurlywedge \!\Phi \Rightarrow s\approx t$ (written $\mathbf{A}\models \bigcurlywedge \!\Phi \Rightarrow s\approx t$) iff $\Phi\vDash_{\{\mathbf{A}\}}s \approx t$. This lifts naturally to classes of $\nu$-algebras, so we have $\mathfrak{K}\models {\bigcurlywedge \!\Phi \Rightarrow s\approx t}$ iff $\Phi\vDash_{\mathfrak{K}}s \approx t$. 

Now let~$\Theta(\mathfrak{K},X)$ be the smallest congruence on~$\mathbf{T}_{\nu}(X)$ such that the corresponding quotient algebra embeds in some member of~$\mathfrak{K}$. Define $\overline{X} := \{ [x]_{\Theta(\mathfrak{K},X)} \mid x \in X \}$. The quotient algebra $\mathbf{F}_{\mathfrak{K}}(\overline{X})=\mathbf{T}_{\nu}(X)/\Theta(\mathfrak{K},X)$ is \emph{freely generated} by~$\overline{X}$\footnote{That is,~$\mathbf{T}_{\nu}(X)/\Theta(\mathfrak{K})$ is generated by~$\overline{X}$ and has the \emph{universal mapping property} for~$\mathfrak{K}$ over~$\overline{X}$.} and is referred to as the (\emph{$\mathfrak{K}$-})\emph{free algebra} over~$\overline{X}$. Since the construction of~$\mathbf{F}_{\mathfrak{K}}(\overline{X})$ depends on (the cardinality of)~$X$, when no ambiguity arises we denote this structure by~$\mathbf{F}_X$. Recall that a nontrivial variety $\mathfrak{V}$ contains free algebras $\mathbf{F}_X$ for any nonempty $X$ \cite[Cor. 4.119]{ALV1} and, moreover, if~$X$ has denumerably many elements, then $\mathfrak{V} = \mathbb{V}(\mathbf{F}_X)$ \cite[Cor. 4.132]{ALV1}.
\begin{example}\cite[pp. 239--240]{ALV1}
Let~$\Sigma$ be a countable alphabet. The set~$\Sigma^* := \bigcup_{n \in \mathbb{N}} \Sigma^n$ of all finite sequences (strings) over~$\Sigma$ forms a monoid under string concatenation, with the empty string~$\varepsilon$ as the identity element. It is straightforward to verify that this is the free $\mathfrak{M}$-algebra over~$\overline{\Sigma}$, where~$\mathfrak{M}$ is the variety of monoids. Take $\Sigma$ as a set of variables, and let $h$ be the homomorphism of $\mathbf{T}_{\langle 2,0\rangle}(\Sigma)$ onto $ \langle \Sigma^*, \cdot, \varepsilon \rangle$ sending each~$t \in T_{\langle 2,0\rangle}(\Sigma)$ to the string over~$\Sigma$ obtained by deleting all operation symbols in $t$, while preserving order and multiplicity of variable occurrences (e.g., $h((x \cdot (y \cdot 1))\cdot x) = xyx$; clearly {$h(1)=\varepsilon$}). Since $h$ is surjective, by the Homomorphism Theorem, $\mathbf{T}_{\langle 2,0\rangle}(\Sigma)/\ker h \cong \langle \Sigma^*, \cdot, \varepsilon \rangle $. Finally, observe that if two terms~$s$ and~$t$ differ by even a single variable, then they are not equivalent modulo~$\ker h$; hence~$\ker h \subseteq \Theta(\mathfrak{M},\Sigma)$, and consequently~$\ker h = \Theta(\mathfrak{M},\Sigma)$.
\end{example}
It is well known that, for any variety~$\mathfrak{V}$, free algebras~$\mathbf{F}_X \in \mathfrak{V}$ have exactly the same equational theory as~$\mathfrak{V}$ \cite[Thm. 4.127]{ALV1}. That is, for all~$s, t \in T_{\nu}(X)$,
\[
\mathfrak{V} \models s \approx t 
\;\; \Leftrightarrow \;\; 
\langle s, t \rangle \in \Theta(\mathfrak{V},X) 
\;\; \Leftrightarrow \;\; 
s^{\mathbf{F}_X} = t^{\mathbf{F}_X} 
\;\; \Leftrightarrow \;\; 
\mathbf{F}_X \models s \approx t.
\]
Moreover, the equational theories of varieties coincide with those of their generators: indeed, if~$\mathfrak{X} \models s \approx t$, then~$s \approx t$ is satisfied by every homomorphic image of a subalgebra of a direct product of algebras in~$\mathfrak{X}$, and hence $\mathbb{V}(\mathfrak{X}) \models s \approx t$ \cite[Lem. 4.128]{ALV1}.

\subsection{Ordered monoids}\label{sect: monoids}
In what follows, for $ \mathfrak{d} \in \{l, r\}$, we write  $x \cdot_{\mathfrak{d}} y$ to denote monoid multiplication with argument positions determined by $\mathfrak{d} $. Specifically:
\[ x \cdot_{\mathfrak{d}} y=
\begin{cases}
x \cdot y \hfill & \text{if $\mathfrak{d}=l$,}\\
y \cdot x \hfill &\text{if $\mathfrak{d}=r$}
\end{cases}
\]
Consider a pair~$\langle \mathbf{M}, \preceq\rangle$, where~$\mathbf{M} = \langle M, \cdot, 1 \rangle$ is a monoid and~$\preceq$ is a preorder on~$M$. We say that $\preceq$ is \textit{compatible} with multiplication if, for all $a,b,c \in M$ and $\mathfrak{d} \in \lbrace l,r \rbrace$:
\begin{equation}\label{eq: COMP}
\tag{COMP}a\preceq b\; \Rightarrow \; c\cdot_{\mathfrak{d}} a\preceq c \cdot_{\mathfrak{d}} b
\end{equation}
If (\ref{eq: COMP}) holds just for $\mathfrak{d}=l$ (resp. $\mathfrak{d}=r$), then we say that $\preceq$ is \textit{$l$-compatible} (resp. \textit{$r$-compatible}) with multiplication. We call~$\langle \mathbf{M}, \preceq\rangle$ a \emph{preordered} (resp. \textit{partially ordered}) \emph{monoid} if $\preceq$ is a preorder (resp. partial order) satisfying (\ref{eq: COMP}). Analogously, for a fixed $\mathfrak{d}\in \{l,r\}$, we call~$\langle \mathbf{M}, \preceq\rangle$ a \emph{$\mathfrak{d}$-preordered} (resp. \emph{partially $\mathfrak{d}$-ordered}) \emph{monoid} if~$\preceq$ is a preorder (resp. a partial order) that is $\mathfrak{d}$-compatible with multiplication. 

An \textit{$\ell$-monoid} (short for \textit{lattice-ordered monoid}) is a partially ordered monoid~$\langle \mathbf{M}, \preceq\rangle$ where $\preceq$ is a lattice order.\footnote{These structures were introduced by Birkhoff \cite[§XIII]{Bi48} under the name of \textit{lattice-ordered semigroups.}} We denote by $\mathfrak{LM}$ the class of all $\ell$-monoids. This is in fact a variety, as an $\ell$-monoid may be defined as an algebra $\mathbf{M}=\langle M, \cdot, \wedge, \vee, 1\rangle$ of type $\langle 2,2,2,0\rangle$ such that $\langle M, \cdot,1\rangle$ is a monoid, $\langle M, \wedge, \vee \rangle$ is a lattice, and multiplication distributes over join on both sides, i.e., for~$\mathfrak{d} \in \{l, r\}$, the following equation holds:
\begin{equation}\label{eq: DMJ}
\tag{DMJ} x \cdot_{\mathfrak{d}} (y \vee z) \approx (x \cdot_{\mathfrak{d}}  y) \vee (x \cdot_{\mathfrak{d}}  z) 
\end{equation}
One readily verifies that (\ref{eq: DMJ}) implies (\ref{eq: COMP}). Recall that the~$\{\wedge\}$-free reducts of $\ell$-monoids form the variety $\mathfrak{ISR}$ of \textit{idempotent semirings}, which is axiomatised by equations for monoids, join-semilattices, and the distributivity equation (\ref{eq: DMJ}). An $\ell$-monoid is \textit{bounded} if so is its lattice reduct. In this case, we consider a type expansion with constants $\bot$ and $\top$ for the bounds, and assume axioms~$\bot \leq x$~and~$x \leq \top$.

In this paper, we deal with $\ell$-monoids whose underlying lattices are distributive. However, we avoid the term “distributive $\ell$-monoid”, which is typically reserved for inverse-free reducts of $\ell$-groups \cite{CGMS2022}.
\subsection{Divisibility in monoids}\label{sect: div}
In monoids and, more generally, in algebras with a monoid reduct, divisibility preorders provide a natural generalisation of the classical divisibility relation on integers. Let $\mathbf{M}\in\mathfrak{M}$, and let $a,b \in M$. For $\mathfrak{d}\in \{l,r\}$, define:
\[
 a \mid_{\mathfrak{d}} b \Leftrightarrow \exists c\in M\; (b=a \cdot_{\mathfrak{d}} c)
\]
Whenever $a \mid_l b$ (resp. $a \mid_r b$), we say that $b$ is \textit{left-divisible} (resp. \textit{right-divisible}) by $a$ and call $a$ a \textit{left divisor} (resp. \textit{right divisor}) of $b$. We call $\mid_l$ and $\mid_r$ are the \emph{left} and \emph{right divisibility} relations on $\mathbf{M}$, respectively\footnote{Clearly, if $\mathbf{M}$ is commutative, then $\mid_l$ and $\mid_r$ coincide.}.
The following lemma collects some well-known facts about divisibility in monoids.
\begin{lemma}\label{lem: divbasics}
For a monoid $\mathbf{M}=\langle M, \cdot, 1 \rangle$, the structure $\langle M, \mid_{\mathfrak{d}} \rangle$ is a preordered set with least element $1$. Moreover, for all $a,b,c \in M$, the following conditions hold:
\begin{enumerate}
    \item If $a \mid_{\mathfrak{d}} 1$, then $a$ admits a right inverse when $\mathfrak{d}=l$, and a left inverse when $\mathfrak{d}=r$;
    \item If $a \mid_{\mathfrak{d}} b$, then $a \mid_{\mathfrak{d}} b \cdot_{\mathfrak{d}} c$;
    \item If $a \mid_{\mathfrak{d}} b$, then $c \cdot_{\mathfrak{d}}a \mid_{\mathfrak{d}} c\cdot_{\mathfrak{d}} b$ (that is, $\mid_{\mathfrak{d}}$ is $\mathfrak{d}$-compatible with multiplication).
\end{enumerate}
\end{lemma}
We now recall a sufficient condition under which divisibility preorders are partial orders. Consider the following quasi-equations of type $\langle 2,0\rangle$:
\begin{align*}
&\tag{CANC} z \cdot_{\mathfrak{d}} x \approx z\cdot_{\mathfrak{d}} y \Rightarrow x \approx y\qquad\text{(where $\mathfrak{d}\in\{l,r\}$)}\label{eq: CANC}\\
&\tag{CON$_1$} x \cdot y \approx 1 \Rightarrow x \approx 1\label{eq: CON1}\\
&\tag{CON$_2$} x \cdot y \approx 1 \Rightarrow y \approx 1\label{eq: CON2}
\end{align*}
A monoid is said to be \textit{cancellative} \cite{CP61} if it satisfies (\ref{eq: CANC}), and \textit{conical}~\cite{Wehrung96} if it satisfies both (\ref{eq: CON1}) and (\ref{eq: CON2}).
Observe that cancellativity allows to reduce equations in analogy with the simplifications permitted by inverses in groups\footnote{In fact, any group trivially satisfies (\ref{eq: CANC}). For this reason, cancellativity plays a key role in the study of the embeddability of semigroups into groups. In particular, it is a necessary and sufficient condition for embedding a commutative semigroup into a group, whereas in the non-commutative case it is only necessary. For details, see \cite[§1.10]{CP61} and \cite[§12]{CP67}.}. Conicality, in turn, expresses the non-existence of non-trivial invertible elements in a monoid. In other words, a monoid is conical when the only element admitting both a right and a left inverse is the unit (which therefore coincides with its inverses). The following examples show that cancellativity and conicality are independent notions.
\begin{example}\label{ex: ConCanc}
Free monoids, as well as $\langle \mathbb{N}, +, 0 \rangle$, are both cancellative and conical. 
A monoid with an absorbing element cannot be cancellative, but it may still be conical (the most obvious example is $\langle \mathbb{N}, \cdot, 1 \rangle$). A well-known class of conical non-cancellative monoids is the variety of join-semilattices with zero (dually, meet-semilattices with unit). On the other hand, a noteworthy example of a cancellative but non-conical monoid is $\langle \{0,1\}, \oplus, 0 \rangle$, where $\oplus$ is the XOR operator. This monoid is cancellative, since $c \oplus a $ never coincides with $c \oplus b$ when $a\neq b$\footnote{Note that $\langle \{0,1\}, \oplus, 0 \rangle$ is the inverse-free reduct of the Abelian group $\langle \{0,1\}, \oplus, ', 0 \rangle$, where $'$ is defined by $x' \approx x$ (i.e., each element is its own inverse).}. Moreover, $1 \oplus 1 = 0$ although $1 \neq 0$, so conicality does not hold. Finally, the monoid of self-maps $f\colon A\to A$ on a set $A$ is neither cancellative nor conical. Indeed, $h\circ f=h\circ g$ does not generally imply $f=g$, and the identity map $\mathrm{id}_A$ is not the unique invertible element, as every bijection is invertible too.
\end{example}
Let now $\mathbf{M}$ be a cancellative, conical monoid. For $\mathfrak{d}\in\{l,r\}$ and $a,b \in M$, suppose that $a\mid_{\mathfrak{d}} b$ and $b \mid_{\mathfrak{d}} a$. Then there exist $c,d \in M$ such that $b=a \cdot_{\mathfrak{d}} c$ and $a=b \cdot_{\mathfrak{d}} d$. Substituting $a \cdot_{\mathfrak{d}} c$ for $b$ in the second expression gives $a = (a \cdot_{\mathfrak{d}} c)\cdot_{\mathfrak{d}} d = a \cdot_{\mathfrak{d}} (c\cdot_{\mathfrak{d}} d)$. As $a=a \cdot_{\mathfrak{d}} 1$, by cancellativity it follows that $c\cdot_{\mathfrak{d}} d = 1$. Since $\mathbf{M}$ is conical, this implies $c=d=1$. Hence $a=b\cdot_{\mathfrak{d}}1=b$. We just showed that:
\begin{proposition}
If a monoid is both cancellative and conical, then its divisibility preorders are antisymmetric.
\end{proposition}
In this paper, we focus on divisibility in free monoids. In this setting, it is straightforward to observe that left (resp. right) divisibility corresponds to the \emph{prefix} (resp. \emph{suffix}) \emph{relation} on strings. The following result is a corollary of the preceding discussion.
\begin{proposition}\label{prop: divstrings}
	Let $\mathbf{F}_{\Sigma}$ be a free monoid. 
	Then, for $\mathfrak{d} \in \{l, r\}$, the structure 
	$\langle \mathbf{F}_{\Sigma}, \mid_{\mathfrak{d}} \rangle$ 
	is a partially $\mathfrak{d}$-ordered monoid with least element~$\varepsilon$.
\end{proposition}
\subsection{Unary residuation}
Let $\mathbf{A}=\langle A, \preceq^{\mathbf{A}}\rangle$ and $\mathbf{B}=\langle B, \preceq^{\mathbf{B}}\rangle$ be partially ordered sets (posets). A map $f\colon A\rightarrow B$ is \textit{residuated} if there exists $g\colon B\rightarrow A$ such that
\begin{equation}
\tag{1RES} f(a) \preceq^{\mathbf{B}} b \Leftrightarrow a \preceq^{\mathbf{A}} g(b)\text{,}  \label{eq: R1} 
\end{equation}
for all $a\in A$ and all $b \in B$. When (\ref{eq: R1}) holds, we say that $g$ is the (\textit{right}) \textit{residual} of $f$, and we call $\langle f,g \rangle$ a \textit{residuated pair}. We refer to (\ref{eq: R1}) as the \textit{unary residuation law}. The following result provides a necessary and sufficient condition for (\ref{eq: R1}).
\begin{proposition}\cite[Lem. 3.2]{yellowbook}\label{prop: eqres}
	Let $\mathbf{A}$ and $\mathbf{B}$ be posets, and let 
	$f \colon A \to B$ and $g \colon B \to A$ be maps. 
	Then $\langle f, g \rangle$ is a residuated pair if and only if:
	\begin{enumerate}
		\item both $f$ and $g$ are order-preserving;
		\item $\mathrm{id}_{\mathbf{A}} \preceq^{\mathbf{A}} g \circ f$ and 
		$f \circ g \preceq^{\mathbf{B}} \mathrm{id}_{\mathbf{B}}$,  
		where (with a slight abuse of notation) 
		$\preceq^{\mathbf{A}}$ and $\preceq^{\mathbf{A}}$ denote the pointwise orders on the sets of self-maps on $P$ and $Q$, respectively.
	\end{enumerate}
\end{proposition}
We now recall some basic facts about residuated pairs (for details, see \cite[Lems.~3.1 and 3.3]{yellowbook} and \cite[Lem.~1.2.8]{MPT}).
\begin{proposition}\label{prop: ncres}
Let $\mathbf{A}$, $\mathbf{B}$ be posets, and let $\langle f,g \rangle$ be a residuated pair with $f \colon A \to B$ and $g \colon B \to A$. The following conditions hold:
\begin{enumerate}
    \item for $a\in A$, $b \in B$, $f(a)=\min \{ b\,|\, a \preceq^{\mathbf{A}} g(b) \}$ and $g(b)=\max \{ a \,|\, f(a) \preceq^{\mathbf{B}} b \}$;
    \item $g \circ f$ and $f \circ g$ are, respectively, a closure and an interior operator;
    \item $f \circ g \circ f = f$ and $g \circ f \circ g = g$.
    \item If $\min \mathbf{A}$ (resp. $\max \mathbf{B}$) exists, then $\min \mathbf{B}$ (resp. $\max \mathbf{A}$) exists as well, and $f(\min \mathbf{A}) = \min \mathbf{B}$ (resp. $\max \mathbf{A} = g(\max \mathbf{B})$).
\end{enumerate}
\end{proposition}
Unary residuation is naturally realised by certain pairs of operators defined from binary relations (see again, for example, \cite[pp. 143–144]{yellowbook} or \cite[Ex.~1.2.2]{MPT}). We recall this construction in the following example, introducing a notation that will be adopted throughout the remainder of the paper.
\begin{example}\label{ex: 1res}
Let $A$ and $B$ be sets, and let $R \subseteq A \times B$. The direct image operator $\posv_R$ and the inverse image operator $\pos_R$ associated with $R$ are defined as follows:
\begin{align*}
\posv_R\colon &\wp(A)\rightarrow\wp(B) & \pos_R\colon& \wp(B)\rightarrow\wp(A)\\
    &C \mapsto \lbrace b \in B\,|\, \exists c ( cRb \curlywedge c \in C) \rbrace &
	&D \mapsto \lbrace a\in A\,|\, \exists d (aRd \curlywedge d\in D) \rbrace
\end{align*}
where $C\subseteq A$ and $D\subseteq B$. Note that $\pos_R$ can equivalently be defined as the direct image operator associated with the converse relation $R^{-1}$ (i.e. $\pos_R=\posv_{R^{-1}}$). Both $\posv_R$ and $\pos_R$ are residuated w.r.t. the inclusion order $\subseteq$, with residuals respectively defined by:
\begin{align*}
\nec_R\colon &\wp(B)\rightarrow\wp(A) & \necv_R\colon& \wp(A)\rightarrow\wp(B)\\
    &D \mapsto \lbrace a \in A\,|\, \forall d ( aRd \Rightarrow d \in D) \rbrace &
	&C \mapsto \lbrace b\in B\,|\, \forall c (cRb \Rightarrow c\in C) \rbrace
\end{align*}
\end{example}
In the preceding example, when $A = B$, the pair $\langle A, R \rangle$ may be regarded as a Kripke frame for some modal logic, and the operator definitions give rise to the semantic clauses for the usual connectives $\pos, \nec$ and their “backward” counterparts $\posv, \necv$\footnote{For instance, if $R$ is interpreted as a temporal precedence relation, the resulting structure yields a frame for Prior’s Tense Logic \cite{Prior1957}, with $\pos, \nec, \posv,$ and $\necv$ corresponding, respectively, to the 
tense modalities $F, G, P,$ and $H$.}. Clearly, if $R$ is symmetric (as in the case of the modal logic $\mathsf{S5}$), 
then $\pos_R = \posv_R$ and $\nec_R = \necv_R$. 

In this paper, we consider residuated pairs of the form $\langle \pos, \necv \rangle$ defined from preorders. Let $\mathbf{A}$ be a preordered set. A subset $D\subseteq A$ is a \textit{downset} in $\mathbf{A}$ if $a \in D$ and $b\preceq a$ imply $b \in D$, for all $a,b \in A$. The \textit{principal downset} generated by an element $a \in A$ is the set $\downarrow\!\!a \,:=\{ b \in A \mid b\preceq a\}$. For $B \subseteq A$, we define $\downarrow\!\!B \,:=\bigcup_{b \in B}\downarrow\!\!b $. We call $\downarrow\!\!B$ the downset generated by $B$. It is immediate to to check that $\downarrow$ defines a closure operator on $\langle \wp(A), \subseteq \rangle$. Moreover, observe that $\bigcup_{b \in B}\downarrow\!\!b=\{a \in A \mid \exists b(a\preceq b \curlywedge b \in B)\}$. Hence $\downarrow\!\!B=\pos_{\preceq} B$. By (\ref{eq: R1}), Prop. \ref{prop: ncres} (1), and the axioms for closure operators, it follows that $\necv_{\preceq}$ is an interior operator defined by $C \mapsto \bigcup\{ B \mid \pos_{\preceq} B \subseteq C \}$, for any $C \subseteq A$.
\section{Trace-based specification of concurrent systems}
\label{sect: thcomp}
\subsection{Modelling framework}\label{sect: modframework}
The mathematical representation of concurrency can be given within a variety of modelling paradigms \cite{SNW1996,WM1995}. 
We place ourselves within the framework of interleaving concurrency and take labelled transition systems (LTSs) \cite{Keller1976} as our reference model of computation. This is a natural choice, as LTSs provide the operational semantics for process calculi like $\mathsf{CCS}$ \cite{Milner80} or $\mathsf{CSP}$ \cite{Hoare85}; however, the following discussion can be readily adapted to automata, Kripke structures, and similar models. 
\begin{definition}\label{def: LTS}
A \textit{labelled transition system} (\textit{LTS}) is a triple $\mathcal{T}=\langle S, \mathit{Act}, T\rangle$ where: 
\begin{itemize}
    \item $S$ is a nonempty set of \textit{states};
    \item $\mathit{Act}$ is a nonempty set of \textit{actions};
    \item $T\subseteq S \times \mathit{Act} \times S$ is a \textit{transition relation}.
\end{itemize}
A \textit{rooted LTS} is a quadruple $\mathcal{T}=\langle S, \mathit{Act}, T,s_0\rangle$ where $\langle S, \mathit{Act}, T\rangle$ is an LTS and $s_0\in S$ is a distinguished state called the \textit{initial state}.
\end{definition}
All LTSs considered in this work are assumed to be rooted, so we henceforth write “LTS” for “rooted LTS”. 
\begin{definition}
Given an LTS $\mathcal{T}=\langle S, \mathit{Act}, T,s_0\rangle$, we refer to elements of $T$ as \textit{transitions} and, for ${s,s'\in S}$ and ${\mathtt{a}\in \mathit{Act}}$, we write $s \xrightarrow{\mathtt{a}}s'$ if ${\langle s,\mathtt{a},s'\rangle \in T}$. A \textit{path} in $\mathcal{T}$ is a (possibly infinite) sequence of transitions such that whenever $s_i \xrightarrow{\mathtt{a}_i} s_j$ and $s_k \xrightarrow{\mathtt{a}_j} s_h$ occur consecutively, we have $s_j=s_k$. We denote paths using the streamlined notation $\cdots s_i \xrightarrow{\mathtt{a}_i} s_j\xrightarrow{\mathtt{a}_j} s_k \cdots $. A path is \textit{rooted} if its first transition starts from $s_0$.
\end{definition}
\begin{rexample}
We consider an LTS, which we call $\mathcal{E}$, modelling a simplified request-response protocol (system architecture is displayed in Figure~\ref{fig: run_ex}). At the initial state, the~action $\mathtt{conn}$ establishes a connection between client and server. The client then transmits a request via $\mathtt{snd}$. At this stage, the system evolves non-deterministically depending on the server’s response: if the request is correctly received, the server sends an acknowledgement and the requested data (action $\mathtt{ack}$, from $s_2$ to $s_4$); otherwise, the server signals a transmission problem via a negative acknowledgment (action $\mathtt{nack}$, from $s_2$ to $s_3$) and the protocol resets (action $\mathtt{end}$, from $s_3$ to $s_0$). In $s_4$, the client may either terminate the protocol by executing $\mathtt{end}$ and returning to $s_0$ or initiate another request (action $\mathtt{req}$, from $s_4$ to $ s_1$).
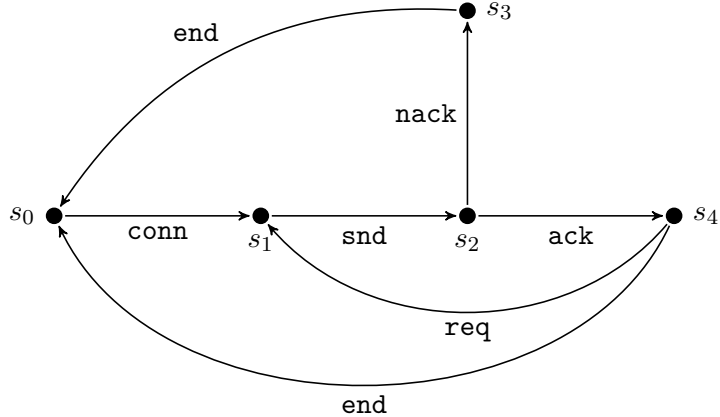
\begin{figure}[ht!]
\begin{tikzpicture}[modal]

\node[point] (root) [label = 180: {$s_0$}]                   {};
\node[point] (p1)   [right = of root, label = 270: {$s_1$}] {};
\node[point] (p2)   [right = of p1, label = 270: {$s_2$}]       {};
\node[point] (p4)   [right = of p2, label = 0: {$s_4$}]  {};
\node[point] (p3)   [above = of p2, label = 0: {$s_3$}]  {};
\path[->] (root) edge node[swap] {\texttt{conn}} (p1);
\path[->] (p1) edge node[swap] {\texttt{snd}}  (p2);
\path[->] (p2) edge node {\texttt{nack}}  (p3);
\path[->] (p2) edge node[swap] {\texttt{ack}}  (p4);
\path[->, bend right=30] (p3) edge node[swap] {\texttt{end}}  (root);
\path[->, bend left=65] (p4) edge node {\texttt{end}}  (root);
\path[->, bend left=50] (p4) edge node {\texttt{req}}  (p1);
\end{tikzpicture}
\caption{Request-response protocol.}
\label{fig: run_ex}
\end{figure}
\end{rexample}

Paths in LTSs provide an immediate representation of the behaviour of the systems being modelled. However, to formalise (temporal) properties of concurrent processes and to study their structure, we require an additional abstraction, which is provided by the notion of an \textit{execution trace}. Since this paper does not address any specific class of concurrent systems, for the general definitions in the present subsection we adopt the perspective of “mixed” executions, following \cite{Rosu2012,MaPa17}. In particular, by slightly adapting the modelling framework proposed in \cite[§7.2.2]{clsc}, we represent traces as finite or infinite sequences of state/action pairs, and classify them as valid or invalid according to a formal notion of compatibility with the transition relation.
\begin{definition}\label{def: trace}
For $\mathcal{T} = \langle S, \mathit{Act}, T, s_0\rangle$ an LTS, let $\Sigma_{\mathcal{T}} := S \times \mathit{Act}$.  
A \textit{$\mathcal{T}$-trace} is any sequence of elements from $\Sigma_{\mathcal{T}}$, possibly empty or infinite.
\end{definition}
For $\langle s, \mathtt{a}\rangle \in \Sigma_{\mathcal{T}}$, we define projections $\mathit{proj}_1\colon \langle s, \mathtt{a}\rangle \mapsto s $ and $\mathit{proj}_2\colon \langle s, \mathtt{a}\rangle \mapsto \mathtt{a} $. Treating $\Sigma_{\mathcal{T}}$ as an alphabet, finite $\mathcal{T}$-traces correspond to words in $\Sigma_{\mathcal{T}}^*$, and infinite ones to $\omega$-words in $\Sigma_{\mathcal{T}}^{\omega}$. Whenever it is necessary to distinguish finite from infinite traces, we use $w, u, v, \dots$ (possibly with annotations) to denote elements of~$\Sigma_{\mathcal{T}}^*$, and $\alpha, \beta, \gamma, \dots$ to denote elements of~$\Sigma_{\mathcal{T}}^{\omega}$. For mixed traces, we adopt the same notation as in the finite case. Let $\Sigma_{\mathcal{T}}^{\infty} :=\Sigma_{\mathcal{T}}^*\cup\Sigma_{\mathcal{T}}^{\omega} $. For $w \in \Sigma_{\mathcal{T}}^{\infty} $ and $i \in\mathbb{N}$, $w[i]$, $w[\triangleright i]$, and $w[i\triangleright ]$ respectively denote the $i$-th pair occurring in $w$, the prefix $w[0],\dots,w[i]$, and the suffix $w[i],w[i+1],\dots$. If $w$ is finite, we denote by $|w|$ its length. We say that $v\in\Sigma_{\mathcal{T}}^*$ is a \textit{prefix} of $w$ (written $v \sqsubseteq w$) if there exists $u\in\Sigma_{\mathcal{T}}^{\infty} $ such that $w=vu$. Finally, we define $\mathit{pref}(w):=\{ v  \in\Sigma_{\mathcal{T}}^*\mid v \sqsubseteq w  \}$.
\begin{definition}\label{def: vtrace}
Let $\mathcal{T}$ be an LTS. A finite $\mathcal{T}$-trace $w \in \Sigma_{\mathcal{T}}^*$ is \textit{valid} if and only if:
\begin{equation}
\tag{VAL$_*$} \text{for all $i<|w|-1$, $\mathit{proj}_1(w[i])\xrightarrow{\mathit{proj}_2(w[i])}\mathit{proj}_1(w[i+1])$}  \label{eq: VAL*} 
\end{equation}
Similarly, an infinite $\mathcal{T}$-trace $\alpha\in \Sigma_{\mathcal{T}}^{\omega}$ is \textit{valid} if and only if:
\begin{equation}
\tag{VAL$_{\omega}$} \text{for all $i\in \mathbb{N}$, $\mathit{proj}_1(\alpha[i])\xrightarrow{\mathit{proj}_2(\alpha[i])}\mathit{proj}_1(\alpha[i+1])$}  \label{eq: VAL1o} 
\end{equation}
A $\mathcal{T}$-trace $w \in \Sigma_{\mathcal{T}}^{\infty}$ is \textit{invalid} if it is not valid.
\end{definition}
Definitions \ref{def: trace} and \ref{def: vtrace} show that the notion of a trace captures and extends that of a path, as each valid trace corresponds to a unique path. From the perspective of formal verification, the distinction between valid and invalid traces may appear redundant, since only traces corresponding to paths are relevant in that context. However, as we will see, it becomes useful when one aims to provide an algebraic characterisation of properties of finite traces as subsets of a free monoid. This in no way affects verification: determining whether a valid trace satisfies a property reduces to checking whether it belongs to a set of strings with certain features, regardless of whether the set contains only valid traces or also includes invalid ones.
\begin{rexample}
Below are examples of valid $\mathcal{E}$-traces with their associated paths.
\begin{table}[H]
\begin{tabular}{|c|c|c|}
\hline
Class & Example $\mathcal{E}$-trace & Path \\
\hline
\multirow{2}{*}[-0.83em]{\centering$\Sigma_{\mathcal{E}}^*$} & \makecell{$\langle s_0, \mathtt{conn}\rangle \langle s_1, \mathtt{snd}\rangle \langle s_2, \mathtt{ack}\rangle \langle s_4, \mathtt{end}\rangle$} &
\makecell{$s_0 \xrightarrow{\mathtt{conn}} s_1 \xrightarrow{\mathtt{snd}} s_2 \xrightarrow{\mathtt{ack}} s_4$} \\
\cline{2-3}
& \makecell{$(\langle s_1, \mathtt{snd}\rangle \langle s_2, \mathtt{nack}\rangle \langle s_3, \mathtt{end}\rangle \langle s_0, \mathtt{conn}\rangle)^3$} &
\makecell{$\pi \xrightarrow{\mathtt{conn}} \pi \xrightarrow{\mathtt{conn}} \pi$, where \\ $\pi = s_1 \xrightarrow{\mathtt{snd}} s_2 \xrightarrow{\mathtt{nack}} s_3 \xrightarrow{\mathtt{end}} s_0$} \\
\hline
\multirow{2}{*}{$\Sigma_{\mathcal{E}}^{\omega}$} & \makecell{$(\langle s_0, \mathtt{conn}\rangle \langle s_1, \mathtt{snd}\rangle \langle s_2, \mathtt{nack}\rangle \langle s_3, \mathtt{end}\rangle)^{\omega}$} &
\makecell{$s_0 \xrightarrow{\mathtt{conn}} s_1 \xrightarrow{\mathtt{snd}} s_2 \xrightarrow{\mathtt{nack}} s_3 \xrightarrow{\mathtt{end}} s_0 \cdots$} \\
\cline{2-3}
& \makecell{$(\langle s_1, \mathtt{snd}\rangle \langle s_2, \mathtt{ack}\rangle \langle s_4, \mathtt{req}\rangle)^{\omega}$} &
\makecell{$s_1 \xrightarrow{\mathtt{snd}} s_2 \xrightarrow{\mathtt{ack}} s_4 \xrightarrow{\mathtt{req}} s_1 \cdots$} \\
\hline
\end{tabular}
\end{table}
The following are instead examples of invalid $\mathcal{E}$-traces:
\[\text{any element of $\Sigma_{\mathcal{E}}$} \qquad \langle s_1, \mathtt{snd}\rangle(\langle s_2, \mathtt{ack}\rangle\langle s_2, \mathtt{nack}\rangle)^*\qquad \langle s_0, \mathtt{conn}\rangle^{\omega}\]
\end{rexample}
\subsection{Trace properties for LTSs (mixed case)}
We now discuss trace properties for LTSs, beginning with the general case of mixed executions and then turning to finite-trace models.
\begin{definition}
Let $\mathcal{T}$ be an LTS. A \textit{trace property} for $\mathcal{T}$ is any set $P\subseteq \Sigma_{\mathcal{T}}^{\infty}$. A $\mathcal{T}$-trace $w$ \textit{satisfies} a property $P$ if $w \in P$. 
\end{definition}
\begin{remark}\label{rem: tbspec}
Note that an LTS $\mathcal{T}$ may be identified with a special trace property, namely the set $V_{\mathcal{T}}$ of its valid traces. In our setting, let $\mathit{FV}_{\mathcal{T}}:=\{ w \in \Sigma_{\mathcal{T}}^*\mid \text{$w$ satisfies (\ref{eq: VAL*})} \}$ and $\mathit{IV}_{\mathcal{T}}:=\{ w \in \Sigma_{\mathcal{T}}^{\omega}\mid \text{$w$ satisfies (\ref{eq: VAL1o})} \}$. Then $V_{\mathcal{T}}:=\mathit{FV}_{\mathcal{T}} \cup \mathit{IV}_{\mathcal{T}}$.
\end{remark}
\begin{definition}\label{def: SLm}
For $\mathcal{T}$ an LST, let $P\subseteq \Sigma_{\mathcal{T}}^{\infty}$. Consider the following two conditions:
\begin{align}
\tag{SP$_{\infty}$} &\forall w \in \Sigma_{\mathcal{T}}^{\infty}\;(w \in P \Leftrightarrow \mathit{pref}(w)\subseteq P)  \label{eq: SP}\\ 
\tag{LP$_{\infty}$} &\forall w\in \Sigma_{\mathcal{T}}^*\;\exists w'\in P \; (w \in \mathit{pref}(w'))  \label{eq: LP} 
\end{align}
Then $P$ is a \textit{safety property} \cite{Rosu2012} if it satisfies (\ref{eq: SP}) and is a \textit{liveness property} \cite{MaPa17} if it satisfies (\ref{eq: LP}).
\end{definition}
\begin{remark}
$\emptyset$ is a safety property, while $\Sigma_{\mathcal{T}}^*$ and $\Sigma_{\mathcal{T}}^{\infty}$ are both safety and liveness properties.
\end{remark}
The computational interpretation of conditions (\ref{eq: SP}) and (\ref{eq: LP}) is well known \cite{AlpSch85}. A safety property $S$ is satisfied by a trace $w$ if $S$ is preserved in each prefix of $w$. If $w$ is valid, this amounts to saying that no event violating $S$ occurs during the computation; if it did, it would be observable in some finite subcomputation—possibly in the last transition of the path associated with $w$, when $w$ is finite.
In turn, an execution satisfies a liveness property $L$ if some event witnessing $L$ eventually occurs during the computation—possibly in the limit, in the case of infinite behaviours. Formally, any finite trace $w \in \Sigma_{\mathcal{T}}^*$ admits an extension, finite or infinite, to a trace $w'\in \Sigma_{\mathcal{T}}^{\infty}$ that satisfies $L$.
\begin{remark}
A classical theorem by Alpern and Schneider \cite[Thm. 1]{AlpSch85} states that in infinite-trace systems, every property is the intersection of a safety and a liveness property. This is established via a topological argument, as in that context safety and liveness correspond, respectively, to closure and density in the Plotkin topology of observable properties\footnote{For further details, see Smyth \cite{Smyth1993}.}. Pasqua and Mastroeni \cite[Prop. 3]{MaPa17} generalised this result to mixed traces, showing that (\ref{eq: SP}) and (\ref{eq: LP}) define the closed and dense sets in the topology on $\Sigma_{\mathcal{T}}^{\infty}$ induced by the closure operator $P\mapsto \lim(\bigcup_{w \in P}\mathit{pref}(w))$, where $\lim$ is the Eilenberg-limit operator \cite{Eilenberg74}.
\end{remark}
\begin{rexample}
Consider the following system policies for $\mathcal{E}$:
\begin{enumerate}
    \item The client never receives a message from the server before having sent a request.
    \item Whenever the client forwards a request to the server, it receives a response.
\end{enumerate}
Following the terminology of Manna and Pnueli~\cite{MannaPnueli90}, policies~(1) and~(2) are instances of the \textit{causal dependence} and \textit{responsiveness} property patterns, respectively.

Policy~(1) prescribes that neither state~$s_3$ nor state~$s_4$ can be reached before the system has entered state~$s_2$. In other words, $\mathcal{E}$ cannot enable actions~$\mathtt{ack}$ or~$\mathtt{nack}$ without having previously enabled action~$\mathtt{snd}$.
Abstractly, this policy corresponds to the set $P_1$ of $w \in\Sigma_{\mathcal{E}}^{\infty}$ such that, for every substring~$w[i-1]w[i]$ with $\mathit{proj}_2(w[i-1]) \in \{\mathtt{ack}, \mathtt{nack}\}$, there exists~$j < i$ with $\mathit{proj}_2(w[j-1]) = \mathtt{snd}$. Observe that, for any $w \in P_1$ and $v \sqsubseteq w$, the prefix~$v$ also satisfies~$P_1$. If $v$ contains a pair of the form $\langle s, \mathtt{ack}\rangle$ or $\langle s, \mathtt{nack}\rangle$, it must be preceded by some pair $\langle s', \mathtt{snd}\rangle$. If $v$ contains no $\langle s, \mathtt{ack}\rangle$ or $\langle s, \mathtt{nack}\rangle$ pair, then the property holds vacuously. Conversely, let $w \in \Sigma_{\mathcal{T}}^{\infty}$ be such that $\mathit{pref}(w) \subseteq P_1$. If $w$ is finite, then trivially $w \in P_1$; if $w$ is infinite, suppose for a contradiction that $w \notin P_1$. Then there exists $i \in \mathbb{N}$ such that $w[\triangleright i] \notin P_1$, but this is impossible since all prefixes of $w$ are in $P_1$.
Therefore,~$P_1$ is a safety property and can be formalised both in classical~$\mathsf{LTL}$~\cite{pn} and in its finite-trace variant~$\mathsf{LTL}_f$~\cite{DeGVardi13} by means of the following formula \cite{Schnoebelen2003} (where $U$ is the \textit{until} operator):
\[ \neg (\neg\mathit{request}\mathrel{U} (\mathit{response}\wedge \neg \mathit{request})) \]
Policy~(2) states that once the system reaches~$s_2$, it is guaranteed to subsequently reach either state~$s_3$ or state~$s_4$. Therefore, enabling~$\mathtt{snd}$ implies the eventual execution of~$\mathtt{ack}$ or~$\mathtt{nack}$. 
We can thus define the abstract property corresponding to this policy as the set~$P_2$ of $\mathcal{E}$-traces $w \in \Sigma_{\mathcal{E}}^{\infty}$ such that, for every substring~$w[i-1]w[i]$ with $\mathit{proj}_2(w[i-1]) = \mathtt{snd}$, there exists~$j > i$ with $\mathit{proj}_2(w[j-1]) \in \{\mathtt{ack}, \mathtt{nack}\}$.
$P_2$ is a liveness property: any $v \in \Sigma_{\mathcal{E}}^*$ is indeed a prefix of some~$w\in P_2$. The formalisation of~$P_2$ in~$\mathsf{LTL}$ or~$\mathsf{LTL}_f$ is given by the following formula (where $G$ and $F$ are the \textit{globally} and \textit{eventually} operators, respectively):
\[G (\mathit{request}\rightarrow F(\mathit{response}))\]
Observe that all valid $\mathcal{E}$-traces satisfy both~$P_1$ and~$P_2$, that is, $V_{\mathcal{E}} \subseteq P_1$ and~$V_{\mathcal{E}} \subseteq P_2$.
\end{rexample}
\subsection{Trace properties for LTSs (finite case)}\label{sect: ftraces}
We proceed to examine f-properties. In this case, for an LTS $\mathcal{T}$, our domain of reference is simply $\wp(\Sigma_{\mathcal{T}}^*)$, while $\mathcal{T}$ itself may be specified as the f-property $\mathit{FV}_{\mathcal{T}}$ defined in Remark~\ref{rem: tbspec}. As shown in Prop. \ref{prop: divstrings}, the prefix relation $\sqsubseteq$ is the left divisibility order on $\Sigma_{\mathcal{T}}^*$; hence, for $w \in \Sigma_{\mathcal{T}}^*$, $\mathit{pref}(w)$ is the principal downset $\downarrow\!w$ in $\langle \Sigma_{\mathcal{T}}^*,\sqsubseteq\rangle$. The inverse image operator $\pos=\pos_{\sqsubseteq}$ on $\wp(\Sigma_{\mathcal{T}}^*)$ defined~by
\[ \pos\! P:=\bigcup_{w\in P}\mathit{pref}(w)=\{v \in \Sigma_{\mathcal{T}}^*\mid \exists w(v\sqsubseteq w \curlywedge w\in P)\}\]
is called the \textit{prefix-closure operator}. Since $\sqsubseteq$ is a partial order, $\pos$ is a topological closure operator on $\langle\wp(\Sigma_{\mathcal{T}}^*),\subseteq\rangle$, i.e. for $ P,Q\subseteq \Sigma_{\mathcal{T}}^*$, it satisfies the Kuratowski axioms:
\begin{itemize}
    \item[(D$_1$)] $P \subseteq \pos\! P$;
    \item[(D$_2$)] $\pos\!\pos\!P= \pos\!P$;
    \item[(D$_3$)] $\pos(P \cup Q) = \pos\!P \cup \pos\!Q$;
    \item[(D$_4$)] $\pos\emptyset = \emptyset$
\end{itemize}
Ro\c{s}u \cite{Rosu2012} observed that the notion of safety for f-properties coincides precisely with being a fixpoint of $\pos$, while Pasqua and Mastroeni \cite{MaPa17} showed that restricting (\ref{eq: LP}) to finite traces characterises liveness f-properties as those whose image under $\pos$ is the entire $\Sigma_{\mathcal{T}}^*$. We can thus restate Definition \ref{def: SLm} as follows:
\begin{definition}
For $\mathcal{T}$ an LST, let $P\subseteq \Sigma_{\mathcal{T}}^*$. Consider the following two conditions:
\begin{align}
\tag{SP$_*$} &  \pos\! P =P \label{eq: SP*}\\ 
\tag{LP$_*$} & \pos\! P= \Sigma_{\mathcal{T}}^* \label{eq: LP*} 
\end{align}
Then $P$ is a \textit{safety f-property} \cite{Rosu2012} if it satisfies (\ref{eq: SP*}) and is a \textit{liveness f-property} \cite{MaPa17} if it satisfies (\ref{eq: LP*}).
\end{definition}
One easily checks that restricting (\ref{eq: SP}) to $\Sigma_{\mathcal{T}}^*$ yields a condition equivalent to~(\ref{eq: SP*}).
\begin{proposition}\label{prop: sbox}
For all $P\subseteq \Sigma_{\mathcal{T}}^*$, $\pos\!P =P$ iff $P=\{w \in \Sigma_{\mathcal{T}}^* \mid \mathit{pref}(w)\subseteq P\}$.
\begin{proof}
Suppose that $P$ is a fixpoint of $\pos$. Since $\pos\! P=\bigcup_{w \in P}\mathit{pref}(w)$, it follows that $\mathit{pref}(v)\subseteq \pos\! P=P$ for all $v \in P$. This establishes the left-to-right implication. For the converse, assume \textit{ad absurdum} that $P\neq\pos\! P$ and that $P=\{w \in \Sigma_{\mathcal{T}}^* \mid \mathit{pref}(w)\subseteq P\}$. Since $\pos$ is a closure operator, we then have $P\subsetneq\pos\!P$, hence there exists $v\in \pos\! P$ such that $v \not\in P$. But then, for some $u \in P$, we must have $v \in \mathit{pref}(u)$, and as $\mathit{pref}(u)\subseteq P$, it follows that $v \in P$—a contradiction.
\end{proof}
\end{proposition}
\begin{rexample}
The properties~$P_1$ and~$P_2$ introduced in~§\ref{sect: modframework} for mixed traces can be restricted to the finite case by defining
$P_{1,f} := P_1 \cap \Sigma_{\mathcal{E}}^*$ and $P_{2,f} := P_2 \cap \Sigma_{\mathcal{E}}^*$. In this setting, the temporal logic specification naturally reduces to~$\mathsf{LTL}_f$ alone.
\end{rexample}
\begin{remark}
Pasqua and Mastroeni~\cite[§3.2]{MaPa17} define liveness f-properties as follows:
\[
L \subseteq \Sigma_{\mathcal{T}}^* \text{ is liveness}
\;\;\Leftrightarrow\;\;
\forall w \in \Sigma_{\mathcal{T}}^* \; \exists w'\in L\;
(w \sqsubseteq w').
\]
Observe that the existentially quantified subformula is precisely the requirement for $w$ to belong to $\pos\!L$. Hence we can rewrite:
\[
L \subseteq \Sigma_{\mathcal{T}}^* \text{ is liveness}
\;\;\Leftrightarrow\;\;
\forall w \in \Sigma_{\mathcal{T}}^* (w \in \pos\!L).
\]
Therefore $L$ is a liveness f-property if and only if
$\Sigma_{\mathcal{T}}^* \subseteq \pos\!L$, and thus $\Sigma_{\mathcal{T}}^* = \pos\!L$.
\end{remark}
The following result is an immediate consequence of the above discussion.
\begin{proposition}\textnormal{\cite[Prop. 1]{MaPa17}}
Safety and liveness f-properties are, respectively, the closed and dense sets in the topology over $\Sigma_{\mathcal{T}}^*$ induced by $\pos$.
\end{proposition}
\section{An algebraic model for f-properties}\label{sect: algsect}
In this section, we develop an algebraic framework for representing f-properties, with a focus on those specified via the prefix-closure operator $\pos$ and its residual $\necv=\necv_{\sqsubseteq}$.
As noted in the introduction, for a free monoid $\mathbf{F}_{\Sigma}=\langle \Sigma^*, \cdot, \varepsilon\rangle$, several algebras can be defined on $\wp(\Sigma^*)$. In what follows, we adopt a “minimal” approach, taking as our base structure the bounded $\ell$-monoid $\wp(\mathbf{F}_{\Sigma})=\langle \wp(\Sigma^*), \cap, \cup, \cdot, \{\varepsilon\}, \emptyset, \Sigma^*\rangle$, where multiplication in the monoid reduct $\langle \wp(\Sigma^*), \cdot, \{\varepsilon\}\rangle$ is given by the standard language concatenation defined as $P\cdot Q := \{vw \mid v \in P \curlywedge w \in Q\}$ for $P,Q\subseteq \Sigma^*$. 

Let $\wp(\mathbf{F}_{\Sigma})_+$ be the algebra obtained by expanding $\wp(\mathbf{F}_{\Sigma})$ with the residuated pair $\langle \pos, \necv\rangle$. The algebraic model we are going to introduce, which we refer to as a \textit{closure $\ell$-monoid}, arises as an abstraction of $\wp(\mathbf{F}_{\Sigma})_+$ (see Lemma \ref{lem: VnGEN}). The reasons for this is twofold. First, to the best of our knowledge, there is no complete axiomatisation of $\wp(\mathbf{F}_{\Sigma})$, nor is it known whether such an axiomatisation could be finitary. Second, as we shall see, the proof-theoretic framework adopted in Section~\ref{sect: LMC} is too restrictive to permit the derivation of (the sequent counterparts of) some conditions holding in~$\wp(\mathbf{F}_{\Sigma})_+$.
\subsection{Prefix-closure}
We begin by examining how $\pos$ interacts with the operations of $\wp(\mathbf{F}_{\Sigma})$. In the case of lattice operations, $\pos$ behaves as a standard $\mathsf{S4}$ diamond. Hence, using conditions (D$_1$)–(D$_4$) in §\ref{sect: ftraces}, one readily shows, for instance, that
\[\begin{array}{ccc}
 \text{$P\subseteq Q \Rightarrow \pos\!P\subseteq\pos\!Q$ }&  \pos(P \cap Q)\subseteq \pos\!P\cap\pos\!Q  & \pos(P \cap Q)= \pos(\pos\!P\cap\pos\!Q)   \\
\end{array}\]
for all $P,Q\subseteq \Sigma^*$. As for constants, we already know that $\pos$ preserves $\emptyset$. Of course, $\pos\! \Sigma^* = \Sigma^*$, and since $\varepsilon = \min \langle \Sigma^*, \sqsubseteq \rangle$, we also have $\pos\{\varepsilon\} = \{\varepsilon\}$. What is less immediate, however, is how $\pos$ behaves with respect to language concatenation. First, observe that for $P, Q \subseteq \Sigma^*$, a word $w$ belongs to $\pos(P \cdot Q)$ if there exists $uv\in P\cdot Q$ (with $u \in P$ and $v \in Q$) such that exactly one of the cases in the following table applies.
\begin{table}[H]
    \centering
    \begin{tabular}{|c|c|}
    \hline
    \makecell{Case 1. $w\sqsubseteq u$, $w\neq u$} & \makecell{Case 2. $w = u$}\\
\hline
\makecell{\begin{tikzpicture}
\draw[|-|,dotted]  (0,0)--node[above] {$u$} (2,0);
\draw[|-|,dotted]  (2,0)--node[above] {$v$} (4,0);
\draw[|-|,dotted]  (0,-0.7)--node[above] {$w$} (1,-0.7);
\end{tikzpicture}}  & \makecell{\begin{tikzpicture}
\draw[|-|,dotted]  (0,0)--node[above] {$u$} (2,0);
\draw[|-|,dotted]  (2,0)--node[above] {$v$} (4,0);
\draw[|-|,dotted]  (0,-0.7)--node[above] {$w$} (2,-0.7);
\end{tikzpicture}}  \\
\hline
\makecell{Case 3. $w= uv'$, $v'\sqsubseteq v$, $v'\neq v$} & \makecell{Case 4. $w=uv$}\\
\hline
\makecell{\begin{tikzpicture}
\draw[|-|,dotted]  (0,0)--node[above] {$u$} (2,0);
\draw[|-|,dotted]  (2,0)--node[above] {$v$} (4,0);
\draw[|-|,dotted]  (0,-0.7)--node[above] {$w$} (3,-0.7);
\end{tikzpicture}}   &  \makecell{\begin{tikzpicture}
\draw[|-|,dotted]  (0,0)--node[above] {$u$} (2,0);
\draw[|-|,dotted]  (2,0)--node[above] {$v$} (4,0);
\draw[|-|,dotted]  (0,-0.7)--node[above] {$w$} (4,-0.7);
\end{tikzpicture}}\\
\hline
    \end{tabular}
    \caption{The structure of elements of $\pos(P\cdot Q)$.}
    \label{tab: diamprod}
\end{table}
From the cases in Table~\ref{tab: diamprod}, it follows immediately that the following property holds in $\langle \mathbf{F}_{\Sigma}, \sqsubseteq \rangle$, for all $u, v, w \in \Sigma^*$:
\begin{equation}
\tag{RDP$_{*}$} \text{$w\sqsubseteq uv$ implies $w=u'v'$, with $u'\sqsubseteq u$ and $v'\sqsubseteq v$}  \label{eq: RDP2} 
\end{equation}
Readers having some familiarity with the theory of $\ell$-groups will have recognised in (\ref{eq: RDP2}) a particular instance of the well-known \textit{Riesz Decomposition Property} (\textit{RDP})\footnote{Originally introduced in the context of functional analysis \cite{Riesz29}, specifically in connection with vector lattices, the RDP is defined for $\ell$-groups as follows \cite[p.~3]{anfe88}. Let $\langle \mathbf{G}, \preceq \rangle$ be an (additive) $\ell$-group. Define $G_+ := \{ a \in G \mid 0 \prec a \}$ (this set is called the \textit{positive cone} of $\langle \mathbf{G}, \preceq \rangle$). Let $a_1, \dots, a_n \in G_+$. Then $\langle \mathbf{G}, \preceq \rangle$ has the RDP if $0 \preceq b \preceq \sum_{1 \leq i \leq n} a_i$ implies $b = \sum_{1 \leq i \leq n} b_i$ with $0 \preceq b_i \preceq a_i$, for all $i \in \{1, \dots, n\}$.}. It is a standard fact that the RDP is satisfied by some structures with weaker requirements than $\ell$-groups (e.g. Riesz groups \cite{Fuchs65}). We show that any monoid endowed with a divisibility preorder has the RDP. In the following lemma, for a monoid~$\mathbf{M}$, $\mathfrak{d}\in\{l,r\}$, and elements $a_1, \dots, a_n \in M$, we denote by~$[a_1 \cdots a_n]_{\mathfrak{d}}$ the product $a_1 \cdots a_n$ if~$\mathfrak{d}=l$, and $a_n \cdots a_1$ if~$\mathfrak{d}=r$.

\begin{lemma}
Let $\mathbf{M}$ be a monoid. Then, for $\mathfrak{d}\in\{l,r\}$, $\langle \mathbf{M},\mid_{\mathfrak{d}}\rangle$ has the RDP.
\begin{proof}
We prove by induction on $n\in \mathbb{Z}_+$ that, for all $a_1,\dots,a_n, b\in M$:
\begin{equation}
\tag{RDP} \text{$b \mid_{\mathfrak{d}} [a_1 \cdots  a_n]_{\mathfrak{d}}$ implies $b =[b_1 \cdots  b_n]_{\mathfrak{d}}$, with $b_i \mid_{\mathfrak{d}} a_i$ for $i \in \{1,\dots,n\}$}  \label{eq: RDP}
\end{equation}
The base case ($n = 1$) is immediate. Suppose now that (\ref{eq: RDP}) holds for $n=j$ ($j>1$); that is, whenever $b \mid_{\mathfrak{d}} [a_1 \cdots  a_j]_{\mathfrak{d}}$, we have $b =[b_1 \cdots  b_j]_{\mathfrak{d}}$ with $b_i \mid_{\mathfrak{d}} a_i$ for all $i \leq j$. Let $b \mid_{\mathfrak{d}} [a_1 \cdots  a_j]_{\mathfrak{d}}$, and take $a_{j+1} \in M$.  
By Lem. \ref{lem: divbasics}(2), we have $b \mid_{\mathfrak{d}} [a_1 \cdots  a_j a_{j+1}]_{\mathfrak{d}}$ and $b_j \mid_{\mathfrak{d}} a_j \cdot_{\mathfrak{d}} a_{j+1}$. By the induction hypothesis, we thus have $b_j=b_j' \cdot_{\mathfrak{d}}b_{j+1}$ with $b_j'\mid_{\mathfrak{d}} a_j$ and $b_{j+1}\mid_{\mathfrak{d}} a_{j+1}$. Hence $b =[b_1 \cdots  b_{j-1}b_j'b_{j+1}]_{\mathfrak{d}}$, and (\ref{eq: RDP}) holds for $n = j + 1$.
\end{proof}
\end{lemma} 
\begin{corollary}
$\langle \mathbf{F}_{\Sigma}, \sqsubseteq \rangle$ has the RDP.
\end{corollary}
From the above considerations, we obtain three fundamental laws linking $\cdot$ and $\pos$.
\begin{proposition}
For all $P,Q \subseteq \Sigma^*$, the following conditions hold:
\begin{align}
\tag{K}   \pos(P\cdot Q)&\subseteq\pos\!P \cdot \pos\!Q \label{eq: K}\\ 
\tag{C}  \pos(P\cdot Q)&=\pos\!P \cup (P \cdot\pos\!Q) \label{eq: C} \\
\tag{I}  P \cap Q&\subseteq \pos(P\cdot Q)\label{eq: I}
\end{align}
\begin{proof}
It is straightforward to check that (\ref{eq: K}) is equivalent to (\ref{eq: RDP2}). We further observe that $\langle \mathbf{F}_{\Sigma}, \sqsubseteq \rangle$ may be viewed as an expanded ternary frame in which string concatenation and the prefix order serve, respectively, as accessibility relations for language concatenation and prefix closure on $\wp(\Sigma^*)$\footnote{Abstracting away from the specific structure under consideration, this is nothing but the frame semantics employed for modal Lambek calcului within the framework of categorial type logics \cite{Moortgat96,Kurtonina95,ArBerMo03,ArecesBernardi04}.}. In such a setting, (\ref{eq: RDP2}) is precisely the first-order frame condition corresponding to (\ref{eq: K}).

For the left-to-right inclusion of (\ref{eq: C}), suppose that $w \in \pos(P \cdot Q)$. Then $w \sqsubseteq uv$ for some $u \in P$ and $v \in Q$. Let us reconsider Table~\ref{tab: diamprod}. In Case~1, we have $w = u'$, for $u'$ a proper prefix of $u$, and $v' = \varepsilon$. Case~2 is analogous, except that $w = u = u'$. In either situation, we observe that $w \in \pos\! P$.
In Case~3, we have $w = u v'$ for $v'$ a proper, non-empty prefix of $v$. Since $u \in P$ and $v' \in \pos Q$, it follows that $w \in P \cdot \pos\! Q$.
Finally, in Case~4, where $u' = u$ and $v' = v$, since $u \in P$ and $v \in Q \subseteq \pos Q$, we again obtain $w \in P \cdot \pos Q$.
Conversely, suppose that $w \in \pos\!P \cup (P \cdot \pos\! Q)$. Therefore, $w$ is either a (possibly improper) prefix of some string in $P$, or it is obtained by concatenating a string in $P$ with a (possibly improper) prefix of some string in $Q$.
Observe that these conditions are precisely those expressed in Table~\ref{tab: diamprod}; hence $w \in \pos(P \cdot Q)$.

Finally, suppose that $w \in P \cap Q$. Then $w^2 \in P \cdot Q$, and since $w \sqsubseteq w^2$, it follows that $w \in \pos(P \cdot Q)$. This establishes (\ref{eq: I}).
\end{proof}
\end{proposition}
\begin{remark}
One easily verifies that, for all $w, u, v \in \Sigma^*$, 
the first-order frame conditions corresponding to (\ref{eq: C}) and (\ref{eq: I}) are, respectively,
\[w \sqsubseteq uv \Leftrightarrow (w \sqsubseteq u \curlyvee \exists v' \, (v' \sqsubseteq v \curlywedge w \sqsubseteq uv')) \quad\text{and}\quad w \sqsubseteq w^2.\]
Observe that~(\ref{eq: K}) and~(\ref{eq: I}) hold for every operator~$\pos_{\mid_{\mathfrak{d}}}$ defined from a divisibility preorder.
\end{remark}
\subsection{On the operator \texorpdfstring{$\necv$}{◫}}
We now briefly discuss some properties of the residual~$\necv$ of $\pos$, showing how it naturally arises in the trace-based specification of concurrent systems and why it provides a useful addition to our algebraic model for~f-properties. 

We first review some basic algebraic laws involving~$\necv$ (complete proofs will be given in the next subsection, together with proofs of further equations). We have seen that, in the $\{\cap, \cup, \pos, \emptyset, \Sigma^*\}$-reduct of~$\wp(\mathbf{F}_{\Sigma})_+$, prefix closure is essentially an $\mathsf{S4}$ diamond. When examining the interaction between~$\necv$ and~$\pos$, however, the residuation law leads us to a setting that is instead reminiscent of~$\mathsf{S5}$, as well as of tense logics with both future and past modalities. In particular, by ~Prop.~\ref{prop: eqres}(2), we have
\[
\begin{array}{ccc}
	P \subseteq \necv\!\pos\!P & \text{and} &
	\pos\!\necv\!P \subseteq P 
\end{array}
\]
for all $P\subseteq \Sigma^*$. Since the prefix relation is a partial order, by residuation, $\necv$ is a topological interior operator on~$\langle \wp(\Sigma^*), \subseteq\rangle$, that is, it satisfies the \emph{duals} of the four Kuratowski axioms, namely
\begin{itemize}
	\item[(B$_1$)] $\necv\!P \subseteq  P$;
	\item[(B$_2$)] $\necv\!\necv\!P = \necv\!P$;
	\item[(B$_3$)] $\necv(P \cap Q) = \necv\!P \cap \necv\!Q$;
	\item[(B$_4$)] $\necv\!\Sigma^* = \Sigma^*$.
\end{itemize}
for all~$P,Q \subseteq \Sigma^*$. Analogously to the case of~$\pos$, in the $\{\cap,\cup,\necv\}$-reduct of~$\wp(\mathbf{F}_{\Sigma})_+$, $\necv$ is an $\mathsf{S4}$ box. Indeed, from~(B$_1$)--(B$_4$), one derives, for example:
\[
\begin{array}{ccc}
	\necv(P \cup Q)\subseteq \necv\!P\cup\necv\!Q  & \text{and} & 
	\necv(P \cup Q)= \necv(\necv\!P\cup\necv\!Q)
\end{array}
\]
Moreover, by~(B$_1$), the empty set is a fixed point of~$\necv$. By~Prop.~\ref{prop: ncres}(2), we also know that the compositions $\necv\!\pos$ and $\pos\!\necv$ are, respectively, a closure and an interior operator. These modalities in fact coincide with~$\pos$ and~$\necv$ themselves: applying~(\ref{eq: R1}) to $\pos\!\pos\! P = \pos\! P$ and $\necv\!\necv\!P = \necv\! P$ yields $\pos\! P = \necv\!\pos\! P$ and $\necv\! P = \pos\!\necv\! P$, respectively. As for the interaction with language concatenation, one can show, for example, a condition dual to~(\ref{eq: K}): since $\pos(\necv\!P \cdot \necv\!Q) \subseteq \pos\!\necv\!P \cdot \pos\!\necv\!Q$, $\pos\!\necv P \subseteq P$, and $\pos\!\necv\!Q \subseteq Q$, we obtain $\pos(\necv\!P \cdot \necv\!Q) \subseteq P \cdot Q$; hence, by~(\ref{eq: R1}), $\necv\!P \cdot \necv\!Q \subseteq \necv(P \cdot Q)$.

Most importantly, residuation, together with the closure and interior properties of our operators, implies that the fixed points of~$\pos$ coincide with those of~$\necv$. This not only entails that~$\necv$ preserves~$\{\varepsilon\}$, but also allows for an alternative characterisation of safety~f-properties. This is what we are about to discuss.

Manna and Pnueli \cite{MannaPnueli90} showed that, in the context of infinite-trace systems, safety properties can be specified by appropriately combining past and future operators in temporal logic. In what follows, we denote by $\mathsf{LTLB}$~\cite{Emerson1990} the standard extension of classical $\mathsf{LTL}$ (with \emph{next} and \emph{until}) with the past connectives \emph{previous} and \emph{since}. Two $\mathsf{LTLB}$-formulas $\varphi,\psi$ are \emph{initially equivalent} if $\alpha, 0 \Vdash \varphi$ iff $\alpha, 0 \Vdash \psi$, for all $\alpha\in \Sigma_{\mathcal{T}}^{\omega}$. It is well known that $\mathsf{LTL}$ and $\mathsf{LTLB}$ are equally expressive with respect to initial equivalence~\cite{gpss}.

For an LTS~$\mathcal{T}$, let $\mathtt{A}\colon\wp(\Sigma_{\mathcal{T}}^*) \rightarrow \wp(\Sigma_{\mathcal{T}}^{\omega})$ be defined by $Q \mapsto \{ \alpha \in \Sigma_{\mathcal{T}}^{\omega} \mid \mathit{pref}(\alpha) \subseteq Q \}$. This operator is used in \cite{MannaPnueli90} to define safety for infinite traces: a set $P \subseteq \Sigma_{\mathcal{T}}^{\omega}$ is a safety property if and only if $P = \mathtt{A}(Q)$ for some $Q \subseteq \Sigma_{\mathcal{T}}^*$. Let now $\varphi$ be an $\mathsf{LTLB}$-formula. We define $\mathrm{Sat}(\varphi) := \{ \alpha\in \Sigma_{\mathcal{T}}^{\omega} \mid \alpha \Vdash \varphi \}$, that is, the property corresponding to~$\varphi$.

We say that $\varphi$ is a \emph{past formula} if it is either an atomic proposition or if it contains only past operators. A finite trace $w \in \Sigma_{\mathcal{T}}^*$ \emph{end-satisfies} a past formula $\varphi$ (written $w \Vdash_{\mathrm{e}} \varphi$) if there exists $\alpha \in \Sigma_{\mathcal{T}}^{\omega}$ such that $w \sqsubseteq \alpha$ and $\alpha, |w| - 1 \Vdash \varphi$. In other words, $w \Vdash_{\mathrm{e}} \varphi$ if there exists an infinite extension~$\alpha$ of~$w$ in which the satisfaction of~$\varphi$ depends on its truth at the last position of~$w$, and therefore, since $\varphi$ is a “backward looking” formula, on the whole~$w$. We write $\mathrm{eSat}(\varphi) := \{ w \in \Sigma_{\mathcal{T}}^* \mid w \Vdash_{\mathrm{e}} \varphi \}$ for the f-property given by all finite traces end-satisfying $\varphi$. An f-property $Q \subseteq \Sigma_{\mathcal{T}}^*$ is \emph{expressible} in $\mathsf{LTLB}$ if $Q=\mathrm{eSat}(\varphi)$, for some past formula $\varphi$. Observe that safety properties~$\mathtt{A}(\mathrm{eSat}(\varphi))$ coincide with the properties~$\mathrm{Sat}(G\varphi)$, where~$G$ is the \emph{globally} operator \cite[p. 392]{MannaPnueli90}. Indeed, we have:
\begin{align*}
	\alpha \in \mathrm{Sat}(G\varphi) &\Leftrightarrow \alpha \Vdash G\varphi \\
	&\Leftrightarrow \text{$w \Vdash_{\mathrm{e}} \varphi$ for all $w \in \mathit{pref}(\alpha)$}\\
	&\Leftrightarrow \mathit{pref}(\alpha) \subseteq \mathrm{eSat}(\varphi)\\
	&\Leftrightarrow \alpha \in \mathtt{A}(\mathrm{eSat}(\varphi)).
\end{align*}
Formulas of the form~$G\varphi$ are called \emph{safety formulas}. It is therefore clear that every safety formula specifies a safety property expressible in~$\mathsf{LTL}$ by means of a formula~$\psi$ that is initially equivalent to~$G\varphi$.
\begin{rexample}
	Let us consider the safety property~$P_1$ introduced in the previous section. We have seen that, for infinite traces, it can be specified by the $\mathsf{LTL}$ formula $\neg (\neg\mathit{request}\mathrel{U} (\mathit{response}\wedge \neg \mathit{request}))$, which is initially equivalent to the $\mathsf{LTLB}$ formula $G(\mathit{request} \rightarrow P(\mathit{response}))$, where~$P$ is the \emph{once} operator (see \cite{Schnoebelen2003}).
\end{rexample}
We show how the above outlined framework readily adapts to the case of finite traces. Let $\mathsf{PLTL}_f$ be the \emph{pure-past} counterpart of ~$\mathsf{LTL}_f$. It is well-known that these two logics are expressively equivalent~\cite{DDFR2021}. For a $\mathsf{PLTL}_f$-formula $\varphi$, we denote by $\varphi_{\triangleright}$ an~$\mathsf{LTL}_f$-formula that is equivalent to~$\varphi$. Then, in~$\mathsf{LTL}_f$, we have:
\begin{align*}
	w \in \mathrm{Sat}(G\varphi_{\triangleright})&\Leftrightarrow \text{$v \Vdash \varphi$ for all $v \in \mathit{pref}(w)$} \\
	&\Leftrightarrow \mathit{pref}(w) \subseteq \mathrm{Sat}(\varphi_{\triangleright})\\
	&\Leftrightarrow w \in \necv\!\mathrm{Sat}(\varphi_{\triangleright}).
\end{align*}
Clearly, $G\varphi_{\triangleright}$ specifies a safety property. Indeed, $\mathrm{Sat}(G\varphi_{\triangleright}) = \necv\!\mathrm{Sat}(\varphi_{\triangleright}) = \pos\!\necv\!\mathrm{Sat}(\varphi_{\triangleright})$. This allows us to provide an alternative, residuation-theoretic proof of Prop. \ref{prop: sbox}:
\begin{align*}
    \text{$P\subseteq \Sigma_{\mathcal{T}}^*$ is a safety property} & \Leftrightarrow P = \pos \! P\\
    &\Leftrightarrow P = \necv\! P\\
    & \Leftrightarrow P = \{ w \in \Sigma_{\mathcal{T}}^* \mid \forall v (v \sqsubseteq w\Rightarrow v \in P)\}\\
    & \Leftrightarrow P = \{ w \in \Sigma_{\mathcal{T}}^* \mid \mathit{pref}(w)\subseteq P\}
\end{align*}
We conclude that the residuation relation between $\pos$ and $\necv$ underlies the finite-trace variant of the Manna-Pnueli method for specifying safety properties by means of past formulas. To the best of our knowledge, this connection between residuation and the formal specification of properties has not been explicitly recognised in the model-checking and temporal-logic communities, and it deserves further investigation.
\subsection{\texorpdfstring{Closure $\ell$-monoids}{Closure l-monoids}}
In this subsection, we introduce and study the variety of closure $\ell$-monoids, conceived as an algebraic abstraction—amenable to a characterisation in structural proof-theoretic terms—of f-properties specified by means of $\pos$ and $\necv$.
\begin{definition}\label{def: clm}
A \emph{closure $\ell$-monoid} is an algebra $\mathbf{A}=\langle A, \cdot, \wedge,\vee, \pos, \necv , 1, \bot,\top\rangle$ such that $\langle A, \wedge,\vee,\bot,\top\rangle$ is a bounded distributive lattice, $\langle A, \cdot,\wedge,\vee\rangle$ is an $\ell$-monoid, $\bot$ is an absorbing element for monoid multiplication, and $\pos,\necv$ are unary operations defined by the following identities:
\begin{enumerate}
    \item $x \leq \pos x$
    \item $\pos\!\pos x \approx \pos x$
    \item $\pos (x \vee y) \approx \pos x \vee \pos y$
    \item $\necv (x \wedge y) \approx \necv x \wedge \necv y$
    \item $\pos(x \cdot y)\leq \pos\!x \cdot \pos\!y$
    \item $\pos\!\necv x \leq x$
    \item $x \leq \necv\! \pos x$
\end{enumerate}
\noindent We denote by $\mathfrak{LMC}$ the variety of closure $\ell$-monoids.
\end{definition}
\begin{example}
The algebra $\wp(\mathbf{F}_{\Sigma})_+$ is a closure $\ell$-monoid. In general, let $\langle \mathbf{M}, \preceq\rangle$ be a monoid endowed with a preorder satisfying the RDP. Then 
\[\wp(\mathbf{M})_+=\langle \wp(M), \cdot, \cap, \cup, \pos_{\preceq}, \necv_{\preceq}, \{1\}, \emptyset, M \rangle \]
is a closure $\ell$-monoid. It follows that, for $\mathfrak{d}\in\{l,r\}$, any algebra $\wp(\mathbf{M})_+$ arising from $\langle \mathbf{M}, \mid_{\mathfrak{d}}\rangle$ is in $\mathfrak{LMC}$. 
\end{example}
Recall that an interior operator $\iota$ on an $\ell$-monoid $\mathbf{M}$ is a \emph{co-nucleus} \cite[§3.14.16]{yellowbook} if the following conditions are satisfied, for all $a,b \in M$:
\begin{align*}
&\tag{CN$_1$} \iota(\iota(a)\cdot\iota(b))=\iota(a)\cdot \iota(b) \label{eq: CN1}\\
&\tag{CN$_2$} \iota(a)\cdot\iota(1)=\iota(1)\cdot\iota(a)=\iota(a)\label{eq: CN2}
\end{align*}
If $\iota$ satisfies just (\ref{eq: CN1}), then it is called a \emph{weak co-nucleus} \cite{BGM2022}. Note that, in $\wp(\mathbf{F}_{\Sigma})_+$, $\necv$ satisfies both (\ref{eq: CN1}) and (\ref{eq: CN2}).
\begin{lemma}\label{lem: dereq1}
In any closure $\ell$-monoid, $\langle \pos, \necv\rangle$ is a residuated pair where $\pos$ is a topological closure operator and $\necv$ is a weak co-nucleus\footnote{We thank Simon Santschi for pointing out the weak co-nuclearity of $\necv$.}. Moreover, the two operators preserve the bounds.
\end{lemma}
\begin{proof}
\normalfont Let $\mathbf{A}\in \mathfrak{LMC}$. By axioms (3) and (4), both $\pos$ and $\necv$ are order-preserving. Together with (6) and (7), this implies that $\langle \pos, \necv\rangle$ is a residuated pair (Prop.~\ref{prop: eqres}). Note that $\pos\!\bot = \bot$ by Prop.~\ref{prop: ncres}(4). It then follows from axioms~(1),~(2), and~(3) that $\pos$ is a topological closure operator. Moreover, since $\top \leq \pos \!\top$ by~(1), $\pos$ is bound-preserving.

We now prove that $\necv$ is a weak co-nucleus. Let $a,b\in A$. First, observe that from (1) and (6) we obtain $\necv\! a \leq \pos\! \necv\! a \leq a$, and therefore, by transitivity, $\necv\! a \leq a$. Hence $\necv$ is contractive, and in particular $\necv\! \necv\!a \leq \necv\! a$. Next, since (6) and (2) yield $\pos\!\pos\! \necv\! a \leq a$, repeated application of (\ref{eq: R1}) gives $\necv\! a \leq \necv\!\necv\! a$. Thus $\necv$ is idempotent. Since it is also meet-preserving, we conclude that $\necv$ is an interior operator. It is in fact a topological, bound-preserving one, as $\necv\!\top=\top$ by Prop.~\ref{prop: ncres}(4) and $\necv\!\bot\leq \bot$ by contractiveness. It remains to verify that $\necv\! a \cdot \necv\! b = \necv(\necv\! a \cdot \necv \!b)$. By (5), $\pos(\necv\! a \cdot \necv\! b) \leq \pos\! \necv\! a \cdot \pos\! \necv\! b$. Note that $\pos\! \necv\! a = \necv\! a$: the left-to-right inequality follows from (1), while the converse is obtained from $\necv\! a \leq \necv\! \necv\! a$ using (\ref{eq: R1}). Hence $\pos(\necv\! a \cdot \necv\! b) \leq \pos\! \necv\! a \cdot \pos\! \necv\! b$ can be rewritten as $\pos(\necv\! a \cdot \necv\! b) \leq \necv\! a \cdot \necv\! b$. Applying (\ref{eq: R1}), we conclude that $\necv\! a \cdot \necv\! b \leq \necv(\necv\! a \cdot \necv\! b)$. Since the converse inequality holds by contractiveness, we obtain $\necv(\necv\! a \cdot \necv\! b) = \necv\! a \cdot \necv\! b$.
\end{proof}
In the proof of the preceding lemma we derived some important laws from the axioms in Definition \ref{def: clm}. We report them below.
\begin{multicols}{2}
\begin{enumerate}\setlength{\itemindent}{0.5cm}
\item[(i)] $\necv\!x \leq x$;
    \item[(ii)] $\necv\!\necv\! x \approx \necv\! x$;
    \item[(iii)] $\necv(\necv\! x \cdot \necv\! y) \approx \necv\! x \cdot \necv\! y$;
    \item[(iv)] $\pos\! \necv\! x \approx \necv\! x$
    \item[(v)] $\pos\!\bot \approx \bot$;
    \item[(vi)] $\pos\!\top \approx \top$;
    \item[(vii)] $\necv\!\bot \approx \bot$;
    \item[(viii)] $\necv\!\top \approx \top$;
\end{enumerate}
\end{multicols}
In the next result we extend the above list with further noteworthy equations.
\begin{lemma}\label{lem: dereq2}
The following identities hold in any closure $\ell$-monoid:
\begin{multicols}{2}
\begin{enumerate}\setlength{\itemindent}{0.5cm}
\item[\textup{(ix)}] $\pos(x \wedge y)\leq\pos\!x \wedge \pos\! y$;
\item[\textup{(x)}] $\necv\!x \vee \necv\! y \leq \necv(x \vee y)$;
\item[\textup{(xi)}] $\necv\!\pos\!x \approx \pos\! x$;
    \item[\textup{(xii)}] $\pos(\pos\! x \cdot \pos\! y) \approx \pos\! x \cdot \pos\!y$;
    \item[\textup{(xiii)}] $\pos(\pos\!x \wedge \pos\!y) \approx \pos\!x \wedge \pos \!y $;
    \item[\textup{(xiv)}] $\necv(\necv\!x\vee \necv\! y) \approx \necv\! x \vee \necv\! y$.
    \item[\textup{(xv)}] $\necv \!x \cdot \necv\!y \leq \necv(x\cdot y)$.
\end{enumerate}
\end{multicols}
Moreover, $x \leq \pos\! y$ is equivalent to $\pos\! x \leq \pos \!y$.
\end{lemma}
\begin{proof} Again, let $\mathbf{A}\in \mathfrak{LMC}$ and let $a,b\in A$.
\begin{itemize}\setlength{\itemindent}{0.3cm}
    \item[(ix)] Since $\pos$ is order-preserving, we have $\pos(x \wedge y) \leq \pos\!x$ and $\pos(x \wedge y) \leq \pos\!y$. The desired inequality then follows by lattice-theoretic properties.
    \item[(x)] Dual argument to (ix).
    \item[(xi)] By (i), we have $\necv\!\pos\!a \leq \pos\! a$. Moreover, by (2), $\pos\!\pos\!a \leq \pos\!a$, and hence by (\ref{eq: R1}) we conclude that $\pos\! a \leq \necv\!\pos\!a$.
    \item[(xii)] By (1), we have $\pos\!a \cdot \pos\!b \leq \pos(\pos\!a \cdot \pos\!b)$. As for the converse, (5) gives $\pos(\pos\!a \cdot \pos\!b) \leq \pos\!\pos\!a \cdot \pos\!\pos\!b$, and an application of (2) then yields the desired inequality.
    \item[(xiii)] Analogous argument to (xii), using (ix) in place of (5).
    \item[(xiv)] By (i), we have $\necv(\necv\!a \vee \necv\!b) \leq \necv\!a \vee \necv\!b$. 
Now, since $\necv$ is order-preserving, $\necv\!\necv\!a \leq \necv(\necv\!a \vee \necv\!b)$ and $\necv\!\necv\!b \leq \necv(\necv\!a \vee \necv\!b)$. 
It follows, by (ii) and lattice-theoretic properties, that the right-to-left inequality holds. 
    \item[(xv)] By (5) and (1), we have 
$\necv\!x \cdot \necv\!y \leq \pos(\necv\!x \cdot \necv\!y) \leq \pos\!\necv\!x \cdot \pos\!\necv\!y$. 
Hence, the desired inequality follows by (iv) and transitivity.
\end{itemize}
Finally, suppose that $a \leq \pos\!b$. 
Then $\pos\!a \leq \pos\!\pos\!b$, and hence, by (2), $\pos\!a \leq \pos\!b$. 
Conversely, if $\pos\!a \leq \pos\!b$, then, by (1) and transitivity, it follows that $a \leq \pos\!b$.
\end{proof}
\begin{lemma}\label{lem: VnGEN}
$\mathfrak{LMC}\neq \mathbb{V}(\wp(\mathbf{F}_{\Sigma})_+)$.
\begin{proof}
We aim to show that there exists an equation satisfied by~$\wp(\mathbf{F}_{\Sigma})_+$ but not by~$\mathfrak{LMC}$. Our candidate is $\pos(x \cdot y) \leq \pos\!x \,\vee\, (x \cdot \pos\!y)$, which abstracts the left-to-right direction of~(\ref{eq: C}). Let $\mathbf{End}_A = \langle \mathrm{End}_A, \circ, \mathrm{id}_A \rangle$ be the monoid of self-maps on a set~$A$. Consider the closure $\ell$-monoid $\wp(\mathbf{End}_A)_+ = \langle \wp(\mathrm{End}_A), \circ, \cap, \cup, \pos, \necv, \{\mathrm{id}_A\}, \emptyset, \mathrm{End}_A \rangle$, where $\pos = \pos_{\mid_l}$, $\necv = \necv_{\mid_l}$, and for all $F,G \subseteq \mathrm{End}_A$, $F \circ G := \{ f \circ g \mid f \in F,\ g \in G \}$.
	
We begin by observing that, for $f,g \in \mathrm{End}_A$, $f \mid_l g$ iff there exists $f'\in \mathrm{End}_A$ such that $g=f \circ f'$, i.e., the diagram in Figure \ref{fig: diag1} is commutative.
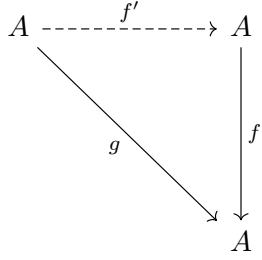
\begin{figure}[H]
	\centering
\begin{tikzcd}[row sep=6em, column sep=6em]
	A \arrow[r, "f'",dashed]\arrow[dr, "g",swap] & A\arrow[d, "f"]\\
	 & A
\end{tikzcd}
\caption{Left divisibility in $\mathbf{End}_A$.}
\label{fig: diag1}
\end{figure}
We now establish that $\wp(\mathbf{End}_A)_+\not\models \pos(x \cdot y)\leq \pos\!x \vee (x \cdot \pos\!y)$. If the opposite were true, then, in particular $\pos(\{f\circ g\})\subseteq \pos\{f\} \cup (\{f\} \circ \pos\{g\})$, for all $f,g\in \mathrm{End}_A$\footnote{Note that $\pos(\{f\}\circ\{g\})=\pos(\{f\circ g\})$.}. We show that, for a suitable choice of $A$ and of $f,g \in \mathrm{End}_A$, there exists $k \in \mathrm{End}_A$ such that $k \mid_l f \circ g$ and both of the following conditions hold:
\begin{enumerate}
	\item $k \nmid_l f$;
	\item for all $m \in \mathrm{End}_A$ such that $m \mid_l g$, it holds that $k \neq f \circ m$.
\end{enumerate}
Let $A = \mathbb{Z}$, and let $f, g \in \mathrm{End}_{\mathbb{Z}}$ be defined as follows, for $n \in \mathbb{Z}$:
\[
f(n):=
\begin{cases}
1 \hfill & \text{if $n \leq 0$,}\\
2 \hfill &\text{if $n>0$}
\end{cases}\qquad\qquad
g(n):=-1
\]
Then $f \circ g$ is a constant map, as $f(g(n))=1$ for all $n \in \mathbb{Z}$. We now seek a left divisor $k$ of $f \circ g$ (with respect to~$\circ$) such that there is no $h \in \mathrm{End}_{\mathbb{Z}}$ satisfying $k \circ h = f$. Define
\[
k(n):=
\begin{cases}
0 \hfill & \text{if $n \leq 0$,}\\
1 \hfill &\text{if $n>0$}
\end{cases}
\]
and let $k':n \mapsto 2$. Clearly, for all $n \in \mathbb{Z}$, $k(k'(n)) = 1 = f(g(n))$ and therefore $k\mid_l f \circ g$. The functions $f$ and $k$ have, respectively, ranges $\{1,2\}$ and $\{0,1\}$, hence $k \circ h \neq f$ for all $h \in \mathrm{End}_{\mathbb{Z}}$. This establishes~(1). The proof of (2) is immediate: as $f$ takes only the values $1$ and $2$, it is clear that $k \neq f \circ m$ for all $m \in \mathrm{End}_{\mathbb{Z}}$, and hence (2) holds.
\end{proof}
\end{lemma}
\section{\texorpdfstring{The Gentzen-style System $\mathsf{LMC}$}{The Gentzen-style System LMC}}\label{sect: LMC}
In this section, we introduce the Gentzen-style calculus~$\mathsf{LMC}$ for closure~$\ell$-monoids. Let $\nu_1$ be the similarity type $\langle2,2,2,1,1,0,0,0\rangle$. The set $\mathit{Fm}$ of \emph{$\mathsf{LMC}$-formulas} is the set of $\nu_1$-terms generated over a denumerable set $X$ of variables according to the BNF:
\[\varphi \,\, ::= \,\, x \:| \: 1 \:| \: \bot \:| \: \top \:| \: \psi_1 \cdot \psi_2 \: | \: \psi_1 \wedge \psi_2 \: | \: \psi_1\vee \psi_2 \: | \:\! \pos\!\psi \: | \:\! \necv\! \psi\]
where $x \in X$ and $\psi,\psi_1,\psi_2 \in \mathit{Fm}$.
Set $\mathbf{Fm}:=\mathbf{T}_{\nu}(X)$. We introduce the following measure of the complexity of formulas. Let $\mathrm{cp}\colon\mathit{Fm}\rightarrow \mathbb{N}$ be the map defined as follows:
\[
\begin{array}{ll}
   \mathrm{cp}(\varphi)=0  & \;\;\;\text{if $\varphi \in X\cup \{ 1,\bot, \top \}$} \\
   \mathrm{cp}(\varphi \divideontimes \psi) = \mathrm{cp}(\varphi)+\mathrm{cp}(\psi)+1  & \;\;\;\text{for $\divideontimes \in \{ \cdot, \wedge,\vee \}$}\\
   \mathrm{cp}(\deo\! \varphi) = \mathrm{cp}(\varphi)+1  & \;\;\;\text{for $\deo \in \{ \pos, \necv \}$} \\
\end{array}
\] 
Let $\nu_2$ be the similarity type $\langle2,2,1,0\rangle$. The set $\mathit{Struct}$ of \emph{structural terms} for $\mathsf{LMC}$ is the set of $\nu_2$-terms generated over $\mathit{Fm}$ according to the BNF:
\[\Gamma \,\, ::= \,\, \varphi \:| \: \epsilon \:|\: \Delta_1 \circ \Delta_2 \:| \: \Delta_1 \sqcap \Delta_2 \:| \: \langle \Delta \rangle\]
where $\varphi \in \mathit{Fm}$, $\Delta,\Delta_1,\Delta_2\in\mathit{Struct}$, $\epsilon$ is a structural constant, $\circ$ and $\sqcap$ are binary structural operators, and $\langle - \rangle$ is a unary structural operator.  In particular, $\langle\circ,\epsilon \rangle$ and $\sqcap$ correspond, respectively, to monoid- and meet-semilattice-like term grouping, while $\langle - \rangle$ is intended to express $\pos$-like term modalisation. We define $\mathbf{Struct}:=\mathbf{T}_{\nu_2}(\mathit{Fm})$.
\begin{remark}
It is important to observe that the syntax of structural terms is not two-layered, in the sense that the construction of a formula does not constitute a special case of constructing a structural term. From an algebraic perspective, each $\Gamma\in\mathit{Struct}$ is just a polynomial symbol generated over $\mathit{Fm}$, and therefore $\mathsf{LMC}$-formulas (of any complexity) should be regarded as \textit{atomic structural terms}. This means that, e.g., $x$ is a subformula (but \textit{not} a structural subterm) of $x \wedge y$.
\end{remark}
\begin{table}[ht!]
\small
	\begin{framed}
            \begin{minipage}{11.6cm}
            \begin{center}
                \underline{\textbf{Axioms and structural rules}}\vspace{0.2cm}
            \end{center}
            \end{minipage}
\[
\begin{bprooftree}
	\def\fCenter{\mbox{\ $\Yright$\ }}
	\RightLabel{$\mathsf{init}$}
	\AxiomC{\phantom{$\Delta \Yright \varphi$}}
	\UnaryInfC{$\varphi \Yright \varphi$}
\end{bprooftree}
\begin{bprooftree}
	\RightLabel{\textsf{cut}}
	\AxiomC{$\Delta \Yright \varphi$}
	\AxiomC{$\Gamma\{ \varphi \}\Yright \psi $}
	\BinaryInfC{$\Gamma \{\Delta\}\Yright \psi$}
\end{bprooftree}
\begin{bprooftree}
	\def\fCenter{\mbox{\ $\Yright$\ }}
	\RightLabel{$\circ$\textsf{A}}
	\AxiomC{$\Gamma\lbrace((\Delta_1 \circ \Delta_2)\circ\Delta_3)\rbrace \Yright \varphi$}
	\doubleLine
	\UnaryInfC{$\Gamma\lbrace(\Delta_1 \circ(\Delta_2\circ \Delta_3))\rbrace \Yright \varphi$}
\end{bprooftree}  
\]\vspace{0.1cm}
\[
\begin{bprooftree}
	\def\fCenter{\mbox{\ $\Yright$\ }}
	\RightLabel{$\circ\, \epsilon$}
	\AxiomC{$\Gamma\{\Delta\}\Yright \varphi$}
	\UnaryInfC{$\Gamma\{\Delta \circ \epsilon \}\Yright \varphi$}
\end{bprooftree}
\begin{bprooftree}
	\def\fCenter{\mbox{\ $\Yright$\ }}
	\RightLabel{$\epsilon\,\circ$}
	\AxiomC{$\Gamma\{\Delta\}\Yright \varphi$}
	\UnaryInfC{$\Gamma\{\epsilon \circ \Delta \}\Yright \varphi$}
\end{bprooftree} 
\begin{bprooftree}
	\def\fCenter{\mbox{\ $\Yright$\ }}
	\RightLabel{$\sqcap$\textsf{A}}
	\AxiomC{$\Gamma\lbrace((\Delta_1 \sqcap \Delta_2)\sqcap\Delta_3)\rbrace \Yright \varphi$}
	\doubleLine
	\UnaryInfC{$\Gamma\lbrace(\Delta_1 \sqcap(\Delta_2\sqcap\Delta_3))\rbrace \Yright \varphi$}
\end{bprooftree} 
\]\vspace{0.1cm}
\[
\begin{bprooftree}
    \def\fCenter{\mbox{\ $\Yright$\ }}
	\RightLabel{$\sqcap\mathsf{W}_1$}
	\AxiomC{$\Gamma\{\Delta_1\}\Yright \varphi$}
	\UnaryInfC{$\Gamma\{\Delta_1 \sqcap \Delta_2 \}\Yright \varphi$}
\end{bprooftree}
\begin{bprooftree}
	\def\fCenter{\mbox{\ $\Yright$\ }}
	\RightLabel{$\sqcap$\textsf{E}}
	\AxiomC{$\Gamma\{\Delta_1 \sqcap \Delta_2\}\Yright \varphi$}
	\UnaryInfC{$\Gamma\{\Delta_2 \sqcap \Delta_1\}\Yright \varphi$}
\end{bprooftree}
\begin{bprooftree}
	\def\fCenter{\mbox{\ $\Yright$\ }}
	\RightLabel{$\sqcap$\textsf{C}}
	\AxiomC{$\Gamma\{ \Delta \sqcap \Delta \} \Yright \varphi$}
	\UnaryInfC{$\Gamma\{ \Delta\} \Yright \varphi$}
\end{bprooftree}
\]\vspace{0.1cm}
\[
\begin{bprooftree}
    \def\fCenter{\mbox{\ $\Yright$\ }}
    \RightLabel{\textsf{K}}
	\AxiomC{$\Gamma \{ \langle\Delta_1\rangle \circ \langle\Delta_2\rangle \} \Yright \varphi$}
	\UnaryInfC{$\Gamma \{ \langle\Delta_1 \circ \Delta_2\rangle \} \Yright \varphi$}
\end{bprooftree}
\begin{bprooftree}
    \def\fCenter{\mbox{\ $\Yright$\ }}
    \RightLabel{\textsf{T}}
    \AxiomC{$\Gamma\{\langle \Delta\rangle \} \Yright \varphi$}
    \UnaryInfC{$\Gamma\{ \Delta\} \Yright \varphi$}
\end{bprooftree}
\begin{bprooftree}
    \def\fCenter{\mbox{\ $\Yright$\ }}
    \RightLabel{\textsf{4}}
    \AxiomC{$\Gamma\{\langle \Delta\rangle \} \Yright \varphi$}
    \UnaryInfC{$\Gamma\{ \langle\langle\Delta \rangle\rangle\} \Yright \varphi$}
\end{bprooftree}\vspace{0.3cm}
\]
\begin{minipage}{11.6cm}
\begin{center}
    \underline{\textbf{Logical rules}}\vspace{0.2cm}
\end{center}
\end{minipage}
\[
\begin{bprooftree}
	\def\fCenter{\mbox{\ $\Yright$\ }}
	\RightLabel{$\cdot\,$\textsf{L}}
	\AxiomC{$\Gamma\{ \varphi \circ \psi\} \Yright \chi$}
	\UnaryInfC{$\Gamma\{ \varphi \cdot \psi\} \Yright \chi$}
\end{bprooftree}
\begin{bprooftree}
	\def\fCenter{\mbox{\ $\Yright$\ }}
	\RightLabel{$\cdot\,$\textsf{R}}
	\AxiomC{$\Gamma_1 \Yright \varphi$}
	\AxiomC{$\Gamma_2 \Yright \psi$}
    \BinaryInfC{$\Gamma_1 \circ \Gamma_2 \Yright \varphi \cdot \psi$}
\end{bprooftree}
\begin{bprooftree}
	\def\fCenter{\mbox{\ $\Yright$\ }}
	\RightLabel{$\wedge$\textsf{L}}
	\AxiomC{$\Gamma\{\varphi \sqcap \psi\}\Yright \chi$}
	\UnaryInfC{$\Gamma\{\varphi \wedge \psi\}\Yright \chi$}
\end{bprooftree}
\begin{bprooftree}
	\def\fCenter{\mbox{\ $\Yright$\ }}
	\RightLabel{$\wedge$\textsf{R}}
	\AxiomC{$\Gamma \Yright \varphi$}
    \AxiomC{$\Gamma \Yright \psi$}
	\BinaryInfC{$\Gamma\Yright \varphi \wedge \psi$}
\end{bprooftree}  
\]  
\vspace{0.1cm}
\[
\begin{bprooftree}
	\def\fCenter{\mbox{\ $\Yright$\ }}
	\RightLabel{$\vee$\textsf{L}}
	\AxiomC{$\Gamma\{\varphi\} \Yright \chi$}
    \AxiomC{$\Gamma\{\psi\} \Yright \chi$}
	\BinaryInfC{$\Gamma\{\varphi \vee \psi\} \Yright \chi$}
\end{bprooftree}
\begin{bprooftree}
	\def\fCenter{\mbox{\ $\Yright$\ }}
	\RightLabel{$\vee$\textsf{R}$_1$}
	\AxiomC{$\Gamma \Yright \varphi$}
	\UnaryInfC{$\Gamma \Yright \varphi \vee \psi$}
\end{bprooftree}
\begin{bprooftree}
    \def\fCenter{\mbox{\ $\Yright$\ }}
	\RightLabel{$\vee$\textsf{R}$_2$}
	\AxiomC{$\Gamma \Yright \psi$}
	\UnaryInfC{$\Gamma \Yright \varphi \vee \psi$}
\end{bprooftree}
\]\vspace{0.1cm}
\[
\begin{bprooftree}
	\def\fCenter{\mbox{\ $\Yright$\ }}
    \RightLabel{$\pos$\textsf{L}}
	\AxiomC{$\Gamma \{ \langle \varphi \rangle\} \Yright \psi$}
	\UnaryInfC{$\Gamma \{ \pos\!\varphi\} \Yright \psi$}
\end{bprooftree}
\begin{bprooftree}
	\def\fCenter{\mbox{\ $\Yright$\ }}
	\RightLabel{$\pos$\textsf{R}}
	\AxiomC{$\Gamma \Yright \varphi$}
	\UnaryInfC{$\langle \Gamma\rangle \Yright \pos\!\varphi$}
\end{bprooftree}
\begin{bprooftree}
	\def\fCenter{\mbox{\ $\Yright$\ }}
	\RightLabel{$\necv$\textsf{L}}
	\AxiomC{$\Gamma \{ \varphi\} \Yright \psi$}
	\UnaryInfC{$\Gamma \{ \langle\necv\!\varphi\rangle \} \Yright \psi$}
\end{bprooftree}
\begin{bprooftree}
	\def\fCenter{\mbox{\ $\Yright$\ }}
	\RightLabel{$\necv$\textsf{R}}
	\AxiomC{$\langle \Gamma \rangle \Yright \psi$}
	\UnaryInfC{$\Gamma \Yright \necv\!\psi$}
\end{bprooftree}
\]\vspace{0.1cm}
\[
\begin{bprooftree}
	\def\fCenter{\mbox{\ $\Yright$\ }}
	\RightLabel{$1$\textsf{L}}
	\AxiomC{$\Gamma\{\epsilon \} \Yright \varphi$}
	\UnaryInfC{$\Gamma\{1 \} \Yright \varphi$}
\end{bprooftree}
\begin{bprooftree}
    \def\fCenter{\mbox{\ $\Yright$\ }}
	\RightLabel{$1\mathsf{R}$}
	\AxiomC{\phantom{$\Gamma\{\}$}}
	\UnaryInfC{$ \phantom{\{}\!\!\!\epsilon \Yright 1$}
\end{bprooftree}
\begin{bprooftree}
	\def\fCenter{\mbox{\ $\Yright$\ }}
	\RightLabel{$\cdot1$}
	\AxiomC{$\Gamma \Yright \varphi$}
	\UnaryInfC{$\Gamma\Yright \varphi \cdot 1$}
\end{bprooftree}
\begin{bprooftree}
	\def\fCenter{\mbox{\ $\Yright$\ }}
	\RightLabel{$1\cdot$}
	\AxiomC{$\Gamma \Yright \varphi$}
	\UnaryInfC{$\Gamma\Yright 1 \cdot \varphi$}
\end{bprooftree}
\begin{bprooftree}
    \def\fCenter{\mbox{\ $\Yright$\ }}
	\RightLabel{$\bot\mathsf{L}$}
	\AxiomC{\phantom{$\Gamma\{\}$}}
    \UnaryInfC{$ \Gamma\{ \bot\} \Yright \varphi$}
\end{bprooftree}
 \begin{bprooftree}
	\def\fCenter{\mbox{\ $\Yright$\ }}
	\RightLabel{$\top\mathsf{R}$}
	\AxiomC{\phantom{$\Gamma\{\}$}}
	\UnaryInfC{$ \Gamma \Yright \top$}
\end{bprooftree}
\]
		\end{framed}
    \caption{The Gentzen-style system $\mathsf{LMC}$.}
    \label{tab: LMC}
	\end{table}
Complexity of structural terms is measured by the map $\mathrm{cp}_s:\mathit{Struct}\rightarrow \mathbb{N}$ defined by:
\[
\begin{array}{ll}
   \mathrm{cp}_s(\Gamma)=0  & \;\;\;\text{if $\Gamma \in \mathit{Fm}\cup \{ \epsilon\}$} \\
   \mathrm{cp}_s(\Gamma \divideontimes \Delta) = \mathrm{cp}_s(\Gamma)+\mathrm{cp}_s(\Delta)+1  &  \;\;\;\text{for $\divideontimes \in \{ \circ, \sqcap \}$}\\
   \mathrm{cp}_s(\langle\Gamma\rangle) = \mathrm{cp}_s(\Gamma)+1  &  \\
\end{array}
\]
An \textit{$\mathsf{LMC}$-sequent} is a pair $\langle \Gamma, \varphi\rangle$, denoted $\Gamma \Yright \varphi$, where $\varphi\in \mathit{Fm}$ and $\Gamma$ is a structural term. Axioms and inference rules of $\mathsf{LMC}$ are provided in Table \ref{tab: LMC} where:
\begin{itemize}
    \item For $\Gamma,\Delta \in \mathit{Struct}$, we write $\Gamma\lbrace\Delta\rbrace$ to indicate a distinguished occurrence of $\Delta$ as a structural subterm of $\Gamma$. We also stipulate that, in any structural term containing multiple occurrences of expressions of the form $\Gamma\lbrace\Delta\rbrace$, each instance refers to the same specific occurrence of $\Delta$ within $\Gamma$.
    \item A rule of the form{\def\fCenter{\mbox{\ $\Yright$\ }}\AxiomC{$\Gamma\Yright \varphi$}\doubleLine\UnaryInfC{$\Delta\Yright \psi$}\DisplayProof}is shorthand for the two rules{\def\fCenter{\mbox{\ $\Yright$\ }}\AxiomC{$\Gamma\Yright \varphi$}\UnaryInfC{$\Delta\Yright \psi$}\DisplayProof}and{\def\fCenter{\mbox{\ $\Yright$\ }}\AxiomC{$\Delta\Yright \psi$}\UnaryInfC{$\Gamma\Yright \varphi$}\DisplayProof}.
\end{itemize}
In the subsequent discussion, we will often need to easily switch between the proof-theoretic language of $\mathsf{LMC}$ the and the algebraic language of closure $\ell$-monoids, and \textit{vice versa}. By the \textit{formula translation} of a structural term we mean a map $^{\natural}\colon\mathit{Struct}\to\mathit{Fm}$ recursively defined as follows:
\[
\begin{array}{ll}
   \Gamma^{\natural}=\Gamma \;\text{(if $\Gamma \in \mathit{Fm}$)} &  \\
   (\Gamma \circ \Delta)^{\natural}=\Gamma^{\natural}\cdot \Delta^{\natural} & (\Gamma \sqcap \Delta)^{\natural}=\Gamma^{\natural}\wedge \Delta^{\natural}\\
   \langle\Gamma\rangle^{\natural}=\pos\!\Gamma^{\natural} & \epsilon^{\natural}=1
\end{array}
\]
The reader can easily check that $^{\natural}$ is a surjective homomorphism of $\mathbf{Struct}$ onto the $\{\vee,\necv,\bot \}$-free reduct of $\mathbf{Fm}$. Formula translations of structural terms enable us to give an \textit{(in)equational translation} for $\mathsf{LMC}$-sequents, i.e. a map $^{\sharp}\colon\mathit{Struct}\times\mathit{Fm}\rightarrow \mathit{Fm}^2$ defined by $(\Gamma \Yright \varphi)^{\sharp}=\Gamma^{\natural}\leq \varphi$. Now we want to go the other way round and introduce a \textit{sequent translation} for equations of type $\nu$. Since every equation $\varphi\approx \psi$ may be seen as a pair of inequations $\{\varphi \leq \psi, \psi \leq \varphi\}$, we define a function 
\[^{\flat}\colon\mathit{Fm}^2\uplus\mathit{Fm}^2\rightarrow (\mathit{Struct}\times\mathit{Fm})\cup \wp(\mathit{Struct}\times\mathit{Fm})\] such that: $(\varphi \leq \psi)^{\flat}=\varphi\Yright \psi$, for every inequation $\varphi\leq \psi$; $(\varphi \approx \psi)^{\flat}=\{\varphi\Yright \psi, \psi\Yright \varphi\}$, for every equation $\varphi\approx \psi$.
\begin{proposition}\label{prop: deriv-rules}
The following inference rules are derivable in $\mathsf{LMC}$.
\small{\[
\begin{bprooftree}
	   \def\fCenter{\mbox{\ $\Yright$\ }}
	   \RightLabel{$\sqcap\mathsf{W}_2$}
	   \AxiomC{$\Gamma\{\Delta_2\}\Yright \varphi$}
	   \UnaryInfC{$\Gamma\{\Delta_1 \sqcap \Delta_2 \}\Yright \varphi$}
\end{bprooftree}
\begin{bprooftree}
	   \def\fCenter{\mbox{\ $\Yright$\ }}
	   \RightLabel{$\wedge\mathsf{L}_1$}
	   \AxiomC{$\Gamma\{\varphi_1\}\Yright \chi$}
	   \UnaryInfC{$\Gamma\{\varphi_1 \wedge \varphi_2\}\Yright \chi$}
	\end{bprooftree}
    \begin{bprooftree}
	   \def\fCenter{\mbox{\ $\Yright$\ }}
	   \RightLabel{$\wedge\mathsf{L}_2$}
	   \AxiomC{$\Gamma\{\varphi_2\}\Yright \chi$}
	   \UnaryInfC{$\Gamma\{\varphi_1 \wedge \varphi_2 \}\Yright \chi$}
	\end{bprooftree}
\]
\[
\begin{bprooftree}
	\def\fCenter{\mbox{\ $\Yright$\ }}
	\RightLabel{$\cdot\mathsf{iso}$}
	\AxiomC{$\varphi_1 \Yright \psi_1$}
    \AxiomC{$\varphi_2 \Yright \psi_2$}
	\BinaryInfC{$\varphi_1 \cdot \varphi_2 \Yright \psi_1 \cdot \psi_2$}
	\end{bprooftree}
    \begin{bprooftree}
	\def\fCenter{\mbox{\ $\Yright$\ }}
	\RightLabel{$\wedge\mathsf{iso}$}
	\AxiomC{$\varphi_1 \Yright \psi_1$}
    \AxiomC{$\varphi_2 \Yright \psi_2$}
	\BinaryInfC{$\varphi_1 \wedge\varphi_2 \Yright \psi_1 \wedge \psi_2$}
	\end{bprooftree}
    \begin{bprooftree}
	\def\fCenter{\mbox{\ $\Yright$\ }}
	\RightLabel{$\vee\mathsf{iso}$}
	\AxiomC{$\varphi_1 \Yright \psi_1$}
    \AxiomC{$\varphi_2 \Yright \psi_2$}
	\BinaryInfC{$\varphi_1 \vee \varphi_2 \Yright \psi_1 \vee \psi_2$}
	\end{bprooftree}
\]}\normalsize
\end{proposition}
\begin{proof}
    Left to the reader.
\end{proof}
In what follows, for $\mathsf{r}$ a rule of $\mathsf{LMC}$, $(\mathsf{r})\!\!\downarrow_n$ denotes $n$ successive applications of $\mathsf{r}$.
\begin{lemma}\label{lem: id-deriv}
The sequent translations of all equations in \textnormal{Def. \ref{def: clm}} are derivable in $\mathsf{LMC}$.
\end{lemma}
\begin{proof}
We first give derivations for equations (1), (2), (5), (6), and (7).
{\small\[
\begin{bprooftree}
	\def\fCenter{\mbox{\ $\Yright$\ }}
	\AxiomC{}
    \RightLabel{$\mathsf{init}$}
    \UnaryInfC{$x \Yright x$}
    \RightLabel{$\pos\!\mathsf{R}$}
	\UnaryInfC{$\langle x\rangle  \Yright \pos\!x$}
    \RightLabel{$\mathsf{T}$}
    \UnaryInfC{$x \Yright \pos\! x$}
\end{bprooftree}
    \begin{bprooftree}
	\def\fCenter{\mbox{\ $\Yright$\ }}
	\AxiomC{}
    \RightLabel{$\mathsf{init}$}
    \UnaryInfC{$x \Yright x$}
    \RightLabel{$\necv\mathsf{L}$}
	\UnaryInfC{$\langle \necv\!x \rangle \Yright x$}
    \RightLabel{$\pos\!\mathsf{L}$}
    \UnaryInfC{$ \pos\!\necv\!x \Yright x$}
\end{bprooftree}
\begin{bprooftree}
	\def\fCenter{\mbox{\ $\Yright$\ }}
	\AxiomC{}
    \RightLabel{$\mathsf{init}$}
    \UnaryInfC{$x \Yright x$}
    \RightLabel{$\pos\!\mathsf{R}$}
	\UnaryInfC{$\langle x \rangle \Yright \pos\!x$}
    \RightLabel{$\necv\!\mathsf{R}$}
    \UnaryInfC{$ x \Yright \necv\!\pos\!x$}
\end{bprooftree}
\]}\vspace{0.1cm}
{\small \[
\begin{bprooftree}
	\def\fCenter{\mbox{\ $\Yright$\ }}
	\AxiomC{}
    \RightLabel{$\mathsf{init}$}
    \UnaryInfC{$x \Yright x$}
    \RightLabel{$\pos\!\mathsf{R}$}
	\UnaryInfC{$\langle x \rangle \Yright \pos\!x$}
    \RightLabel{$\mathsf{4}$}
    \UnaryInfC{$ \langle \langle x \rangle \rangle \Yright \pos\!x$}
    \RightLabel{$(\pos\!\mathsf{L})\!\!\downarrow_2$}
    \UnaryInfC{$ \pos\!\pos\! x \Yright \pos\!x$}
\end{bprooftree}
    \begin{bprooftree}
	\def\fCenter{\mbox{\ $\Yright$\ }}
	\AxiomC{}
    \RightLabel{$\mathsf{init}$}
    \UnaryInfC{$x \Yright x$}
    \RightLabel{$\pos\!\mathsf{R}$}
	\UnaryInfC{$\langle x \rangle \Yright \pos\!x$}
    \AxiomC{}
    \RightLabel{$\mathsf{init}$}
    \UnaryInfC{$y \Yright y$}
    \RightLabel{$\pos\!\mathsf{R}$}
	\UnaryInfC{$\langle y \rangle \Yright \pos\!y$}
    \RightLabel{$\cdot\,\mathsf{R}$}
    \BinaryInfC{$ \langle x \rangle \circ \langle y \rangle\Yright \pos\!x \cdot \pos\!y$}
    \RightLabel{$\mathsf{K}$}
    \UnaryInfC{$ \langle x \circ y \rangle\Yright \pos\!x \cdot \pos\!y$}
    \RightLabel{$\cdot\,\mathsf{L}$}
    \UnaryInfC{$ \langle x \cdot y \rangle\Yright \pos\!x \cdot \pos\!y$}
    \RightLabel{$\pos\!\mathsf{L}$}
    \UnaryInfC{$ \pos(x \cdot y)\Yright \pos\!x \cdot \pos\!y$}
\end{bprooftree}
\begin{bprooftree}
	\def\fCenter{\mbox{\ $\Yright$\ }}
	\AxiomC{}
    \RightLabel{$\mathsf{init}$}
    \UnaryInfC{$x \Yright x$}
    \RightLabel{$\pos\!\mathsf{R}$}
	\UnaryInfC{$\langle x \rangle \Yright \pos\!x$}
    \RightLabel{$\mathsf{T}$}
    \UnaryInfC{$ x \Yright \pos\!x$}
    \RightLabel{$\pos\!\mathsf{R}$}
	\UnaryInfC{$\langle x\rangle \Yright \pos\!\pos\!x$}
    \RightLabel{$\pos\!\mathsf{L}$}
    \UnaryInfC{$ \pos\!x \Yright \pos\!\pos\!x$}
\end{bprooftree}
\]}
Note that, since $(\pos\!\pos x \approx \pos x)^{\flat}=\{ \pos\!\pos x \Yright\pos x, \pos x \Yright \pos\!\pos x\}$ the equation (2) corresponds to two derivations (on the bottom right and on the bottom left). The same applies to (3) and (4). By way of example, we provide the derivations of $\pos(x \vee y)\Yright \pos\!x \vee \pos\!y$ and $ \necv \!x \wedge \necv\!y \Yright \necv(x\wedge y)$.
{\small
\[
\begin{bprooftree}
	\def\fCenter{\mbox{\ $\Yright$\ }}
	\AxiomC{}
	\RightLabel{$\mathsf{init}$}
	\UnaryInfC{$x \Yright x$}	
	\RightLabel{$\pos \mathsf{R}$}
	\UnaryInfC{$\langle x \rangle \Yright \pos\!x$}
    \RightLabel{$\vee\mathsf{R}_1$}
	\UnaryInfC{$\langle x \rangle \Yright \pos\!x \vee \pos\!y$ }
	\AxiomC{}
    \RightLabel{$\mathsf{init}$}
	\UnaryInfC{$y \Yright y$}	
	\RightLabel{$\pos \mathsf{R}$}
	\UnaryInfC{$\langle y \rangle \Yright \pos\!y$}
    \RightLabel{$\vee\mathsf{R}_2$}
	\UnaryInfC{$\langle y \rangle \Yright \pos\!x\vee \pos\!y$ }
	\RightLabel{$\vee \mathsf{L}$}
	\BinaryInfC{$\langle x \vee y \rangle \Yright \pos\!x \vee \pos\!y$ }
    \RightLabel{$\pos\mathsf{L}$}
	\UnaryInfC{$\pos(x \vee y)\Yright \pos\!x \vee \pos\!y$}
\end{bprooftree}
\begin{bprooftree}
	\def\fCenter{\mbox{\ $\Yright$\ }}
	\AxiomC{}
	\RightLabel{$\mathsf{init}$}
	\UnaryInfC{$x \Yright x$}	
    \RightLabel{$\necv\mathsf{L}$}
	\UnaryInfC{$\langle \necv \!x \rangle \Yright x$ }
    \RightLabel{$\wedge\mathsf{L}_1$}
    \UnaryInfC{$\langle \necv \!x \wedge \necv\!y \rangle \Yright x$}
	\AxiomC{}
	\RightLabel{$\mathsf{init}$}
	\UnaryInfC{$y \Yright y$}	
    \RightLabel{$\necv\mathsf{L}$}
	\UnaryInfC{$\langle \necv \!y\rangle \Yright y$ }
    \RightLabel{$\wedge\mathsf{L}_2$}
    \UnaryInfC{$\langle \necv \!x \wedge \necv\!y \rangle \Yright y$}
    \RightLabel{$\wedge\mathsf{R}$}
	\BinaryInfC{$\langle \necv \!x \wedge \necv\!y \rangle \Yright x\wedge y$}
    \RightLabel{$\necv\mathsf{R}$}
    \UnaryInfC{$ \necv \!x \wedge \necv\!y \Yright \necv(x\wedge y)$}
\end{bprooftree}\]}
The converse cases are left to the reader.
\end{proof}
It is not difficult to exhibit derivations in~$\mathsf{LMC}$ for the equations proved in Lemmas~\ref{lem: dereq1} and~\ref{lem: dereq2}. For example, here are proofs for the weak co-nuclearity of~$\necv$.
{\small
\[
\begin{bprooftree}
	\def\fCenter{\mbox{\ $\Yright$\ }}
	\AxiomC{}
	\RightLabel{$\mathsf{init}$}
	\UnaryInfC{$x \Yright x$}	
	\RightLabel{$\necv\! \mathsf{L}$}
	\UnaryInfC{$\langle \necv\!x \rangle \Yright x$}
    \RightLabel{$\mathsf{4}$}
	\UnaryInfC{$\langle\langle \necv\!x \rangle\rangle \Yright x$}
    \RightLabel{$\necv\! \mathsf{R}$}
	\UnaryInfC{$\langle \necv\!x \rangle \Yright \necv\!x$}
	\AxiomC{}
    \RightLabel{$\mathsf{init}$}
	\UnaryInfC{$y \Yright y$}	
	\RightLabel{$\necv\!\mathsf{L}$}
	\UnaryInfC{$\langle \necv\!y \rangle \Yright y$}
    \RightLabel{$\mathsf{4}$}
	\UnaryInfC{$\langle\langle \necv\!y \rangle\rangle \Yright y$}
    \RightLabel{$\necv\! \mathsf{R}$}
	\UnaryInfC{$\langle \necv\!y \rangle \Yright \necv\!y$}
    \RightLabel{$\cdot\mathsf{R}$}
	\BinaryInfC{$\langle \necv\!x \rangle \circ \langle\necv\!y \rangle \Yright \necv\!x \cdot \necv\!y$ }
    \RightLabel{$\mathsf{K}$}
	\UnaryInfC{$\langle \necv\!x \circ \necv\!y \rangle \Yright \necv\!x \cdot \necv\!y$}
    \RightLabel{$\cdot \mathsf{L}$}
    \UnaryInfC{$\langle \necv\!x \cdot \necv\!y \rangle \Yright \necv\!x \cdot \necv\!y$}
    \RightLabel{$\necv\!\mathsf{L}$}
    \UnaryInfC{$\langle\langle \necv( \necv\!x \cdot \necv\!y) \rangle\rangle \Yright \necv\!x \cdot \necv\!y$}
    \RightLabel{$(\mathsf{T})\!\!\downarrow_2$}
    \UnaryInfC{$ \necv( \necv\!x \cdot \necv\!y) \Yright \necv\!x \cdot \necv\!y$}
\end{bprooftree}
\begin{bprooftree}
	\def\fCenter{\mbox{\ $\Yright$\ }}
	\AxiomC{}
	\RightLabel{$\mathsf{init}$}
	\UnaryInfC{$x \Yright x$}	
	\RightLabel{$\necv\! \mathsf{L}$}
	\UnaryInfC{$\langle \necv\!x \rangle \Yright x$}
    \RightLabel{$\mathsf{4}$}
	\UnaryInfC{$\langle\langle \necv\!x \rangle\rangle \Yright x$}
    \RightLabel{$\necv\! \mathsf{R}$}
	\UnaryInfC{$\langle \necv\!x \rangle \Yright \necv\!x$}
	\AxiomC{}
    \RightLabel{$\mathsf{init}$}
	\UnaryInfC{$y \Yright y$}	
	\RightLabel{$\necv\!\mathsf{L}$}
	\UnaryInfC{$\langle \necv\!y \rangle \Yright y$}
    \RightLabel{$\mathsf{4}$}
	\UnaryInfC{$\langle\langle \necv\!y \rangle\rangle \Yright y$}
    \RightLabel{$\necv\! \mathsf{R}$}
	\UnaryInfC{$\langle \necv\!y \rangle \Yright \necv\!y$}
    \RightLabel{$\cdot\mathsf{R}$}
	\BinaryInfC{$\langle \necv\!x \rangle \circ \langle\necv\!y \rangle \Yright \necv\!x \cdot \necv\!y$ }
    \RightLabel{$\mathsf{K}$}
	\UnaryInfC{$\langle \necv\!x \circ \necv\!y \rangle \Yright \necv\!x \cdot \necv\!y$}
    \RightLabel{$\necv\! \mathsf{R}$}
    \UnaryInfC{$ \necv\!x \circ \necv\!y \Yright \necv(\necv\!x \cdot \necv\!y)$}
    \RightLabel{$\cdot\mathsf{L}$}
    \UnaryInfC{$ \necv\!x \cdot \necv\!y \Yright \necv(\necv\!x \cdot \necv\!y)$}
\end{bprooftree}\]}
It can likewise be seen that (\ref{eq: CN2}) has no derivable counterparts in our system. In particular, $\nvdash_{\mathsf{LMC}}\necv\!x \Yright \necv\!x \cdot \necv\!1$ and $\nvdash_{\mathsf{LMC}}\necv\!x \Yright \necv\!1 \cdot \necv\!x$. It is also worth noting that $\mathsf{LMC}$ does not derive the sequent translations of $\pos\!1\approx 1$ and $\pos(x \cdot y)\approx \pos\!x \vee (x\cdot \pos\!y)$, since it can be readily verified that $\nvdash_{\mathsf{LMC}} \pos\!1\Yright 1$ and $\nvdash_{\mathsf{LMC}}\pos(x \cdot y)\Yright\pos\!x \vee (x\cdot \pos\!y)$. 

We are now going to show that $\mathsf{LMC}$ is sound and complete with respect to the variety of closure $\ell$-monoids. First, we need to check soundness of individual inference rules.
\begin{lemma}\label{lem: sound-rules}
Let {\def\fCenter{\mbox{\ $\Yright$\ }}\AxiomC{$s_1, \dots, s_n$}\UnaryInfC{$s$}\DisplayProof} be an inference rule of $\mathsf{LMC}$. Then $ s_1^{\sharp}, \dots, s_n^{\sharp} \vDash_{\text{\scriptsize$\mathfrak{LMC}$}} s^{\sharp}$.
\end{lemma}
\begin{proof}
The proof is trivial for almost all the rules displayed in Table \ref{tab: LMC}. The only rules for which some caution is required are $\mathsf{cut}$, $\mathsf{\sqcap\mathsf{W}_1}$, $\mathsf{K}$, $\vee\mathsf{L}$, and $\necv\!\mathsf{L}$. The reason is that, in these cases, the subterms appearing in the left-hand sides of sequents are permitted to occur at an arbitrary depth, yet the rule's structure alone does not guarantee soundness  (unlike, for instance, with $\circ\, \epsilon$, $\sqcap\mathsf{A}$, or $\pos\!\mathsf{L}$). 
\begin{itemize}
\item Let us start by considering the rules $\mathsf{cut}$, $\sqcap\mathsf{W}_1$, $\mathsf{K}$, and $\necv\!\mathsf{L}$. We first note that all fundamental operations of closure $\ell$-monoids are order-preserving. This implies that the interpretation of formula translations for structural terms results in term operations that are inherently order-preserving. Take now $\Gamma\{-\}$ to be a structural term with an empty argument place. For $\Delta\in\mathit{Struct}$ and $\varphi, \psi \in \mathit{Fm}$, suppose that $\mathfrak{LMC}$ satisfies $\Gamma\{\varphi\}^{\natural}\leq \psi$ and $\Delta^{\natural}\leq\varphi$. But then $\mathfrak{LMC}\models\Gamma\{\Delta\}^{\natural}\leq\Gamma\{\varphi\}^{\natural}$ whence, by transitivity, $\mathfrak{LMC}\models \Gamma\{\Delta\}^{\natural}\leq \psi$. This establishes the soundness of $\mathsf{cut}$. Analogous arguments apply to the remaining cases.
\item We now turn to $\vee\mathsf{L}$. Here, it is easy to see that order preservation is of no help in checking soundness for this rule. However, by examining the syntax of structural terms, we can observe that the counterpart of each structural operator is a join-distributive connective. Therefore, to prove that $\vee\mathsf{L}$ is sound, it is enough to show that the equation
\[E:=\Gamma\{ \varphi \}^{\natural}\vee \Gamma\{ \psi \}^{\natural} \approx \Gamma\{ \varphi \vee \psi \}^{\natural}\]
is satisfied by $\mathfrak{LMC}$. We proceed by induction on $\mathrm{cp}_s(\Gamma)$, recalling that the formulas, understood as structural terms, have complexity 0. In the base case our claim clearly holds. Suppose now that $\mathfrak{LMC}\models E$ for all the instances of $\vee\mathsf{L}$ with structural terms of complexity $\leq n$. For $\xi\in\{\varphi,\psi,\varphi\vee\psi\}$, let $\Gamma\{\xi\}:=\Delta \circ \Gamma'\{\xi \}$, where $\mathrm{cp}_s(\Gamma')=n$. Since 
\[ \Delta^{\natural} \cdot ( \Gamma'\{\varphi \}^{\natural} \vee \Gamma'\{\psi \}^{\natural}) \approx (\Delta^{\natural} \cdot\Gamma'\{\varphi \}^{\natural}) \vee (\Delta^{\natural} \cdot\Gamma'\{\psi \}^{\natural})  \]
is an instance of a defining equation of $\mathfrak{LMC}$ and, by induction hypothesis, $\mathfrak{LMC}\models \Gamma'\{ \varphi \}^{\natural}\vee \Gamma'\{ \psi \}^{\natural} \approx \Gamma'\{ \varphi \vee \psi \}^{\natural}$, we immediately obtain that
\[\mathfrak{LMC}\models (\Delta^{\natural} \cdot\Gamma'\{\varphi \}^{\natural}) \vee (\Delta^{\natural} \cdot\Gamma'\{\psi \}^{\natural}) \approx \Delta^{\natural} \cdot \Gamma'\{ \varphi \vee \psi \}^{\natural}. \]
We leave it to the reader to check the cases in which $\Gamma\{\xi\}$ is defined as $\Gamma'\{\xi \}\circ \Delta$, $\Delta \sqcap \Gamma'\{\xi \}$, $\Gamma'\{\xi \}\sqcap \Delta$, and $\langle\Gamma'\{\xi \}\rangle$.\qedhere
\end{itemize}
\end{proof}
\begin{corollary}
Any derivable rule of $\mathsf{LMC}$ is sound.
\end{corollary}
\begin{theorem}\label{theo: completeness}
Let $\Gamma \Yright \varphi$ be an $\mathsf{LMC}$-sequent. Then $\vdash_{\mathsf{LMC}}\Gamma \Yright \varphi$ iff $\mathfrak{LMC}\models (\Gamma \Yright \varphi)^{\sharp}$.
\end{theorem}
\begin{proof}
Drawing upon Lemma \ref{lem: sound-rules}, soundness is obtained via a straightforward induction on the height of derivations in $\mathsf{LMC}$. We establish completeness via the Lindenbaum-Tarski method. Let $\Theta (\mathsf{LMC})\subseteq \mathit{Fm}^2$ be defined by
\[ \Theta (\mathsf{LMC}) := \left\{ \langle \varphi, \psi \rangle \in \mathit{Fm}^2 \,|\, \text{$\vdash_{\mathsf{LMC}}\varphi \Yright \psi$ and $\vdash_{\mathsf{LMC}}\psi \Yright \varphi$} \right\}.\]
It is easy to see that $\Theta (\mathsf{LMC})$ is an equivalence relation on $\mathit{Fm}$: it is symmetric by design, while reflexivity and transitivity follow from the rules $\mathsf{init}$ and $\mathsf{cut}$, respectively. We now show that $\Theta (\mathsf{LMC})$ is a congruence on $\mathbf{Fm}$. For $\varphi_1,\varphi_2,\psi_1,\psi_2\in \mathit{Fm}$, suppose that $\langle \varphi_1, \psi_1 \rangle, \langle \varphi_2, \psi_2 \rangle \in \Theta (\mathsf{LMC})$. By using the rules $\cdot\mathsf{iso}$, $\wedge\mathsf{iso}$, and $\vee\mathsf{iso}$ from Proposition \ref{prop: deriv-rules}, we can derive in one step the sequents:
\[
\begin{array}{lll}
   \varphi_1 \cdot \varphi_2 \Yright \psi_1 \cdot \psi_2  &\phantom{A} &\psi_1 \cdot \psi_2 \Yright \varphi_1 \cdot\varphi_2 \\
    \varphi_1 \wedge \varphi_2 \Yright \psi_1 \wedge \psi_2  &\phantom{A}& \psi_1 \wedge \psi_2 \Yright \varphi_1 \wedge\varphi_2 \\
    \varphi_1 \vee \varphi_2 \Yright \psi_1 \vee \psi_2  &\phantom{A} &\psi_1 \vee \psi_2 \Yright \varphi_1 \vee\varphi_2 \\
\end{array}
\]
Therefore $\langle \varphi_1 \divideontimes \varphi_2, \psi_1 \divideontimes \psi_2\rangle \in \Theta (\mathsf{LMC})$, for $\divideontimes \in \{ \cdot, \wedge, \vee \}$. Suppose now that $\langle \varphi, \psi \rangle \in \Theta (\mathsf{LMC})$. Then $\langle \pos\!\varphi, \pos\!\psi \rangle, \langle \necv\!\varphi, \necv\!\psi \rangle \in \Theta (\mathsf{LMC})$, as shown by the following proofs:\small
\[\begin{bprooftree}
	\def\fCenter{\mbox{\ $\Yright$\ }}
	\AxiomC{$d$}
    \UnaryInfC{$\varphi \Yright \psi$}
    \RightLabel{$\pos\!\mathsf{R}$}
	\UnaryInfC{$\langle \varphi \rangle \Yright \pos\!\psi$}
    \RightLabel{$\pos\!\mathsf{L}$}
    \UnaryInfC{$ \pos\!\varphi \Yright \pos\!\psi$}
\end{bprooftree}
\begin{bprooftree}
	\def\fCenter{\mbox{\ $\Yright$\ }}
	\AxiomC{$d'$}
    \UnaryInfC{$\psi \Yright \varphi$}
    \RightLabel{$\pos\!\mathsf{R}$}
	\UnaryInfC{$\langle \psi \rangle \Yright \pos\!\varphi$}
    \RightLabel{$\pos\!\mathsf{L}$}
    \UnaryInfC{$ \pos\!\psi \Yright \pos\!\varphi$}
\end{bprooftree}
\begin{bprooftree}
	\def\fCenter{\mbox{\ $\Yright$\ }}
	\AxiomC{$d$}
    \RightLabel{}
    \UnaryInfC{$\varphi \Yright \psi$}
    \RightLabel{$\necv\mathsf{L}$}
	\UnaryInfC{$\langle \necv\!\varphi \rangle \Yright \psi$}
    \RightLabel{$\necv\mathsf{R}$}
    \UnaryInfC{$ \necv\!\varphi \Yright \necv\!\psi$}
\end{bprooftree}
\begin{bprooftree}
	\def\fCenter{\mbox{\ $\Yright$\ }}
	\AxiomC{$d'$}
    \RightLabel{}
    \UnaryInfC{$\psi \Yright \varphi$}
    \RightLabel{$\necv\mathsf{L}$}
	\UnaryInfC{$\langle \necv\!\psi \rangle \Yright \varphi$}
    \RightLabel{$\necv\mathsf{R}$}
    \UnaryInfC{$ \necv\!\psi \Yright \necv\!\varphi$}
\end{bprooftree}\]\normalsize
We conclude that $\Theta (\mathsf{LMC}) \in \mathrm{Con}(\mathbf{Fm})$.
We proceed to show that the quotient algebra $\mathbf{Fm}/\Theta (\mathsf{LMC})$ is a closure $\ell$-monoid. Let $\leq_{\mathsf{LMC}}$ be the binary relation on $\mathit{Fm}/\Theta (\mathsf{LMC}) $ defined by $[\varphi] \leq_{\mathsf{LMC}} [\psi]$ iff $\vdash_{\mathsf{LMC}} \varphi \Yright \psi$, for $\varphi, \psi \in \mathit{Fm}$. Clearly, $\leq_{\mathsf{LMC}}$ is a partial order (reflexifity and transitivity again follow from the rules $\mathsf{init}$ and $\mathsf{cut}$; antisymmetry holds by design). Most importantly, $\leq_{\mathsf{LMC}}$ is a lattice order. Indeed it is immediate to check that $[\varphi] \leq_{\mathsf{LMC}} [\psi]$ is equivalent to $ [\varphi] \wedge [\psi] = [\varphi]$. Therefore $\langle\mathit{Fm}/\Theta (\mathsf{LMC}), \wedge, \vee \rangle$ is a lattice. By constructing suitable derivations, it can be readily verified that, for all $\varphi, \psi, \chi \in \mathit{Fm}$:
\begin{align*}      &([\varphi]\cdot[\psi])\cdot[\chi]=[\varphi]\cdot([\psi]\cdot[\chi])\\
    &[\varphi]\cdot[1]=[1]\cdot[\varphi]=[\varphi]\\
    &[\varphi]\cdot[\bot]=[\bot]\cdot[\varphi]=[\bot]\\
    & \text{$[\bot] \leq_{\mathsf{LMC}} [\varphi]$ and $[\varphi] \leq_{\mathsf{LMC}} [\top]$}\\
    & [\varphi] \leq_{\mathsf{LMC}} [\psi] \Rightarrow [\varphi]\cdot[\chi] \leq_{\mathsf{LMC}} [\psi]\cdot[\chi]\\
    & [\varphi] \leq_{\mathsf{LMC}} [\psi] \Rightarrow [\chi]\cdot[\varphi] \leq_{\mathsf{LMC}}[\chi]\cdot[\psi]
\end{align*}
Hence $\langle\mathit{Fm}/\Theta (\mathsf{LMC}), \wedge, \vee, \cdot, [1], [\bot], [\top] \rangle$ is a bounded $\ell$-monoid with $[\bot]$ an absorbing element for multiplication. In order to prove distributive laws for lattice operations, the role of $\sqcap$ is crucial. We just construct derivations for $x \wedge (y \vee z)\Yright (x \wedge y) \vee (x \wedge z)$ and $ (x \vee y) \wedge (x \vee z) \Yright x \vee (y \wedge z)$.
\begin{center}
\begin{prooftree}
\small
	\def\fCenter{\mbox{\ $\Yright$\ }}
	\AxiomC{}
    \RightLabel{$\mathsf{init}$}
    \UnaryInfC{$x \Yright x$}
    \RightLabel{$\sqcap\mathsf{W}_1$}
    \UnaryInfC{$x \sqcap y \Yright x$}
    \AxiomC{}
    \RightLabel{$\mathsf{init}$}
    \UnaryInfC{$y \Yright y$}
    \RightLabel{$\sqcap\mathsf{W}_2$}
    \UnaryInfC{$x \sqcap y \Yright y$}
    \RightLabel{$\wedge \mathsf{R}$}
    \BinaryInfC{$x \sqcap y \Yright x \wedge y$}
    \RightLabel{$\vee\mathsf{R}_1$}
    \UnaryInfC{$x \sqcap y \Yright (x \wedge y) \vee (x \wedge z)$}
    \AxiomC{}
    \RightLabel{$\mathsf{init}$}
    \UnaryInfC{$x \Yright x$}
    \RightLabel{$\sqcap\mathsf{W}_1$}
    \UnaryInfC{$x \sqcap z \Yright x$}
    \AxiomC{}
    \RightLabel{$\mathsf{init}$}
    \UnaryInfC{$z \Yright z$}
    \RightLabel{$\sqcap\mathsf{W}_2$}
    \UnaryInfC{$x \sqcap z \Yright z$}
    \RightLabel{$\wedge \mathsf{R}$}
    \BinaryInfC{$x \sqcap z \Yright x \wedge z$}
    \RightLabel{$\vee\mathsf{R}_2$}
    \UnaryInfC{$x \sqcap y \Yright (x \wedge y) \vee (x \wedge z)$}
    \RightLabel{$\vee\mathsf{L}$}
    \BinaryInfC{$x \sqcap (y \vee z)\Yright (x \wedge y) \vee (x \wedge z)$}
    \RightLabel{$\wedge\mathsf{L}$}
    \UnaryInfC{$x \wedge (y \vee z)\Yright (x \wedge y) \vee (x \wedge z)$}
\end{prooftree}
\begin{prooftree}
\small
\def\fCenter{\mbox{\ $\Yright$\ }}
\AxiomC{}
\RightLabel{$\mathsf{init}$}
\UnaryInfC{$x \Yright x$}
\RightLabel{$\sqcap\mathsf{W}_2$}
\UnaryInfC{$(x \vee y) \sqcap x \Yright x$}
\RightLabel{$\vee\mathsf{R}_1$}
\UnaryInfC{$(x \vee y) \sqcap x \Yright x \vee (y \wedge z)$}
\AxiomC{}
\RightLabel{$\mathsf{init}$}
\UnaryInfC{$x \Yright x$}
\RightLabel{$\sqcap\mathsf{W}_1$}
\UnaryInfC{$x \sqcap z \Yright x$}
\RightLabel{$\vee\mathsf{R}_1$}
\UnaryInfC{$x \sqcap z \Yright x \vee (y \wedge z)$}
\AxiomC{}
\RightLabel{$\mathsf{init}$}
\UnaryInfC{$y \Yright y$}
\RightLabel{$\sqcap\mathsf{W}_1$}
\UnaryInfC{$y \sqcap z \Yright y$}
\AxiomC{}
\RightLabel{$\mathsf{init}$}
\UnaryInfC{$z \Yright z$}
\RightLabel{$\sqcap\mathsf{W}_2$}
\UnaryInfC{$y \sqcap z \Yright z$}
\RightLabel{$\wedge\mathsf{R}$}
\BinaryInfC{$y \sqcap z \Yright y \wedge z$}
\RightLabel{$\vee\mathsf{R}_2$}
\UnaryInfC{$y \sqcap z \Yright x \vee (y \wedge z)$}
\RightLabel{$\vee\mathsf{L}$}
\BinaryInfC{$ (x \vee y) \sqcap z \Yright x \vee (y \wedge z)$}
\RightLabel{$\vee\mathsf{L}$}
\BinaryInfC{$ (x \vee y) \sqcap (x \vee z) \Yright x \vee (y \wedge z)$}
\RightLabel{$\wedge\mathsf{L}$}
\UnaryInfC{$ (x \vee y) \wedge (x \vee z) \Yright x \vee (y \wedge z)$}
\end{prooftree}
\end{center}
Finally, by using Lemma \ref{lem: id-deriv}, we immediately see that operations $\pos [\varphi]:=[\pos\!\varphi]$ and $\necv [\varphi]:=[\necv\!\varphi]$ satisfy equations (1)-(7) in Def. \ref{def: clm}. Therefore $\mathbf{Fm}/\Theta (\mathsf{LMC})\in \mathfrak{LMC}$. Let now $\mathit{h}\colon\mathbf{Fm}\to\mathbf{Fm}/\Theta (\mathsf{LMC})$ be the natural homomorphism defined by $\varphi \mapsto [\varphi]$ and suppose that, for $\Gamma \Yright \varphi$ an $\mathsf{LMC}$-sequent, $\mathfrak{LMC}\models (\Gamma \Yright \varphi)^{\sharp}$. Therefore $\mathfrak{LMC}\models \Gamma^{\natural} \leq \varphi$ and, in particular, $h(\Gamma^{\natural}) \leq_{\mathsf{LMC}} h(\varphi)$ in $\mathbf{Fm}/\Theta (\mathsf{LMC})$. But then, by the definition of $\leq_{\mathsf{LMC}}$, we have $\vdash_{\mathsf{LMC}}\Gamma^{\natural} \Yright \varphi$. Since obviously $\vdash_{\mathsf{LMC}} \Gamma \Yright \Gamma^{\natural}$, by applying $\mathsf{cut}$ we get $\vdash_{\mathsf{LMC}} \Gamma \Yright \varphi$.
\end{proof}
\section{\texorpdfstring{Cut Elimination for $\mathsf{LMC}$}{Cut Elimination for LMC}}\label{sect: cutelim}
This section is devoted to proving a cut elimination result for $\mathsf{LMC}$. As the reader will have observed, in derivations of the form\small
\[
\begin{bprooftree}
    \AxiomC{$d_1$}
    \UnaryInfC{$\Delta \Yright \varphi$}
    \AxiomC{$d_2$}				
	\UnaryInfC{$\Gamma \{ \Pi\{\varphi\}\sqcap \Pi\{\varphi\}\} \Yright \psi $}
    \RightLabel{$\sqcap\mathsf{C}$}
    \UnaryInfC{$\Gamma \{ \Pi\{\varphi\}\} \Yright \psi $}
    \RightLabel{$\mathsf{cut}$}
	\BinaryInfC{$\Gamma\{ \Pi\{\Delta\}\}\Yright \psi$}
\end{bprooftree}
\]\normalsize
permuting the $\sqcap\mathsf{C}$ with $\mathsf{cut}$ does not result in a decrease in the height of $d_2$. A direct proof of cut elimination for $\mathsf{LMC}$ is therefore precluded. This comes as no surprise, since analogous situations arise (with respect to contraction rules for different connectives) in many well-known logical calculi. 
We follow the standard approach of introducing a mix rule—which subsumes $\mathsf{cut}$ as a special case—to simplify the structure of derivations by eliminating, in a single move, derivation steps that obstruct induction in cut elimination arguments. We then establish a mix elimination theorem for $\mathsf{LMC}$ (without $\mathsf{cut}$ and with the mix rule), from which cut elimination is obtained as a direct corollary.
\subsection{A (Weak) Mix Rule for \texorpdfstring{$\mathsf{LMC}$}{LMC}}
The form of a mix rule depends on the intrinsic features of the calculus for which it is designed. In our framework, two central aspects must be addressed. First, structural terms of $\mathsf{LMC}$ are not sets or multisets, but rather have a \textit{tree structure}. Second, as in the case of the substructural logic $\mathsf{FL_{gc}}$ \cite{Surarso-Ono1996}, the rule $\sqcap\mathsf{C}$ operates on structural terms, hence it is a \textit{global} contraction rule. 

For $n\geq 0$, let us write $\Gamma\lBrace \Delta \rBrace_n$ to denote a structural term $\Gamma$ where we distinguish a family of $n$ occurrences of $\Delta$. Note that each of these occurrences corresponds to a node in the syntactic tree associated with $\Gamma$. Clearly, $\Gamma\lBrace \Delta \rBrace_1$ is just an alternative writing for $\Gamma\{ \Delta \}$. The notations $\Gamma\{-\}$ and $\Gamma\lBrace-\rBrace_n$ can also be combined: we thus write $\Gamma\{\Pi\}\lBrace\Delta\rBrace_n$ to denote a structural term $\Gamma$ in which both a single occurrence of $\Pi$ and $n$ occurrences of $\Delta$ are distinguished. In the case where $\Pi$ contains an occurrence of $\Delta$, we write $\Gamma\{\Pi\{\underline{\Delta}\}\}\lBrace\Delta\rBrace_n$ if the selected $\Delta$ in $\Pi$ is included in $\lBrace\Delta\rBrace_n$, and $\Gamma\{\Pi\{\Delta\}\}\lBrace\Delta\rBrace_n$ otherwise. In this setting, assuming the choice of all occurrences of $\Delta$ to be fixed, it is clear that structural terms of the form $\Gamma\{\Pi\{\underline{\Delta}\}\}\lBrace\Delta\rBrace_n$ and $\Gamma\{\Pi\{\Delta\}\}\lBrace\Delta\rBrace_{n-1}$ coincide. 

We are now ready to define a mix rule for $\mathsf{LMC}$. 
\begin{definition}
Let us denote by $\mathsf{mix}$ the inference rule
\begin{prooftree}
	\AxiomC{$\Delta \Yright \varphi$}
	\AxiomC{$\Gamma \lBrace\varphi\rBrace_n\Yright \psi $}
	\BinaryInfC{$\Gamma\lBrace\Delta\rBrace _n\Yright \psi$}
\end{prooftree}
where $\Gamma\lBrace\Delta\rBrace _n$ is the structural term obtained by replacing, in $\Gamma \lBrace\varphi\rBrace_n$, each of the $n$ occurrences of $\varphi$ with an occurrence of $\Delta$. In an instance of $\mathsf{mix}$, $\varphi$ is called the $\mathsf{mix}$-formula. We define $\mathsf{LMC}_{\ltimes}$ as the system $(\mathsf{LMC}-\mathsf{cut})+\mathsf{mix}$.
\end{definition}
Note that $\mathsf{mix}$ is close in spirit to Kiriyama and Ono's \textit{weak-mix} rule for (the propositional fragment of) $\mathsf{FL_{ec}}$ \cite{Kiriyama-Ono1991}, in that it permits the simultaneous substitution of some, but not necessarily all, occurrences of the $\mathsf{mix}$-formula.
\begin{remark}\label{rem: cutmixeq}
Every instance of $\mathsf{cut}$ is nothing but an instance of $\mathsf{mix}$ where just one occurrence of the $\mathsf{mix}$-formula has been singled out. Conversely, every instance of $\mathsf{mix}$ over $n$ occurrences of the $\mathsf{mix}$-formula can be replaced by a sequence of $n$ instances of~$\mathsf{cut}$.
\end{remark}
\subsection{Mix Elimination for \texorpdfstring{$\mathsf{LMC}_{\ltimes}$}{LMC⧔}}
We establish a mix elimination theorem for $\mathsf{LMC}_{\ltimes}$ by adapting the method of Metcalfe \textit{et al.} \cite{MOG2009}, developed for cut-elimination in hypersequent systems for fuzzy logics and subsequently applied to the hypersequent calculus $\mathsf{CSemFL}$~\cite{MPT}. In what follows, for $d$ a derivation in $\mathsf{LMC}_{\ltimes}$, we denote by $h(d)$ its height. 
\begin{definition} The \textnormal{rank} of an instance of $\mathsf{mix}$
\begin{prooftree}
    \AxiomC{$d_1$}
	\UnaryInfC{$\Delta \Yright \varphi$}
    \AxiomC{$d_2$}
	\UnaryInfC{$\Gamma \lBrace\varphi\rBrace_n\Yright \psi $}
    \RightLabel{$\mathsf{mix}$}
	\BinaryInfC{$\Gamma\lBrace\Delta\rBrace _n\Yright \psi$}
\end{prooftree}
is the ordered triple $\langle \mathrm{cp}(\varphi),p(d_1),h(d_1)+h(d_2)\rangle$, where $p(d_1)=0$ if $d_1$ ends with an application of a rule that (read backwards) decomposes $\varphi$, and $p(d_1)=1$ otherwise. We regard ranks as elements of the partially ordered set $\langle \mathbb{N}^3, \leq_{\mathrm{lex}}\rangle$, where $\leq_{\mathrm{lex}}$ is the lexicographic order defined from three copies of $\mathbb{N}$.
\end{definition}
We write $\vdash^{\mathrm{mf}}_{\mathsf{LMC}_{\ltimes}}\!\!\Gamma\Yright\varphi$ to denote that an $\mathsf{LMC}_{\ltimes}$-sequent $\Gamma \Yright \varphi$ admits a $\mathsf{mix}$-free proof.
\begin{theorem}\label{theo: mixelim}
The rule $\mathsf{mix}$ can be eliminated from $\mathsf{LMC}_{\ltimes}$.
\end{theorem}
\begin{proof}
Our goal is to show that any derivation in $\mathsf{LMC}_{\ltimes}$ making use of $\mathsf{mix}$ can be transformed into a $\mathsf{mix}$-free one. To this end, it suffices to consider uppermost instances of $\mathsf{mix}$ in derivations. We prove by induction on the $\mathsf{mix}$ rank that, if $d_1$ is a $\mathsf{mix}$-free derivation of $\Delta\Yright \varphi$ and $d_2$ is a $\mathsf{mix}$-free derivation of $\Gamma \lBrace\varphi\rBrace_n\Yright \psi$, then it is possible to construct a $\mathsf{mix}$-free proof of $\Gamma \lBrace\Delta\rBrace_n\Yright \psi$.

As for the base case, suppose that $h(d_1)+h(d_2)=2$. We consider the following alternatives:
\begin{enumerate}
    \item If $d_1$ and $d_2$ are instances of $\mathsf{init}$, then $\varphi=\psi$ and our claim trivially holds.
    \item If $d_1$ is an instance of $\bot\mathsf{L}$ and $d_2$ is an instance of $\top\mathsf{R}$, then we have the following transformation:
\[
\begin{bprooftree}
    \AxiomC{}
    \RightLabel{$\bot\mathsf{L}$}
    \UnaryInfC{$\Delta\{\bot\}\Yright \varphi$}
    \AxiomC{}	
    \RightLabel{$\top\mathsf{R}$}
	\UnaryInfC{$\Gamma\lBrace\varphi\rBrace_n\Yright \top$}
    \RightLabel{$\mathsf{mix}$}
	\BinaryInfC{$\Gamma\lBrace\Delta\{\bot\}\rBrace_n\Yright \top$}
    \end{bprooftree}
    \text{\normalsize$\;\dashrightarrow$}
    \begin{bprooftree}
    \AxiomC{}
    \noLine
	\UnaryInfC{\phantom{$\Gamma\{ \Delta\{\bot\}\}^{\bullet}\Yright \top$}}
    \RightLabel{$\bot\mathsf{L}/\top\mathsf{R}$}
	\UnaryInfC{$\Gamma\lBrace\Delta\{\bot\}\rBrace_n\Yright \top$}
	\end{bprooftree}
\]
\item If $d_1$ is $1\mathsf{R}$ and $d_2$ is the initial sequent $1 \Yright 1$, then by applying $\mathsf{mix}$ we would obtain $\epsilon \Yright 1$ and so the whole derivation reduces to $1\mathsf{R}$. The same reasoning applies if: $d_1$ is the initial sequent $\bot \Yright \bot$ and $d_2$ is any instance of $\bot\mathsf{L}$; $d_1$ ($d_2$) is the initial sequent $\top\Yright\top$ and the endsequent of $d_2$ ($d_1$) is of the form $\Gamma\{\top\}\Yright\top$. 
\end{enumerate}
For the induction step, we distinguish two possible classes of cases:
\begin{itemize}
    \item[(i)] Either no new occurrence of the $\mathsf{mix}$-formula is introduced on the right side of the endsequent of $d_1$ or no new occurrence of the $\mathsf{mix}$-formula is introduced on the left side of the endsequent of $d_2$.
    \item[(ii)] A new occurrence of the $\mathsf{mix}$-formula is introduced both on the right side of the endsequent of $d_1$ and on the left side of the endsequent of $d_2$.
\end{itemize}
Let us begin with the cases falling under (i). Suppose that $d_1$ and $d_2$ end with applications of unary rules that, respectively, do not introduce any new occurrences of the $\mathsf{mix}$-formula on the right side of the endsequent of $d_1$ and the left side of the endsequent of $d_2$. In this case, we have the following transformation:
\begin{equation}\label{eq: one}
\begin{bprooftree}
\AxiomC{$d_1$}
\UnaryInfC{$\Delta' \Yright \varphi$}
\RightLabel{$\mathsf{r}_1$}
\UnaryInfC{$\Delta \Yright \varphi$}
\AxiomC{$d_2$}
\UnaryInfC{$\Gamma' \Yright \psi'$}
\RightLabel{$\mathsf{r}_2$}
\UnaryInfC{$\Gamma\lBrace \varphi \rBrace_n \Yright \psi$}
\RightLabel{$\mathsf{mix}$}
\BinaryInfC{$\Gamma\lBrace \Delta \rBrace_n \Yright \psi$}
\end{bprooftree}
\dashrightarrow
\begin{bprooftree}
\AxiomC{$d_1$}
\UnaryInfC{$\Delta' \Yright \varphi$}
\RightLabel{$\mathsf{r}_1$}
\UnaryInfC{$\Delta \Yright \varphi$}
\AxiomC{$d'_2$}
\UnaryInfC{$\Gamma'\lBrace \varphi \rBrace_m \Yright \psi'$}
\RightLabel{$\mathsf{mix}$}
\BinaryInfC{$\Gamma'\lBrace \Delta \rBrace_m \Yright \psi'$}
\RightLabel{$\mathsf{r}_2$}
\UnaryInfC{$\Gamma\lBrace \Delta \rBrace_n \Yright \psi$}
\end{bprooftree}
\end{equation}
Observe that, although $\mathsf{r_2}$ cannot introduce any new occurrence of $\varphi$, it may still remove one (this is the case where $\mathsf{r_2}=\sqcap\mathsf{C}$). Accordingly, in the right derivation pattern we require $m$ to be either $n$ or $n+1$. As $h(d_2')<h(d_2)$, the induction hypothesis (IH) applies, enabling us to conclude that $\vdash^{\mathrm{mf}}_{\mathsf{LMC}_{\ltimes}}\Gamma\lBrace \Delta \rBrace_n \Yright \psi$. 

If instead $\mathsf{r}_2$ introduces a new occurrence of $\varphi$ within $\Gamma$, then two cases must be considered. If $\mathsf{r}_2$ is $\sqcap\mathsf{W}_1$, then $\mathsf{mix}$ can be smoothly permuted with $\sqcap\mathsf{W}_1$ as illustrated in the following transformation schema:
\small
\begin{equation}\label{eq: two}
\begin{bprooftree}
\AxiomC{$d_1$}
\UnaryInfC{$\Delta' \Yright \varphi$}
\RightLabel{$\mathsf{r}_1$}
\UnaryInfC{$\Delta \Yright \varphi$}
\AxiomC{$d_2$}
\UnaryInfC{$\Gamma\{\Pi\} \Yright \psi$}
\RightLabel{$\sqcap\mathsf{W}_1$}
\UnaryInfC{$\Gamma\{\Pi\sqcap\underline{\varphi}\}\lBrace \varphi \rBrace_n \Yright \psi$}
\RightLabel{$\mathsf{mix}$}
\BinaryInfC{$\Gamma\{\Pi\sqcap\underline{\Delta}\}\lBrace \Delta \rBrace_n \Yright \psi$}
\end{bprooftree}
\dashrightarrow
\begin{bprooftree}
\AxiomC{$d_1$}
\UnaryInfC{$\Delta' \Yright \varphi$}
\RightLabel{$\mathsf{r}_1$}
\UnaryInfC{$\Delta \Yright \varphi$}
\AxiomC{$d'_2$}
\UnaryInfC{$\Gamma\{\Pi\}\lBrace \varphi \rBrace_{n-1}\Yright \psi$}
\RightLabel{$\mathsf{mix}$}
\BinaryInfC{$\Gamma\{\Pi\}\lBrace \Delta \rBrace_{n-1}\Yright \psi$}
\RightLabel{$\sqcap\mathsf{W}_1$}
\UnaryInfC{$\Gamma\{\Pi\sqcap\Delta\}\lBrace \Delta \rBrace_{n-1}\Yright \psi$}
\end{bprooftree}
\end{equation}\normalsize
Since $\Gamma\{\Pi\sqcap\Delta\}\lBrace \Delta \rBrace_{n-1}=\Gamma\{\Pi\sqcap\underline{\Delta}\}\lBrace \Delta \rBrace_{n}$ and $h(d'_2)<h(d_2)$, by the IH we have $\vdash^{\mathrm{mf}}_{\mathsf{LMC}_{\ltimes}}\Gamma\{\Pi\sqcap\underline{\Delta}\}\lBrace \Delta \rBrace_n \Yright \psi$.

If $\mathsf{r}_2$ is of $\cdot\,\mathsf{L}$, $\wedge\mathsf{L}$, $\pos\!\mathsf{L}$, or $\necv\!\mathsf{L}$, then we start from derivations of the form:
\begin{equation}\label{eq: three}
\begin{bprooftree}
\AxiomC{$d_1$}
\UnaryInfC{$\Delta' \Yright \varphi$}
\RightLabel{$\mathsf{r}_1$}
\UnaryInfC{$\Delta \Yright \varphi$}
\AxiomC{$d_2$}
\UnaryInfC{$\Gamma\{\Pi'\} \Yright \psi$}
\RightLabel{$\mathsf{r}_2$}
\UnaryInfC{$\Gamma\{\Pi\{\underline{\varphi}\}\}\lBrace \varphi \rBrace_n \Yright \psi$}
\RightLabel{$\mathsf{mix}$}
\BinaryInfC{$\Gamma\{\Pi\{\underline{\Delta}\}\}\lBrace \Delta \rBrace_n \Yright \psi$}
\end{bprooftree}
\end{equation}
First, we consider the penultimate step of $\mathsf{r}_2$ and replace the possible $n-1$ occurrences of $\varphi$ that appear in $\Gamma\{\Pi'\}$:
\begin{equation}\label{eq: four}
\begin{bprooftree}
\AxiomC{$d_1$}
\UnaryInfC{$\Delta' \Yright \varphi$}
\RightLabel{$\mathsf{r}_1$}
\UnaryInfC{$\Delta \Yright \varphi$}
\AxiomC{$d'_2$}
\UnaryInfC{$\Gamma\{\Pi'\}\lBrace \varphi \rBrace_{n-1} \Yright \psi$}
\RightLabel{$\mathsf{mix}$}
\BinaryInfC{$\Gamma\{\Pi'\}\lBrace \Delta \rBrace_{n-1} \Yright \psi$}
\end{bprooftree}
\end{equation}
Observe that, as $h(d_2')<h(d_2)$, $\vdash^{\mathrm{mf}}_{\mathsf{LMC}_{\ltimes}}\Gamma\{\Pi'\}\lBrace \Delta \rBrace_{n-1} \Yright \psi$ by the IH. We now move upwards along $d_1$ until the rule decomposing the $\mathsf{mix}$-formula $\varphi$ is reached, and, depending on the structure of $\Pi'$, we derive $\Gamma\{\Pi\{\underline{\Delta}\}\}\lBrace \Delta \rBrace_n \Yright \psi$ by using (one or more times) lower ranked applications of $\mathsf{mix}$, each one possibly followed by applications of other rules. To clarify this point, let us consider the case where $\mathsf{r}_2=\cdot\,\mathsf{L}$, $\Pi' := \varphi_1 \circ \varphi_2$, $\varphi:=\varphi_1\cdot\varphi_2$, and $d_1$ has the following structure:
\begin{equation}\label{eq: five}
\begin{bprooftree}
\AxiomC{$d_{11}$}
\UnaryInfC{$\Delta_1\{\xi_1\}\Yright\varphi_1$}
\AxiomC{$d_{12}$}
\UnaryInfC{$\Delta_2\Yright\varphi_2$}
\RightLabel{$\cdot\,\mathsf{R}$}
\BinaryInfC{$\Delta_1\{\xi_1\}\circ\Delta_2\Yright\varphi_1\cdot\varphi_2$}
\AxiomC{$d_{13}$}
\UnaryInfC{$\Delta_1\{\xi_2\}\Yright\varphi_1$}
\AxiomC{$d_{14}$}
\UnaryInfC{$\Delta_2\Yright\varphi_2$}
\RightLabel{$\cdot\,\mathsf{R}$}
\BinaryInfC{$\Delta_1\{\xi_2\}\circ\Delta_2\Yright\varphi_1\cdot\varphi_2$}
\RightLabel{$\vee\mathsf{L}$}
\BinaryInfC{$\Delta_1\{\xi_1\vee\xi_2\}\circ\Delta_2\Yright\varphi_1\cdot\varphi_2$}
\RightLabel{$\mathsf{r}_1$}
\UnaryInfC{$\Delta_1\{\xi_1\vee\xi_2\}\circ\Delta_2'\Yright\varphi_1\cdot\varphi_2$}
\end{bprooftree}
\end{equation}
Let us define $\Gamma':=\Gamma\{\varphi_1 \circ \varphi_2\}\lBrace \Delta_1\{\xi_1\vee\xi_2\}\circ\Delta_2' \rBrace_{n-1}$. We can now apply $\mathsf{mix}$ and replace $\varphi_1$ by $\Delta_1\{\xi_1\}$:
\begin{equation}\label{eq: six}
\begin{bprooftree}
\AxiomC{$d_{11}$}
\UnaryInfC{$\Delta_1\{\xi_1\} \Yright \varphi_1$}
\AxiomC{$\text{(\ref{eq: four})}$}
\UnaryInfC{$\Gamma'\{\underline{\varphi_1}\circ\varphi_2\}\lBrace \varphi \rBrace_{1} \Yright \psi$}
\RightLabel{$\mathsf{mix}$}
\BinaryInfC{$\Gamma'\{\underline{\Delta_1\{\xi_1\}}\circ\varphi_2\}\lBrace \Delta_1\{\xi_1\} \rBrace_{1} \Yright \psi$}
\end{bprooftree}
\end{equation}
Since $\mathrm{cp}(\varphi_1)<\mathrm{cp}(\varphi_1\cdot\varphi_2)$, we have that $\vdash^{\mathrm{mf}}_{\mathsf{LMC}_{\ltimes}}\Gamma'\{\underline{\Delta_1\{\xi_1\}}\circ\varphi_2\}\lBrace \Delta_1\{\xi_1\} \rBrace_{1} \Yright \psi$. Upon defining $\Gamma'':=\Gamma'\{\underline{\Delta_1\{\xi_1\}}\circ\varphi_2\}\lBrace \Delta_1\{\xi_1\} \rBrace_{1}$, we apply $\mathsf{mix}$ to the premisses $\Delta_2\Yright\varphi_2$ and $\Gamma'' \{\Delta_1\{\xi_1\}\circ\underline{\varphi_2}\}\lBrace \varphi_2 \rBrace_{1} \Yright \psi$ and derive $\Gamma'' \{\Delta_1\{\xi_1\}\circ\underline{\Delta_2}\}\lBrace \Delta_2 \rBrace_{1} \Yright \psi$, which admits a $\mathsf{mix}$-free proof by the IH. By the same procedure, using the derivations $d_{13}$ and $d_{14}$, it is possible to obtain $\vdash^{\mathrm{mf}}_{\mathsf{LMC}_{\ltimes}}\Gamma'' \{\Delta_1\{\xi_2\}\circ\underline{\Delta_2}\}\lBrace \Delta_2 \rBrace_{1} \Yright \psi$. Finally, by applying $\vee\mathsf{L}$ and then $\mathsf{r}_1$ to these last two sequents, it is not difficult to see that one obtains a proof of the conclusion of (\ref{eq: three})—for $\mathsf{r}_2=\cdot\,\mathsf{L}$, $\Pi' := \varphi_1 \circ \varphi_2$, $\varphi:=\varphi_1\cdot\varphi_2$, and $d_1$ as in (\ref{eq: five})—where only lower ranked instances of $\mathsf{mix}$ have been applied. The proof for $\mathsf{r}_1$/$\wedge\mathsf{L}$ is analogous. The cases $\mathsf{r}_1$/$\pos\!\mathsf{L}$, $\mathsf{r}_1$/$\necv\!\mathsf{L}$, and $\mathsf{r}_1$/$1\mathsf{L}$ are simpler and are left to the reader.

If $\mathsf{r}_2$ is either $\cdot\mathsf{R}$, $\wedge\mathsf{R}$, or $\vee\mathsf{L}$, then it suffices to permute $\mathsf{mix}$ with $\mathsf{r}_1$ as illustrated in the transformation schema below:
\begin{equation}\label{eq: seven}
\begin{bprooftree}
\AxiomC{$d_1$}
\UnaryInfC{$\Delta' \Yright \varphi$}
\RightLabel{$\mathsf{r}_1$}
\UnaryInfC{$\Delta \Yright \varphi$}
\AxiomC{$d_{21}$}
\UnaryInfC{$\Gamma' \Yright \psi'$}
\AxiomC{$d_{22}$}
\UnaryInfC{$\Gamma'' \Yright \psi''$}
\RightLabel{$\mathsf{r}_2$}
\BinaryInfC{$\Gamma\lBrace \varphi \rBrace_n \Yright \psi$}
\RightLabel{$\mathsf{mix}$}
\BinaryInfC{$\Gamma\lBrace \Delta \rBrace_n \Yright \psi$}
\end{bprooftree}\dashrightarrow
\begin{bprooftree}
\AxiomC{$d'_1$}
\UnaryInfC{$\Delta' \Yright \varphi$}
\AxiomC{$d_{21}$}
\UnaryInfC{$\Gamma' \Yright \psi'$}
\AxiomC{$d_{22}$}
\UnaryInfC{$\Gamma'' \Yright \psi''$}
\RightLabel{$\mathsf{r}_2$}
\BinaryInfC{$\Gamma\lBrace \varphi \rBrace_n \Yright \psi$}
\RightLabel{$\mathsf{mix}$}
\BinaryInfC{$\Gamma\lBrace \Delta' \rBrace_n \Yright \psi$}
\RightLabel{$(\mathsf{r}_1)\!\!\downarrow_n$}
\UnaryInfC{$\Gamma\lBrace \Delta \rBrace_n \Yright \psi$}
\end{bprooftree}
\end{equation}
Suppose now that $\Delta \Yright \varphi$ is obtained via an application of $\vee\mathsf{L}$. This is the only binary rule that does not decompose the $\mathsf{mix}$-formula in the right side of the endsequent of $d_1$. We start with a derivation
\begin{equation}\label{eq: eight}
\begin{bprooftree}
\AxiomC{$d_{11}$}
\UnaryInfC{$\Delta\{\xi_1\}\Yright\varphi$}
\AxiomC{$d_{12}$}
\UnaryInfC{$\Delta\{\xi_2\}\Yright\varphi$}
\RightLabel{$\vee\mathsf{L}$}
\BinaryInfC{$\Delta\{\xi_1\vee \xi_2\}\Yright\varphi$}
\AxiomC{$d_2$}
\UnaryInfC{$\Gamma\{\Pi'\}\Yright\psi$}
\RightLabel{$\mathsf{r}_2$}
\UnaryInfC{$\Gamma\{\Pi\{\underline{\varphi}\}\}\lBrace\varphi\rBrace_n\Yright\psi$}
\RightLabel{$\mathsf{mix}$}
\BinaryInfC{$\Gamma\{\Pi\{\underline{\Delta\{\xi_1\vee \xi_2\}}\}\}\lBrace\Delta\{\xi_1\vee \xi_2\}\rBrace_n\Yright\psi$}
\end{bprooftree}
\end{equation}
We first proceed as in (\ref{eq: four}) and replace the $n-1$ selected occurrences of $\varphi$ appearing in $\Gamma\{\Pi'\}$. Then we apply $\mathsf{r}_2$ and introduce $\varphi$. By the IH, we conclude that $\vdash^{\mathrm{mf}}_{\mathsf{LMC}_{\ltimes}}\Gamma\{\Pi\{\varphi\}\}\lBrace\Delta\{\xi_1 \vee \xi_2\}\rBrace_{n-1}\Yright\psi$. We now trace back along $d_{11}$ and $d_{12}$ until we reach the two instances of the rule $\mathsf{r}_{\varphi}$ decomposing $\varphi$. Consider, for instance, the derivation schema (\ref{eq: nine}). For convenience, and without loss of generality, we assume that $\mathsf{r}_{\varphi}$ is unary, and that no applications of $\vee\mathsf{L}$ occur in the intermediate derivation steps between $\Delta_{11}\Yright\varphi$ and $\Delta\{\xi_1\}\Yright\varphi$, or between $\Delta_{12}\Yright\varphi$ and $\Delta\{\xi_2\}\Yright\varphi$.
\begin{equation}\label{eq: nine}
\begin{bprooftree}
\AxiomC{$d_{11}$}
\UnaryInfC{$\Delta_{11}'\Yright\varphi_{11}'$}
\RightLabel{$\mathsf{r}_{\varphi}$}
\UnaryInfC{$\Delta_{11}\Yright\varphi$}
\UnaryInfC{\shortstack{some\\derivation\\steps}}
\UnaryInfC{$\Delta\{\xi_1\}\Yright\varphi$}
\AxiomC{$d_{12}$}
\UnaryInfC{$\Delta_{12}'\Yright\varphi_{12}'$}
\RightLabel{$\mathsf{r}_{\varphi}$}
\UnaryInfC{$\Delta_{12}\Yright\varphi$}
\UnaryInfC{\shortstack{some\\derivation\\steps}}
\UnaryInfC{$\Delta\{\xi_2\}\Yright\varphi$}
\RightLabel{$\vee\mathsf{L}$}
\BinaryInfC{$\Delta\{\xi_1\vee \xi_2\}\Yright\varphi$}
\AxiomC{$d'_2$}
\UnaryInfC{$\Gamma\{\Pi'\}\lBrace\varphi\rBrace_{n-1}\Yright\psi$}
\RightLabel{$\mathsf{mix}$}
\BinaryInfC{$\Gamma\{\Pi'\}\lBrace\Delta\{\xi_1 \vee \xi_2\}\rBrace_{n-1}\Yright\psi$}
\RightLabel{$\mathsf{r}_2$}
\UnaryInfC{$\Gamma\{\Pi\{\varphi\}\}\lBrace\Delta\{\xi_1 \vee \xi_2\}\rBrace_{n-1}\Yright\psi$}
\end{bprooftree}
\end{equation}
Let us call $d_{11}'$ and $d_{12}'$ the subderivations of $d_{11}$ and $d_{12}$ that end with the sequents $\Delta_{11}\Yright\varphi$ and $\Delta_{12}\Yright\varphi$, respectively. Setting $\Gamma':=\Gamma\{\Pi\{\varphi\}\}\lBrace\Delta\{\xi_1 \vee \xi_2\}\rBrace_{n-1}$, we construct the following derivation: 
\begin{equation*}
\begin{bprooftree}
\AxiomC{$d'_{11}$}
\UnaryInfC{$\Delta_{11}'\Yright\varphi_{11}'$}
\RightLabel{$\mathsf{r}_{11}$}
\UnaryInfC{$\Delta_{11}\Yright\varphi$} 
\AxiomC{$\text{(\ref{eq: nine})}$}
\UnaryInfC{$\Gamma'\{\Pi\{\underline{\varphi}\}\}\lBrace \varphi\rBrace_{1}\Yright\psi$}
\RightLabel{$\mathsf{mix}$}
\BinaryInfC{$\Gamma'\{\Pi\{\underline{\Delta_{11}}\}\}\lBrace \Delta_{11}\rBrace_{1}\Yright\psi$}
\UnaryInfC{\shortstack{some\\derivation\\steps}}
\UnaryInfC{$\Gamma'\{\Pi\{\Delta\{\xi_1\}\}\}\Yright\psi$}
\AxiomC{$d'_{12}$}
\UnaryInfC{$\Delta_{12}'\Yright\varphi_{12}'$}
\RightLabel{$\mathsf{r}_{12}$}
\UnaryInfC{$\Delta_{12}\Yright\varphi$} 
\AxiomC{$\text{(\ref{eq: nine})}$}
\UnaryInfC{$\Gamma'\{\Pi\{\underline{\varphi}\}\}\lBrace \varphi\rBrace_{1}\Yright\psi$}
\RightLabel{$\mathsf{mix}$}
\BinaryInfC{$\Gamma'\{\Pi\{\underline{\Delta_{12}}\}\}\lBrace \Delta_{12}\rBrace_{1}\Yright\psi$}
\UnaryInfC{\shortstack{some\\derivation\\steps}}
\UnaryInfC{$\Gamma'\{\Pi\{\Delta\{\xi_2\}\}\}\Yright\psi$}
\RightLabel{$\vee\mathsf{L}$}
\BinaryInfC{$\Gamma'\{\Pi\{\Delta\{\xi_1\vee\xi_2\}\}\}\Yright\psi$}
\end{bprooftree}
\end{equation*}
Note that $p(d_{11}')=p(d_{12}')=0$, whereas in (\ref{eq: eight}) we have $p(d_1)=1$; hence, the IH applies yielding $\vdash^{\mathrm{mf}}_{\mathsf{LMC}_{\ltimes}}\Gamma'\{\Pi\{\underline{\Delta_{11}}\}\}\lBrace \Delta_{11}\rBrace_{1}\Yright\psi$ and $\vdash^{\mathrm{mf}}_{\mathsf{LMC}_{\ltimes}}\!\!\Gamma'\{\Pi\{\underline{\Delta_{12}}\}\}\lBrace \Delta_{12}\rBrace_{1}\Yright\psi$. Therefore $\vdash^{\mathrm{mf}}_{\mathsf{LMC}_{\ltimes}}\Gamma'\{\Pi\{\Delta\{\xi_1\vee\xi_2\}\}\}\Yright\psi$. It is immediate to check that the latter sequent is precisely the conclusion of (\ref{eq: eight}). We leave to the reader the cases in which $\mathsf{r}_1 = \vee\mathsf{L}$ and $\mathsf{r}_2$ is a binary rule: the proofs for the $\vee\mathsf{L}$/$\wedge\mathsf{R}$ and $\vee\mathsf{L}$/$\cdot\,\mathsf{R}$ combinations proceed similarly to (\ref{eq: seven}), the only difference being that the IH must be applied twice. As for the $\vee\mathsf{L}$/$\vee\mathsf{L}$ case, one must first replace the $n-1$ occurrences of the $\mathsf{mix}$-formula on the left sides of the two premisses of $\vee\mathsf{L}$ in the last step of $d_2$; the proof then proceeds as above.

If a new occurrence of the $\mathsf{mix}$-formula is introduced on the right side of the endsequent of $d_1$, and nothing happens on the left side of the endsequent of $d_2$, the proof is straightforward as it suffices to permute $\mathsf{mix}$ with (possibly iterated applications) of $\mathsf{r}_1$.

As for (ii), we restrict our attention to cases involving logical rules. Instances where $d_2$ ends with an application of $\sqcap\mathsf{W}_1$ are treated as in (\ref{eq: two}).

Let us consider the following instance of $\mathsf{mix}$:
\begin{equation}\label{eq: cdot1}
\begin{bprooftree}
\AxiomC{$d_{11}$}
\UnaryInfC{$\Delta_1 \Yright \varphi$}
\AxiomC{$d_{12}$}
\UnaryInfC{$\Delta_2 \Yright \psi$}
\RightLabel{$\cdot\,\mathsf{R}$}
\BinaryInfC{$\Delta_1 \circ \Delta_2 \Yright \varphi \cdot\psi$}
\AxiomC{$d_2$}
\UnaryInfC{$\Gamma\{\varphi \circ \psi\}\Yright \chi$}
\RightLabel{$\cdot \, \mathsf{L}$}
\UnaryInfC{$\Gamma\{\underline{\varphi \cdot \psi}\}\lBrace\varphi\cdot\psi\rBrace_n\Yright \chi$}
\RightLabel{$\mathsf{mix}$}
\BinaryInfC{$\Gamma\{\underline{\Delta_1 \circ \Delta_2}\}\lBrace\Delta_1 \circ \Delta_2\rBrace_n\Yright \chi$}
\end{bprooftree}
\end{equation}
In order to apply the IH, just like in some of the previous cases, we first replace possible occurrences of $\varphi\cdot\psi$ that appear in $\Gamma\{\varphi \circ \psi\}$:
\begin{equation}\label{eq: cdot2}
\begin{bprooftree}
\AxiomC{$d_{11}$}
\UnaryInfC{$\Delta_1 \Yright \varphi$}
\AxiomC{$d_{12}$}
\UnaryInfC{$\Delta_2 \Yright \psi$}
\RightLabel{$\cdot\,\mathsf{R}$}
\BinaryInfC{$\Delta_1 \circ \Delta_2 \Yright \varphi \cdot\psi$}
\AxiomC{$d'_2$}
\UnaryInfC{$\Gamma\{\varphi \circ \psi\}\lBrace\varphi \cdot \psi\rBrace_{n-1}\Yright \chi$}
\RightLabel{$\mathsf{mix}$}
\BinaryInfC{$\Gamma\{\varphi \circ \psi\}\lBrace\Delta_1 \circ \Delta_2\rBrace_{n-1}\Yright \chi$}
\end{bprooftree}
\end{equation}
Since $h(d_2')<h(d_2)$, by the IH $\vdash^{\mathrm{mf}}_{\mathsf{LMC}_{\ltimes}} \Gamma\{\varphi \circ \psi\}\lBrace\Delta_1 \circ \Delta_2\rBrace_{n-1}\Yright \chi$. We now proceed with the replacement of $\varphi$ by $\Delta_1$ and $\psi$ by $\Delta_2$ in $\varphi \circ \psi$. Let us define $\Gamma':=\Gamma\{\varphi \circ \psi\}\lBrace\Delta_1 \circ \Delta_2\rBrace_{n-1}$. We have:
\begin{equation}\label{eq: cdot3}
\begin{bprooftree}
\AxiomC{$d'_{11}$}
\UnaryInfC{$\Delta_1 \Yright \varphi$}
\AxiomC{$\text{(\ref{eq: cdot2})}$}
\UnaryInfC{$\Gamma'\{\underline{\varphi} \circ \psi\}\lBrace\varphi\rBrace_1\Yright \chi$}
\RightLabel{$\mathsf{mix}$}
\BinaryInfC{$\Gamma'\{\underline{\Delta_1} \circ \psi\}\lBrace\Delta_1\rBrace_1\Yright \chi$}
\end{bprooftree}
\end{equation}
where $\vdash^{\mathrm{mf}}_{\mathsf{LMC}_{\ltimes}}\Gamma'\{\underline{\Delta_1} \circ \psi\}\lBrace\Delta_1\rBrace_1\Yright \chi$, since $\mathrm{cp}(\varphi)<\mathrm{cp}(\varphi\cdot \psi)$ and (\ref{eq: cdot2}) applies a lower ranked instance of $\mathsf{mix}$. Let us now define $\Gamma'':= \Gamma'\{\underline{\Delta_1} \circ \psi\}\lBrace\Delta_1\rBrace_m$. We replace $\psi$ by $\Delta_2$:
\begin{equation}\label{eq: cdot4}
\begin{bprooftree}
\AxiomC{$d'_{12}$}
\UnaryInfC{$\Delta_2 \Yright \psi$}
\AxiomC{$\text{(\ref{eq: cdot3})}$}
\UnaryInfC{$\Gamma''\{\Delta_1 \circ \underline{\psi}\}\lBrace\psi\rBrace_1\Yright \chi$}
\RightLabel{$\mathsf{mix}$}
\BinaryInfC{$\Gamma''\{\Delta_1 \circ \underline{\Delta_2}\}\lBrace\Delta_2\rBrace_1\Yright \chi$}
\end{bprooftree}
\end{equation}
As before, since $\mathrm{cp}(\psi)<\mathrm{cp}(\varphi\cdot \psi)$ and (\ref{eq: cdot3}) applies a lower ranked instance of $\mathsf{mix}$, we have $\vdash^{\mathrm{mf}}_{\mathsf{LMC}_{\ltimes}}\Gamma''\{\Delta_1 \circ \underline{\Delta_2}\}\lBrace\Delta_2\rBrace_1\Yright \chi$. It is immediate that $\Gamma''\{\Delta_1 \circ \underline{\Delta_2}\}\lBrace\Delta_2\rBrace_1$ coincides with $\Gamma\{\underline{\Delta_1 \circ \Delta_2}\}\lBrace\Delta_1 \circ \Delta_2\rBrace_n$ in (\ref{eq: cdot1}), whence $\vdash^{\mathrm{mf}}_{\mathsf{LMC}_{\ltimes}}\Gamma\{\underline{\Delta_1 \circ \Delta_2}\}\lBrace\Delta_1 \circ \Delta_2\rBrace_n \Yright \chi$. The treatment of the case $\wedge\mathsf{R}$/$\wedge\mathsf{L}$ is similar. Let us now consider the derivation:
\begin{equation}\label{eq: vee1}
\begin{bprooftree}
\AxiomC{$d_1$}
\UnaryInfC{$\Delta\Yright\varphi$}
\RightLabel{$\vee\mathsf{R}_1$}
\UnaryInfC{$\Delta \Yright \varphi\vee \psi$}
\AxiomC{$d_{21}$}
\UnaryInfC{$\Gamma\{\varphi\}\Yright\chi$}
\AxiomC{$d_{22}$}
\UnaryInfC{$\Gamma\{\psi\}\Yright\chi$}
\RightLabel{$\vee\mathsf{L}$}
\BinaryInfC{$\Gamma\{\underline{\varphi\vee \psi}\}\lBrace\varphi\vee \psi\rBrace_n \Yright \chi$}
\RightLabel{$\mathsf{mix}$}
\BinaryInfC{$\Gamma\{\underline{\Delta}\}\lBrace\Delta\rBrace_n \Yright \chi$}
\end{bprooftree}
\end{equation}
Here too, we begin by replacing selected occurrences of the $\mathsf{mix}$-formula that appear in $\Gamma\{\varphi\}$ (or, equivalently, in $\Gamma\{\psi\}$):
\begin{equation}\label{eq: vee2}
\begin{bprooftree}
\AxiomC{$d_1$}
\UnaryInfC{$\Delta\Yright\varphi$}
\RightLabel{$\vee\mathsf{R}_1$}
\UnaryInfC{$\Delta \Yright \varphi\vee \psi$}
\AxiomC{$d_{21}'$}
\UnaryInfC{$\Gamma\{\varphi\}\lBrace\varphi\vee \psi\rBrace_{n-1}\Yright\chi$}
\RightLabel{$\mathsf{mix}$}
\BinaryInfC{$\Gamma\{\varphi\}\lBrace\Delta\rBrace_{n-1} \Yright \chi$}
\end{bprooftree}
\end{equation}
By the IH, $\vdash^{\mathrm{mf}}_{\mathsf{LMC}_{\ltimes}}\Gamma\{\varphi\}\lBrace\Delta\rBrace_{n-1} \Yright \chi$. Upon defining $\Gamma':=\Gamma\{\varphi\}\lBrace\Delta\rBrace_{n-1}$, we construct the derivation:
\[
\begin{bprooftree}
\AxiomC{$d'_1$}
\UnaryInfC{$\Delta\Yright\varphi$}
\AxiomC{$\text{(\ref{eq: vee2})}$}
\UnaryInfC{$\Gamma'\{\underline{\varphi}\}\lBrace\varphi\rBrace_{1} \Yright \chi$}
\RightLabel{$\mathsf{mix}$}
\BinaryInfC{$\Gamma'\{\underline{\Delta}\}\lBrace\Delta\rBrace_{1} \Yright \chi$}
\end{bprooftree}
\]
So we have $\vdash^{\mathrm{mf}}_{\mathsf{LMC}_{\ltimes}}\Gamma'\{\underline{\Delta}\}\lBrace\Delta\rBrace_{1} \Yright \chi$, and since $\Gamma'\{\underline{\Delta}\}\lBrace\Delta\rBrace_{1}$ coincides with $\Gamma\{\underline{\Delta}\}\lBrace\Delta\rBrace_n$ in (\ref{eq: vee1}), we obtain 
$\vdash^{\mathrm{mf}}_{\mathsf{LMC}_{\ltimes}}\Gamma\{\underline{\Delta}\}\lBrace\Delta\rBrace_n \Yright \chi$. The case $\vee\mathsf{R}_2$/$\vee\mathsf{L}$ follows exactly the same pattern. The proof for the $\pos\!\mathsf{R}$/$\pos\!\mathsf{L}$ and $\necv\!\mathsf{R}$/$\necv\!\mathsf{L}$ combinations follow by analogous, yet simpler, reasoning, and are thus left to the reader.\normalsize
\end{proof}
\begin{corollary}
The rule $\mathsf{cut}$ can be eliminated from $\mathsf{LMC}$.
\end{corollary}
\begin{proof}
    Immediate from Remark \ref{rem: cutmixeq} and Theorem \ref{theo: mixelim}.
\end{proof}
\section{Conclusion}\label{sect: conc}
$\mathsf{LMC}$ is not the first substructural logic to feature a residuated pair of modal operators. The fundamental role played by adjunction in modal logic is well established, as evidenced by Dunn’s Gaggle Theory~\cite{Dunn91,Dunn1993,Dunn1995} and related display calculi (see, e.g.,~\cite{Gore1998}), as well as by the extensive body of work on categorial type logics developed over the past three decades~\cite{ArecesBernardi04,ArBerMo03,KM97,Moortgat96,Moortgat97,MR2012}. It is to the latter research programme that we owe significant advances in the proof-theoretic treatment of residuated or Galois-connected unary modalities, beginning with Moortgat’s (non-associative) modal Lambek calculus $\mathsf{NL}(\pos)$~\cite{Moortgat96}, from which we took both the logical and the structural modal rules for $\mathsf{LMC}$.

However, as noted in~\cite{MenniSmith2014}, despite such a large number of contributions, the development of Gentzen-style systems for \emph{positive}, \emph{implication-free} structures equipped with adjoint pairs of unary operators is a relatively recent research direction. Inspired by Kashima’s nested sequents for tense logic~\cite{Kashima1994}, Sadrzadeh and Dyckhoff~\cite{SADY2009} initiated this line of inquiry by introducing a tree-style sequent calculus for bounded distributive lattices endowed with \emph{families} of residuated pairs of the form $\langle \posv,\nec\rangle$ (backward dimond/forward box). In the present paper, although we focus on structures equipped with a single residuated pair, we work within a richer algebraic framework that includes a (non-commutative) monoid multiplication and imposes stronger conditions on the modal operators. In developing both our algebraic model for f-properties and~$\mathsf{LMC}$, we have deliberately chosen to begin with a small set of fundamental operations. From the viewpoint of rule apparatus, however, the calculus is not strictly minimal due to the inclusion of~$\mathsf{K}$, $\mathsf{T}$, and~$\mathsf{4}$. Finally, while our system satisfies the \emph{subformula property}, it does not satisfy the \emph{substructure property} because of rule~$\mathsf{T}$.

As for future work, several directions remain to be explored. 

First, we intend to establish decidability and complexity results for~$\mathsf{LMC}$, where, as for decidability, the crucial point is clearly to handle the issues arising from the structural non-analyticity of~$\mathsf{T}$. Next, in order to develop more effective tools for reasoning about the structure of f-properties, we plan to extend our research in two parallel directions which, we hope, will eventually converge.
\begin{enumerate}
    \item We aim to provide a complete axiomatisation of~$\wp(\Sigma^{*})$ as a bounded~$\ell$-monoid, together with a structural proof-theoretic counterpart.
    \item We intend to study expansions of closure~$\ell$-monoids obtained by adding operations such as the divisions of residuated lattices or the Kleene star. Besides offering a coherent further development of our algebraic theory of f-properties, we believe that the proof-theoretic counterparts of these enriched structures are intrinsically interesting in their own right. In particular, it would be worthwhile, both for~$\mathsf{LMC}$ and for such extensions thereof, to address the problem of obtaining cut-free completeness results, as achieved in~\cite{Kozak09} for the Distributive Full Lambek Calculus, and in~\cite{FSM2025} for Infinitary Action Logic.
\end{enumerate}
In pursuing these objectives, if necessary, we will explore alternative proof-theoretic frameworks such as nested sequents~\cite{SADY2009,LP2024} or deep inference~\cite{Br2009} Moreover, we believe it would be fruitful to relate~$\mathsf{LMC}$ to other proof-theoretic approaches to the analysis of concurrent systems, in particular those developed in~\cite{AM2025,AMM2025,Guglielmi94,Guglielmi95}.

From a strictly algebraic perspective, we first intend to develop further the study of the variety of closure~$\ell$-monoids by establishing representation and structure theorems. Moreover, as the reader will have noticed, the term “closure~$\ell$-monoid” is clearly inspired by McKinsey and Tarski’s \emph{closure algebras}~\cite{McKT44,McKT46}, which provide the algebraic semantics for~$\mathsf{S4}$. In addition, our terminological choice reflects the fact that the prefix-closure operation plays a central and distinctive role in our theory, and “closure~$\ell$-monoid” is meant to emphasise this. However, one might argue that, in view of the residuation between $\pos$ and $\necv$, a more appropriate name for our structures would be “monadic~$\ell$-monoids”, by analogy with Halmos’ \emph{monadic Boolean algebras}~\cite{Halmos54-56}, the algebraic semantics of~$\mathsf{S5}$. Indeed, letting~$\nu$ be the similarity type of closure~$\ell$-monoids, these structures in fact seem to be m-$\nu$-lattices in the sense of~\cite{CMT2024}; thus it would be interesting to investigate how $\mathfrak{LMC}$ relates to the general theory of one-variable fragments of first-order logics.

\subsubsection*{Acknowledgements} 
This paper benefited greatly from the insightful and thorough comments of George Metcalfe and Simon Santschi, to whom we express our sincere gratitude. We thank Augusto Antonio Basilico for providing valuable remarks that helped improve this work. We are also indebted to Nick Galatos, Peter Jipsen, Francesco Paoli, and Mario Piazza for their constructive observations. Preliminary versions of this paper were presented at the University of Cagliari and at the Scuola Normale Superiore of Pisa, and we wish to thank all those who attended these talks for their helpful feedback.

\bibliographystyle{splncs04}
\bibliography{bibliography}
\end{document}